\RequirePackage[l2tabu, orthodox]{nag}
\documentclass[11pt, fleqn]{article}

\usepackage{microtype}
\usepackage{mathtools, amssymb, amsthm, latexsym}
\usepackage{mathrsfs}
\usepackage{bbm}
\usepackage{natbib}
\usepackage[shortlabels]{enumitem}
\usepackage{graphicx}
\usepackage{booktabs}
\usepackage{threeparttable}
\usepackage{longtable}
\usepackage[table]{xcolor}
\usepackage[letterpaper, top=2.9cm, bottom=2.4cm, hmargin={3cm, 3cm}]{geometry}
\usepackage{verbatim}
\usepackage{hyperref}
   \hypersetup{colorlinks=true, linkcolor=blue, citecolor=blue!50!black, urlcolor=cyan}
\usepackage[capitalize, noabbrev]{cleveref}
\usepackage[titletoc, title]{appendix}
\usepackage{authblk}
\usepackage{setspace}

\usepackage{afterpage}
\usepackage{lscape}
\usepackage{caption}
\usepackage{floatpag}

\theoremstyle{plain}
\newtheorem{Lem}{Lemma}
\newtheorem{Pro}{Proposition}
\newtheorem{Thm}{Theorem}
\newtheorem{Cor}{Corollary}
\newtheorem{Rmk}{Remark}

\newcommand{\Myd}{\;\mathrm{d}}
\newcommand{\Ind}[1]{\mathbbm{1}_{#1}}
\newcommand{\Bigo}[1]{\operatorname{O}\left(#1\right)}
\newcommand{\Bigomega}[1]{\operatorname{\Omega}\left(#1\right)}
\newcommand{\Smallo}[1]{\operatorname{o}\left(#1\right)}
\DeclareMathOperator*{\Exp}{\mathbb{E}}

\title{Model Selection in Utility-Maximizing Binary Prediction}
\author{Jiun-Hua Su\thanks{
This paper is a revision of the third chapter of my dissertation.
I am grateful to the co-editor, the associate editor, and two anonymous referees for constructive comments and suggestions.
I also thank Peter Bartlett, Le-Yu Chen, Yu-Chin Hsu, Hsuan-Tien Lin, Demian Pouzo, James Powell, seminar participants at Erasmus University Rotterdam, and participants at the 2019 Australasian Meetings of the Econometric Society for helpful discussions.
Address correspondence to Jiun-Hua Su, 128 Academia Road, Section 2, Nankang, Taipei, 115 Taiwan; E-mail address: jhsu@econ.sinica.edu.tw.}}
\affil{Institute of Economics\\Academia Sinica}
\begin{document}
\maketitle
\thispagestyle{empty}

\begin{abstract}
The maximum utility estimation proposed by \citet{ElliottLieli2013} can be viewed as cost-sensitive binary classification; thus, its in-sample overfitting issue is similar to that of perceptron learning. A utility-maximizing prediction rule (UMPR) is constructed to alleviate the in-sample overfitting of the maximum utility estimation. We establish non-asymptotic upper bounds on the difference between the maximal expected utility and the generalized expected utility of the UMPR. Simulation results show that the UMPR with an appropriate data-dependent penalty achieves larger generalized expected utility than common estimators in the binary classification if the conditional probability of the binary outcome is misspecified.



\end{abstract}

\fontsize{12}{18pt}\selectfont
\newpage
\setcounter{page}{1}

\section{Introduction}\label{Introduction}
Making a binary decision based on an uncertain binary outcome is common in modern economic activities.
For instance, an investor who considers buying a financial instrument may tend to predict its price change in the future and decide to buy the instrument if the price is predicted to rise.
As suggested by \citet{GrangerMachina2006}, decision-making based on the prediction of a binary outcome should be driven by the preference of the decision maker.
On the one hand, the utility arising from a mismatch between the binary decision and outcome may differ in the realized outcome; on the other hand, the utility may be affected by observable covariates.
In making financial investment decisions, the disutility for the investor who buys the instrument but suffers from a decrease in the price may be greater than that for the investor who does not buy the instrument but finds an increase in the price.
In addition, features of the instrument, for example measures of its price volatility, may affect not only the likelihood of price change but also the investor's utility.\footnote{\citet{BarberisXiong2012} propose a model to explain the individual investor preference for volatile stocks.
}
Further examples illustrating the importance of the decision maker's preference in economic forecasting are provided in \citet{ElliottTimmermann2016}.

Although the subjective preference would be important for the decision-making based on binary prediction, traditional methods of pattern classification rarely take the decision maker's utility into consideration.
Recently, \citet{ElliottLieli2013} propose a maximum utility approach to incorporate the decision maker's utility into prediction of a binary outcome $Y\in\{-1,1\}$ given a vector $X$ of observed covariates.
Instead of globally estimating the conditional probability $p^{*}(x)\equiv\mathbb{P}(Y=1\mid X=x)$, they show the utility-maximizing binary classification problem can be solved by only estimating the sign of $p^{*}(x)-c(x)$, where $c$ is the cutoff function determined by the decision maker's utility function.
Compared with maximum likelihood estimation, their maximum utility estimation is, however, prone to in-sample overfitting.

In this paper, we show that the maximum utility estimation can be viewed as cost-sensitive binary classification; thus, its in-sample overfitting issue is similar to that of perceptron learning in the machine learning literature.
To alleviate the tendency of fitting the in-sample noise by sophisticated models, we follow the structural risk minimization approach proposed by \citet{Vapnik1982}.
More precisely, we pre-specify a hierarchy of classes of (finite-dimensional) functions and consider a \emph{utility-maximizing prediction rule} (UMPR), which is a maximum utility estimator that maximizes a complexity-penalized empirical utility.
To construct the UMPR, we consider a VC-type distribution-free complexity penalty and four data-dependent complexity penalties:
maximal discrepancy (\citet{BartlettBoucheronEtAl2002}), simulated maximal discrepancy, Rademacher complexity (\citet{Koltchinskii2001} and \citet{BartlettBoucheronEtAl2002} among others), and bootstrap complexity (\citet{Fromont2007}).
We evaluate the performance of a prediction rule by the \emph{generalized expected utility}, that is, the expected utility achieved by the decision maker applying this prediction rule, which could be constructed based on in-sample observations, to the classification of an out-of-sample observation.
We emphasize generalization because the expected utility averages over not only the out-of-sample observation but the \emph{ex ante} in-sample observations as well.
We prove that the difference between the maximal expected utility and the generalized expected utility of the UMPR can be bounded by an almost optimal trade-off between the expected complexity penalty and the approximation error, that is, an error due to the approximation of functions in a hierarchy of classes to an optimal decision rule.
Hence, whenever the approximation error is equal to zero for some class of functions, the expected utility of the UMPR increases in the sample size and will asymptotically attain the maximal expected utility.
In other words, the proposed UMPR is \emph{universally utility consistent}.\footnote{
Universal utility consistency has a counterpart in the literature on empirical risk minimization, in which different names are used, for example universal consistency in \citet{DevroyeGyoerfiEtAl1996}, persistence in \citet{GreenshteinRitov2004} and \citet{Greenshtein2006}, and risk consistency in \citet{HomrighausenMcDonald2013}.
Another conceptually different but common term is consistency, and in this paper, it refers to the property that a selection method asymptotically picks a model with the lowest Kullback-Leibler divergence, as in \citet{SinWhite1996}.
}

The idea of complexity penalization has been applied to selection methods in econometrics.
Instead of penalized empirical utility criteria, penalized likelihood criteria are the main concerns in early literature.
One strand of literature adopts the information-theoretic approach.
Classical examples include \citeauthor{Akaike1973}'s (\citeyear{Akaike1973}) information criterion (AIC), \citeauthor{Schwarz1978}'s (\citeyear{Schwarz1978}) information criterion (BIC), and their cousins (TIC and GIC among others).\footnote{
Schwarz's information criterion is also called BIC, as it is usually analyzed from a Bayesian angle.
TIC and GIC are abbreviations for Takeuchi's information criterion and generalized information criterion, respectively.}
Additionally, motivated by different spirits, the leave-one-out cross-validation in the likelihood framework is asymptotically equivalent to the AIC, whereas the minimum description length can be approximated by the BIC.
These early selection methods are well documented in \citet{KonishiKitagawa2008} and \citet{ClaeskensHjort2008}.
With the purpose of achieving shrinkage and variable selection simultaneously, \citeauthor{Tibshirani1996}'s (\citeyear{Tibshirani1996}) least absolute shrinkage and selection operator (LASSO) adopts an $\ell_{1}$ penalty.
All of these selection methods can be applied to the binary prediction in maximum likelihood estimation with the logit specification.
As an alternative to penalized maximum likelihood estimation, penalized maximum score estimation is recently applied to the variable selection in binary prediction by \citet{ChenLee2018b}, in which an $\ell_{0}$ penalty is used.\footnote{
Similarly, by setting a bound on the $\ell_{0}$-norm of covariates, \citet{ChenLee2018} consider the constrained maximum score estimation to select covariates in binary prediction.
}
Furthermore, replacing the zero-one loss with the hinge loss, the $\ell_{1}$-norm support vector machine (SVM), developed by \citet{ZhuRossetEtAl2004} and \citet{FungMangasarian2004} among others in the machine learning community, can effectively select variables in the traditional binary classification, namely binary classification with a symmetric loss independent of covariates.
However, none of these penalty-based selection methods above takes the decision maker's utility into account.

Despite the prevalence of penalty-based selection methods, pretesting and cross-validation are two alternatives in literature.
Both alternatives can be adapted for the selection in maximum utility estimation.
\citet{ElliottLieli2013} propose a general-to-specific pretest to select MU estimators but do not investigate theoretical properties of their post-model-selection MU estimator.
However, as suggested by \citeauthor{LeebPoetscher2005} (\citeyear{LeebPoetscher2005}, \citeyear{LeebPoetscher2008a}), a post-model-selection estimator would have complicated distributional properties.
The complicated distributional properties are partly attributable to using the same data for both estimation and validation.
In contrast, data splitting makes it convenient to evaluate the out-of-sample performance of a prediction rule so that cross-validation can be used to select models in almost any framework.
As argued by \citet{ArlotCelisse2010}, the wide applicability of cross-validation however makes its predictive performance less satisfactory than that of selection methods tailored in a specific framework.
The lack of theoretical analysis of the pretest and cross-validatory estimators motivates the evaluation by Monte Carlo experiments in this paper.
According to the simulation results, the UMPR with an appropriate data-dependent penalty outperforms the pretest and cross-validatory estimators.
The UMPR with an appropriate data-dependent penalty also outweighs the AIC, BIC, and LASSO if the conditional probability of the binary outcome is misspecified, and the $\ell_{1}$-norm SVM if the cutoff function considerably deviates from $1/2$.

The proposed method in this paper can be viewed as an essential complement to the literature concerned with model selection in binary classification.
Although cost-sensitive binary classification is important for decision-making, previous studies on binary classification focus on performance evaluated by the cost-insensitive rate of misclassification.
By applying \citeauthor{McDiarmid1989}'s (\citeyear{McDiarmid1989}) inequality and \citeauthor{Massart2000}'s (\citeyear{Massart2000}) lemma, we extend the aforementioned complexity penalties to the maximum utility estimation, in which the cost of misclassification may depend on the binary outcome and covariates.
Following the structural risk minimization approach, we establish non-asymptotic upper bounds on the difference between the maximal expected utility and the generalized expected utility of the UMPR.
These bounds, similar to those in \citet{Koltchinskii2001} and \citet{BartlettBoucheronEtAl2002} for the traditional binary classification, strike a balance between the approximation error and the expected complexity penalty.
In addition, we establish a non-asymptotic upper bound, which shrinks to zero as the in-sample size tends to infinity, on the expected value of these data-dependent complexity penalties.
We also extend properties of the Bayes decision rule to the maximum utility estimation.
This extension not only confirms \citeauthor{ElliottLieli2013}'s (\citeyear{ElliottLieli2013}) insight that the knowledge of sign of $p^{*}(x)-c(x)$ suffices to achieve the maximal expected utility, but also implies that the approximation error is bounded by the uniform distance between the conditional probability $p^{*}$ and the underlying class of functions.
Consequently, given pre-specified classes of functions, we can examine the universal utility consistency of UMPR, which is a property (corresponding to the risk consistency in loss minimization) rarely investigated in the previous studies on model selection in binary classification.

Throughout this paper, all random variables are defined on the probability space $(\Omega, \mathcal{A}, \mathbb{P})$.
Data are assumed to be i.i.d.\ (independent and identically distributed).
We write $\mathbb{N}$ for the collection of positive integers and $\mathbb{R}$ for the collection of real numbers.
We denote the indicator function by $\Ind{[E]}$, which equals one if event $E$ occurs and equals zero otherwise.
We also denote the sign function by $\text{sign}(z)$, which is equal to $2\Ind{[z\geq 0]}-1$ for any $z\in\mathbb{R}$.

The structure of the remaining paper is as follows.
Section~\ref{Model} describes the maximum utility estimation in \citeauthor{ElliottLieli2013}'s (\citeyear{ElliottLieli2013}) model and the issue of its in-sample overfitting.
Section~\ref{selection} presents the construction of a utility-maximizing prediction rule based on different complexity penalties, and non-asymptotic upper bounds on the difference between the generalized expected utility of the prediction rule and the maximal expected utility.
In Section~\ref{simulation}, Monte Carlo experiments are carried out to evaluate the finite-sample performance of the proposed utility-maximizing prediction rule and the aforementioned estimators.
Section~\ref{conclusion} concludes.
Technical proofs and steps of implementing pretest and cross-validation in the maximum utility estimation are collected in the Appendix.

\section{Maximum Utility Estimation}\label{Model}
\subsection{Model}
We start by describing \citeauthor{ElliottLieli2013}'s (\citeyear{ElliottLieli2013}) model of binary decision-making based on binary prediction:
Before the realization of a binary outcome $Y\in\{-1,1\}$, a decision maker aims to choose a binary decision $a\in\{-1,1\}$ to maximize his or her expected utility conditional on a $d$-dimensional vector of observed covariates $X=x\in\mathcal{X}$.
Concretely, the decision maker solves the optimization problem
\begin{equation}\label{Utility}
\max_{a\in\{-1,1\}}\Exp\left[U\left(a,Y,X \right)\mid X=x\right].
\end{equation}
We abbreviate by writing $u_{a,y}(x)=U(a,y,x)$ for notational simplicity.

Elliott and Lieli show that under some regularity assumptions, we can obtain an optimal decision rule (after observing $X=x$)
\begin{align*}
a^{*}(x) \equiv
\begin{cases}
1, & \text{if } p^{*}(x) > c(x), \\
-1, & \text{otherwise,}
\end{cases}
\end{align*}
where
\begin{align*}
c(x)\equiv\frac{u_{-1,-1}(x)-u_{1,-1}(x)}{u_{1,1}(x)-u_{-1,1}(x)+u_{-1,-1}(x)-u_{1,-1}(x)}
\end{align*}
is a cutoff function derived from the utility function, which is known in principle to the decision maker.
Thus, knowledge of the correct conditional probability $p^{*}(x)$ yields the maximal expected utility of the decision.

\citeauthor{ElliottLieli2013}'s (\citeyear{ElliottLieli2013}) insight is that the correct specification of sign$(p^{*}(x)-c(x))$ rather than that of $p^{*}(x)$ is enough to achieve maximal expected utility in (\ref{Utility});
namely, the knowledge of crossing points between the conditional probability $p^{*}(x)$ and the cutoff $c(x)$ is sufficient.
Moreover, they point out that the decision-making problem in (\ref{Utility}) can be equivalently written as
\begin{align*}
\max_{f}S(f)\equiv\Exp\left[b(X)[Y+1-2c(X)]\text{sign}(f(X)-c(X))\right],
\end{align*}
where
\begin{align*}
b(x)\equiv u_{1,1}(x)-u_{-1,1}(x)+u_{-1,-1}(x)-u_{1,-1}(x)
\end{align*}
is the denominator of $c(x)$ and the maximum is taken over all measurable functions from $\mathcal{X}$ to $\mathbb{R}$.
Since $u_{a,y}(x)=0.25b(x)[y+1-2c(x)]a+0.25b(x)[y+1-2c(x)]+u_{-1,y}(x)$, we normalize the utility function by setting $u_{-1,y}(x)=-0.25b(x)[y+1-2c(x)]$ for all $x\in\mathcal{X}$ and call $S(f)$ the expected utility of $f$.

The insight motivates \citet{ElliottLieli2013} to propose the \emph{maximum utility} (MU) estimation.
To be specific, given observations $\{(Y_{i},X_{i})\}_{i=1}^{n}$ with the sample size $n$ and a pre-specified class $\mathcal{F}$ of functions indexed by a finite-dimensional parameter,
we can choose a measurable maximum utility estimator $\hat{f}\in\mathcal{F}$ that satisfies
\begin{align}\label{MUestimator}
\hat{f}
\in\arg\max_{f\in\mathcal{F}}S_{n}(f)\equiv\frac{1}{n}\sum_{i=1}^{n} b(X_{i})[Y_{i}+1-2c(X_{i})]\text{sign}(f(X_{i})-c(X_{i})),
\end{align}
where ``$\arg$'' stands for the set of estimators in $\mathcal{F}$ that achieve the optimum.\footnote{
Since multiplicity of the maximum utility estimator could be present, the analysis in this paper emphasizes the properties of optimand functions.
See \citet{ElliottLieli2013} for the discussion about the lack of identification of optimizers.
}

The Monte Carlo simulation in \citet{ElliottLieli2013} shows that the maximum utility estimation, compared with traditional maximum likelihood approaches, achieves a large improvement in utility especially when the conditional probability $p^{*}$ is misspecified.
However, the in-sample performance of the maximum utility estimation may be attributed to the overfitting.
\citeauthor*{ElliottLieli2013} further make the following comment:
\begin{quote}
\emph{Both ML and MU have a strong tendency to overfit in sample, however the problem seems more severe for the MU method. This creates challenges for model selection.}
\end{quote}

\subsection{Nature of the Overfitting in MU Estimation}\label{Overfitting}
The in-sample overfitting of maximum utility estimation is similar to that of perceptron learning in the machine learning literature.
The simple perceptron learning is a method of binary pattern recognition that establishes classification based on a threshold function $f(x)=\text{sign}(\theta_{1}^{\top}x-\theta_{0})$, where $\theta_{1}\in\mathbb{R}^{d}$ and $\theta_{0}\in\mathbb{R}$.
More variants of perceptron are well documented in \citet{Vapnik2000} and \citet{AnthonyBartlett1999}.
Since it can be shown that for any $(y,x)\in\{-1,1\}\times \mathcal{X}$,
\begin{align*}
&b(x)[y+1-2c(x)]\text{sign}(f(x)-c(x))\\
=& b(x)[y(1-2c(x))+1]
-2b(x)[y(1-2c(x))+1]\Ind{[y\neq \text{sign}(f(x)-c(x))]},
\end{align*}
the optimization problem in (\ref{MUestimator}) can be viewed as the simple perceptron learning in which the cost of misclassification for each observation $(Y_{i},X_{i})$ may be different.
To be specific, the maximum utility estimator satisfies
\begin{align*}
\hat{f}
\in\arg\min_{f\in\mathcal{F}}\frac{1}{n}\sum_{i=1}^{n} b(X_{i})[Y_{i}(1-2c(X_{i}))+1]\Ind{[Y_{i}\neq \text{sign}(f(X_{i})-c(X_{i}))]},
\end{align*}
where $b(X_{i})[Y_{i}(1-2c(X_{i}))+1]$ is the cost of misclassification for the observation $(Y_{i},X_{i})$, and the cost of misclassification is nonnegative under some mild assumptions.
The simple perceptron learning is a special case of the maximum utility estimation because the cost of misclassification is identical for each observation whenever there is a constant $\bar{u}\in\mathbb{R}_{+}$ such that $u_{1,1}(x)-u_{-1,1}(x)=u_{-1,-1}(x)-u_{1,-1}(x)=\bar{u}$ for all $x\in\mathcal{X}$.
In this case, if $\mathcal{F}$ is further parameterized as $\{x\mapsto\theta_{0}+\theta_{1}^{\top}x: (\theta_{0},\theta_{1})^{\top}\in \mathbb{R}^{(d+1)}\}$, the maximum utility estimation reduces to \citeauthor{Manski1975}'s (\citeyear{Manski1975}, \citeyear{Manski1985}) maximum score estimation.
Moreover, even though the cost of misclassification may be different for each observation $(Y_{i},X_{i})$,
when the in-sample observations (training data set) can be perfectly separated by $\mathcal{F}$ (i.e., classification without error), the maximum utility estimation boils down to the simple perceptron learning.
This is because in this case, the cost of misclassification $b(X_{i})[Y_{i}(1-2c(X_{i}))+1]$ has no effect:
\begin{align*}
0=&\min_{f\in\mathcal{F}}\frac{1}{n}\sum_{i=1}^{n} b(X_{i})[Y_{i}(1-2c(X_{i}))+1]\Ind{[Y_{i}\neq \text{sign}(f(X_{i})-c(X_{i}))]}\\
=&\min_{f\in\mathcal{F}}\frac{1}{n}\sum_{i=1}^{n} \Ind{[Y_{i}\neq \text{sign}(f(X_{i})-c(X_{i}))]}.
\end{align*}
Although perfect separation of the in-sample observations could be accomplished by a sufficiently large class of functions, such sophisticated models will also fit the in-sample noise and thus worsen the out-of-sample performance.

\section{Model Selection}\label{selection}
Motivated by possible in-sample overfitting, we adopt the \emph{structural risk minimization} approach (also known as \emph{complexity regularization}) in machine learning to investigate model selection in cost-sensitive binary classification.
More precisely, our goal is to alleviate the overfitting by selecting a maximum utility estimator from some specific class of functions such that this selected maximum utility estimator, compared with maximum utility estimators from other classes of functions, has the largest complexity-penalized empirical utility.

To explain the idea, we first introduce notation.
Let
\begin{align*}
s(y,x,f)=b(x)[y+1-2c(x)]\text{sign}(f(x)-c(x))
\end{align*}
be the utility of the prediction rule $f$ evaluated at the observation $(y,x)$.\footnote{
More precisely, $s(y,x,f)$ is the double extra gain (loss) in utility arising from a match (mismatch) between the decision $\text{sign}(f(x)-c(x))$ and outcome $y$ given the covariate $x$, because
\begin{align*}
s(y,x,f)=
\begin{cases}
2(u_{1,1}(x)-u_{-1,1}(x))\text{sign}(f(x)-c(x)), & \text{if}\; y=1, \\
2(u_{1,-1}(x)-u_{-1,-1}(x))\text{sign}(f(x)-c(x)), & \text{if}\; y=-1.
\end{cases}
\end{align*}
}
A sample of i.i.d.\ observations with sample size $n$ is denoted by $\mathscr{D}_{n}\equiv\{(Y_{i},X_{i})\}_{i=1}^{n}$.
Given a prediction rule $\hat{f}$ constructed based on $\mathscr{D}_{n}$,
let $S(\hat{f})=\Exp[s(Y,X,\hat{f})\mid \mathscr{D}_{n}]$ and $S_{n}(\hat{f})=\frac{1}{n}\sum_{i=1}^{n}s(Y_{i},X_{i},\hat{f})$ be the expected utility and the empirical utility of the prediction rule $\hat{f}$, respectively.
The expectation involved in the definition of $S(\hat{f})$ is taken with respect to an observation $(Y,X)$, which is independent of $\mathscr{D}_{n}$ and distributed as $(Y_{1},X_{1})$.
Said differently, $S(\hat{f})$ measures the decision maker's expected utility if
he or she uses the prediction rule $\hat{f}$, estimated based on $\mathscr{D}_{n}$,
to classify one additional observation $(Y,X)$ drawn independently of $\mathscr{D}_{n}$.
Note that $S(\hat{f})$ could be random because of the random sample $\mathscr{D}_{n}$.
We suppress the possible dependence of $S(\hat{f})$ on $\mathscr{D}_{n}$ for convenience of exposition.

The structural risk minimization approach consists of the following steps.
First we consider a nondecreasing sieve $\{\mathcal{F}_{k}\}_{k=1}^{\infty}$,\footnote{
In the literature on sieve estimation, a metric space $(\mathcal{F}^{*},\rho)$ is usually pre-specified such that $\mathcal{F}\equiv\bigcup_{k=1}^{\infty}\mathcal{F}_{k}$ is dense in $\mathcal{F}^{*}$ with respect to $\rho$.
See for example \citet{GemanHwang1982}.
In this paper, the denseness is however not assumed and $\mathcal{F}^{*}$ can be treated as the collection of all measurable real-valued functions.
}
that is, a hierarchy of classes of functions
\begin{align*}
\mathcal{F}_{1}\subseteq \mathcal{F}_{2}\subseteq\cdots\subseteq \mathcal{F}_{k}\subseteq\cdots\;\text{and}\;
\mathcal{F}\equiv\bigcup_{k=1}^{\infty}\mathcal{F}_{k}.
\end{align*}
For example, $\mathcal{F}_{k}=\mathcal{P}_{k}$ is the class of polynomial transformations on $\mathcal{X}$ of order at most $k$,\footnote{
A polynomial transformation on $\mathcal{X}\subseteq \mathbb{R}^{d}$ of degree at most $k$ is a function of the form $f(x)=c_{0}+\sum_{j=1}^{q}c_{j}\varrho_{j}(x)$, where $(c_{0},c_{1},\ldots,c_{q})\in\mathbb{R}^{(q+1)}$ and $\varrho_{j}(x)=\prod_{\ell=1}^{d}x_{\ell}^{p_{j\ell}}$ with $\sum_{\ell=1}^{d}p_{j\ell}\leq k$ and $p_{j\ell}\in\mathbb{N}\cup\{0\}$ for each $j$ and $q\in\mathbb{N}$.
}
or it can be further transformed by the logistic function as $\mathcal{F}_{k}=\Lambda(\mathcal{P}_{k})\equiv \left\{x\mapsto \Lambda(f(x)):f\in\mathcal{P}_{k}\right\}$, where $\Lambda(v)=(1+\exp{\{-v\}})^{-1}$ for all $v\in\mathbb{R}$.
We refer the reader to \citet{Chen2007} for more examples of sieves.
For each $\mathcal{F}_{k}$, we select a maximum utility estimator
\begin{align}\label{MUE}
\hat{f}_{k}\in\arg\max_{f\in\mathcal{F}_{k}}S_{n}(f).
\end{align}
In addition, we construct a complexity penalty $C_{n}(k;\alpha)$ for $\mathcal{F}_{k}$, where $\alpha>0$ is a tuning parameter for a technical reason.
This technical reason and the issue of choosing $\alpha$ will be discussed later.
Let
\begin{align*}
\tilde{S}_{n}(f;k,\alpha)\equiv S_{n}(f)-C_{n}(k;\alpha)
\end{align*}
be the associated complexity-penalized empirical utility of a prediction rule $f\in\mathcal{F}_{k}$.
Finally, we define a \emph{utility-maximizing prediction rule} (UMPR) as a maximum utility estimator $\hat{f}_{k}$ in (\ref{MUE}) that maximizes $\tilde{S}_{n}(\hat{f}_{k};k,\alpha)$;\footnote{
Precisely, the UMPR should be referred to as $\text{sign}(\tilde{f}_{n}(\cdot\;;\alpha)-c(\cdot))$.
Since it will be clear from the context in this paper, we still reserve the UMPR for the estimator $\tilde{f}_{n}(\cdot\;;\alpha)$.
} that is,
\begin{align*}
\tilde{f}_{n}(x;\alpha)\equiv \hat{f}_{\hat{k}_{n}(\alpha)}(x)\;,\; \text{where}\; \hat{k}_{n}(\alpha)=\arg\max_{k\in\mathbb{N}}\tilde{S}_{n}(\hat{f}_{k};k,\alpha).
\end{align*}
For ease of presentation, we suppress the dependence of $(\tilde{S}_{n}(f;k,\alpha),\hat{k}_{n}(\alpha),\tilde{f}_{n}(\cdot\;;\alpha))$ on $\alpha$ and write $\tilde{S}_{n}(\tilde{f}_{n})=\tilde{S}_{n}(\hat{f}_{\hat{k}_{n}};\hat{k}_{n})$ for the complexity-penalized empirical utility of the UMPR.

The idea of complexity penalization has been used in the selection methods based on information criteria in econometrics.
These methods aim to maximize the complexity-penalized empirical log-likelihood evaluated at the maximum likelihood estimator.
Specifically, let $\mathcal{L}$ be the log-likelihood function of a single observation $(Y,X)$ and $\hat{f}^{\text{ML}}_k$ be the maximum likelihood estimator in $\mathcal{F}_{k}$.
Given a complexity measure $C^{\text{IC}}_{n}(k)$ for $\mathcal{F}_{k}$, we can construct an estimator
\begin{align*}
\tilde{f}^{\text{IC}}_{n}\equiv \hat{f}^{\text{ML}}_{\check{k}_{n}}\;,\; \text{where}\; \check{k}_{n}=\arg\max_{k\in\mathbb{N}}\frac{1}{n}
\sum_{i=1}^{n}\mathcal{L}(\hat{f}^{\text{ML}}_{k}|Y_{i},X_{i})-C^{\text{IC}}_{n}(k).
\end{align*}
Leading examples include the AIC by setting $C^{\text{IC}}_{n}(k)$ to be the number of free parameters in $\mathcal{F}_{k}$ divided by $n$,
and the BIC by setting $C^{\text{IC}}_{n}(k)$ to be the number of free parameters in $\mathcal{F}_{k}$ multiplied by $\log{\{n\}}/(2n)$.
Although the AIC and BIC only differ in the choice of complexity penalty in the selection procedure,
their asymptotic behaviors are different.
Details on these differences can be found in \citet{KonishiKitagawa2008}, \citet{ClaeskensHjort2008}, and references given there.

The UMPR shares a similar motivation with the AIC.
Just as the AIC adjusts the empirical log-likelihood to approximate the expected log-likelihood, the UMPR adjusts the empirical utility to approximate the expected utility.
Both adjustments are fulfilled by subtracting specific complexity penalties.
However, the AIC attempts to recover $p^{*}$, while the UMPR aims to select a model in which the decision maker only focuses on the local fitting at the crossing points between $p^{*}$ and $c$.

As in the penalized likelihood criteria, the choice of complexity penalty $C_{n}(k;\alpha)$ is essential in the proposed penalized empirical utility criteria.
First, the complexity penalty should be constructed without assuming any knowledge of the conditional probability $p^{*}$ so that the UMPR, like the maximum utility estimator, can perform well when $p^{*}$ is misspecified.
More importantly, the complexity penalty should be an appropriate estimate of the magnitude of overfitting $S_{n}(\hat{f}_{k})-S(\hat{f}_{k})$.
In this case, the expected utility $S(\hat{f}_{k})$ can be recovered by the penalized empirical utility $\tilde{S}_{n}(\hat{f}_{k};k)$ and thus the UMPR $\tilde{f}_{n}$ will have the largest expected utility $S(\tilde{f}_{n})$ among the maximum utility estimators $\{\hat{f}_{k}\}_{k=1}^{\infty}$.
Taking these two requirements into account, we will construct complexity penalties, without assuming the knowledge of $p^{*}$, to non-asymptotically bound the in-sample overfitting.
This non-asymptotic complexity-regularized approach is in marked contrast to the information-theoretic approach, for example AIC and GIC, where penalties are obtained by an asymptotic approximation of the Kullback-Leibler divergence but may not effectively control the in-sample overfitting.
As will be shown later, the non-asymptotic complexity penalties may differ from the information-theoretic complexity penalties in the order of the in-sample size.
We discuss both distribution-free and data-dependent penalty terms constructed by the non-asymptotic complexity-regularized approach, and study the theoretical properties of their associated utility-maximizing prediction rules in the following subsections.

\subsection{UMPR with a Distribution-Free Penalty}
Since the seminal work by \citet{VapnikChervonenkis1971}, there have been many improvements in the VC-type upper bound on the uniform deviation of empirical means from their expectations.
\citet{LugosiZeger1996} further applied the VC-type upper bound to finding a complexity penalty in the traditional binary classification.
Motivated by this idea, we aim to construct a VC complexity penalty for the maximum utility estimation.

We start by making the following assumptions.\\[0.3cm]
\noindent
\textbf{Assumptions}
\begin{enumerate}[label=(A\arabic*)]
\item \label{A1} The conditional probability $p^{*}(x)\equiv\mathbb{P}(Y=1\mid X=x)$ does not depend on the decision $a$.
\item \label{A2} For all $x$ in the support $\mathcal{X}\subseteq \mathbb{R}^{d}$ of $X$, $u_{1,1}(x)>u_{-1,1}(x)$ and $u_{-1,-1}(x)>u_{1,-1}(x)$.
\item \label{A3} For any $a,y\in \{1,-1\}$, $u_{a,y}(\cdot)$ is Borel measurable; in addition, there is some $M>0$ such that $|u_{a,y}(x)|\leq M$ for all $x\in\mathcal{X}$ and $a,y\in \{1,-1\}$.
\item \label{A4} For each $k\in\mathbb{N}$, the class $\mathcal{F}_{k}$ of functions is countable.
\end{enumerate}
The first three assumptions are imposed in \citet{ElliottLieli2013}, and the last assumption is imposed to avoid measurability complications.
Assumption~\ref{A1} excludes the possibility of feedback from the binary action to the binary outcome.
Take the financial investment in Section~\ref{Introduction} as an example.
Under Assumption~\ref{A1}, investors are price takers whose decisions on buying an instrument do not affect the possibility of price change.
Assumption~\ref{A2} implies that the decision maker obtains higher utility when the decision matches the outcome;
in particular, there should be a best response $a$ to a realized outcome $Y$ if this realization were observed by the decision maker before making a decision.
This assumption seems plausible in many situations, for example the aforementioned financial investment.
The uniform boundedness imposed by Assumption~\ref{A3} implies some shape constraint on the utility functions, especially when the support $\mathcal{X}$ is unbounded.
However, this assumption could be compatible with some models of the financial investment.
An example is the exponential utility (also known as constant absolute risk aversion preference) used by \citet{ChristensenLarsenEtAl2012} in an asset pricing model when the decision maker is risk averse.
Assumption~\ref{A4} is inconsequential in practice because there are only countably many computable real numbers evaluated for the parameters of $\mathcal{F}_{k}$ by a computer program.
The use of computable real numbers and functions could also be interpreted as a decision maker's computability-bounded rationality, as in \citet{RichterWong1999}.

These assumptions allow us to establish a VC-type upper bound on the uniform deviation of $S_{n}(f)$ from $S(f)$.
\begin{Pro}\label{Prop2}
Suppose that i.i.d.\ data $\mathscr{D}_{n}=\{(Y_{i},X_{i})\}_{i=1}^{n}$ are available.
Under Assumptions~\ref{A1}-\ref{A4}, we have for any $n,k\in\mathbb{N}$ and $\varepsilon>0$,
\begin{align*}
\mathbb{P}\left(\sup_{f\in\mathcal{F}_{k}}\left(S_{n}(f)-S(f)\right)
>8M\sqrt{\frac{2\log\{\Pi_{k,c}(n)\}}{n}}+\varepsilon\right)
\leq \exp{\left\{-\frac{n\varepsilon^{2}}{32M^{2}}\right\}},
\end{align*}
where $\Pi_{k,c}(\cdot)$ is the growth function of $\mathcal{F}_{k,c}\equiv\left\{x\mapsto \text{sign}(f(x)-c(x)):f\in\mathcal{F}_{k}\right\}$.\footnote{
For any collection $\mathcal{H}$ of functions from $\mathcal{X}$ to $\{-1,1\}$,
the \emph{growth function} $\Pi_{\mathcal{H}}:\mathbb{N}\to \mathbb{N}$ of $\mathcal{H}$ is
\begin{align*}
\Pi_{\mathcal{H}}(\ell)=\max_{(x_{1},\ldots,x_{\ell})\in\mathcal{X}^{\ell}}
\left|\{(h(x_{1}),\ldots,h(x_{\ell})):h\in\mathcal{H}\}\right|.
\end{align*}
That is, the growth function $\Pi_{\mathcal{H}}(\ell)$ is the maximum number of distinct ways in which $\ell$ points $(x_{1},\ldots,x_{\ell})$ can be classified using functions in $\mathcal{H}$.
}
\end{Pro}
\noindent
The maximal inequality in Proposition~\ref{Prop2} suggests that large empirical utility arising from sophisticated models does not guarantee large expected utility.
To see this, note that for any $n,k\in\mathbb{N}$ and $\delta\in(0,1)$, with probability at least $1-\delta$,
\begin{align*}
S_{n}(\hat{f}_{k})-S(\hat{f}_{k})\leq 8M\sqrt{\frac{2\log{\{\Pi_{k,c}(n)\}}}{n}}+8M\sqrt{\frac{\log\{1/\delta\}}{2n}}.
\end{align*}
Thus, given the sample size $n$, an increase in $k$ tends to increase empirical utility $S_{n}$, but it may meanwhile increase $\Pi_{k,c}(n)$.
The growth function $\Pi_{k,c}$ measures the complexity of $\mathcal{F}_{k,c}$ to fit in-sample observations.
Clearly, $\Pi_{k,c}(\ell)\leq 2^{\ell}$ for each $\ell\in\mathbb{N}$.
If $\Pi_{k,c}(\ell)<2^{\ell}$ for some $\ell$, the complexity of $\mathcal{F}_{k,c}$ is restricted because some specific $\ell$ observations $\{(Y_{i},X_{i})\}_{i=1}^{\ell}$ cannot be separated by $\mathcal{F}_{k,c}$ without any classification error.\footnote{
There are $2^{\ell}-\Pi_{k,c}(\ell)$ possible realizations of $(Y_{1},\cdots,Y_{\ell})^{\top}$ such that for all $(x_{1},\cdots,x_{\ell})^{\top}$ and $f\in\mathcal{F}_{k}$,
\begin{align*}
(Y_{1},\cdots,Y_{\ell})^{\top}\neq (\text{sign}(f(x_{1})-c(x_{1})),\cdots,\text{sign}(f(x_{\ell})-c(x_{\ell})))^{\top}.
\end{align*}
}

Proposition~\ref{Prop2} also implies the asymptotic behavior of $\sup_{f\in\mathcal{F}_{k}}|S_{n}(f)-S(f)|$ as follows.
\begin{Cor}\label{Corollary1}
Suppose that the growth function $\Pi_{k,c}$ is of polynomial order for each $k\in\mathbb{N}$.
If the assumptions of Proposition~\ref{Prop2} hold,
then for any $k\in\mathbb{N}$ and $\varepsilon>0$, there exists an integer $n^{*}$ such that
\begin{align*}
\mathbb{P}\left(\sup_{f\in\mathcal{F}_{k}}|S_{n}(f)-S(f)|>\varepsilon\right)
\leq 2\exp{\left\{-\frac{n\varepsilon^{2}}{128M^{2}}\right\}}.
\end{align*}
for all $n\geq n^{*}$.
\end{Cor}
\noindent
This corollary immediately guarantees that given i.i.d.\ observations, $|S_{n}(\hat{f}_{k})-S(\hat{f}_{k})|$ converges almost surely to zero whenever $\Pi_{k,c}$ is of polynomial order.
Hence, the technical conditions imposed by Proposition 2 of \citet{ElliottLieli2013} such as compactness of parameter space and lipschitz continuity of functions with respect to the parameter can be relaxed.

More importantly, the maximal inequality in Proposition~\ref{Prop2} is non-asymptotic; thus, it can be used to estimate the upper bound on $S_{n}(\hat{f}_{k})-S(\hat{f}_{k})$ for every finite sample size $n$ when the growth function $\Pi_{k,c}$ is known.
The calculation of $\Pi_{k,c}(n)$ is, however, not easy in practice.
Thus, $\Pi_{k,c}(n)$ is usually replaced with an upper bound
\begin{align*}
\psi_{c}(k,n)=
\begin{cases}
2^{n}, & \text{if $n\leq V_{k,c}$,} \\
\left(\frac{en}{V_{k,c}}\right)^{V_{k,c}}, & \text{if $n> V_{k,c}$,}
\end{cases}
\end{align*}
where $V_{k,c}$ is the \emph{VC dimension} of the class $\mathcal{F}_{k,c}$, which is the largest integer $\ell$ such that $\Pi_{k,c}(\ell)=2^{\ell}$ by definition.
By Proposition~\ref{Prop2}, the replacement of $\Pi_{k,c}(n)$ with $\psi_{c}(k,n)$ shows that for any $n, k\in\mathbb{N}$ and $\varepsilon>0$,
\begin{align*}
\mathbb{P}\left(\sup_{f\in\mathcal{F}_{k}}\left(S_{n}(f)-S(f)\right)
>8M\sqrt{\frac{2\log\{\psi_{c}(k,n)\}}{n}}+\varepsilon\right)
\leq \exp{\left\{-\frac{n\varepsilon^{2}}{32M^{2}}\right\}}.
\end{align*}
The upper bound $\psi_{c}(k,n)$ follows from a combinatorial result, which is known as Sauer's lemma and can be found in Theorems 3.6 and 3.7 of \citet{AnthonyBartlett1999}.
As a parameter involved in $\psi_{c}(k,n)$, the VC dimension $V_{k,c}$, like the growth function, also restricts the complexity of $\mathcal{F}_{k,c}$.
Note that $\Pi_{k,c}(V_{k,c}+1)<2^{(V_{k,c}+1)}$ if $V_{k,c}<\infty$.
Thus, the classification error of using $\mathcal{F}_{k,c}$ to classify some realization of $\{(Y_{i},X_{i})\}_{i=1}^{n}$ is impossibly eliminated whenever the in-sample size $n$ is greater than $V_{k,c}$.
As shown in Theorem 3.5 of \citet{AnthonyBartlett1999}, the VC dimension $V_{k,c}$ is equal to the dimension of $\mathcal{F}_{k}$ if this class $\mathcal{F}_{k}$ is specified as a vector space of real-valued functions.
For example, the VC dimension $V_{k,c}$ is ${d+k \choose k}$ if $\mathcal{F}_{k}$ is the class $\mathcal{P}_{k}$ of polynomial transformations on $\mathcal{X}$ of order at most $k$ in the absence of dummy covariates.
Even if we consider the logit specification, say $\mathcal{F}_{k,c}=\left\{x\mapsto \text{sign}(f(x)-c(x)):f\in\Lambda(\mathcal{P}_{k})\right\}$, then its VC dimension $V_{k,c}$ can be bounded by ${d+k \choose k}+1$.\footnote{
To see this, note that for a class of Boolean functions, the VC dimension is equal to the VC index minus one by definition.
Since both functions $\text{sign}(\cdot)$ and $\Lambda(\cdot)$ are monotone,
the VC index of $\mathcal{F}_{k,c}$ is less than or equal to that of $\mathcal{P}_{k}$ by Lemma 9.9 (v) and (viii) of \citet{Kosorok2008}.
Applying Lemma 9.6 of \citet{Kosorok2008} yields the result.
}
More generally, if $\mathcal{F}_{k}$ is a VC-subgraph class,\footnote{
Let $\mathcal{C}$ be a collection of subsets of $\mathcal{Z}$.
The collection $\mathcal{C}$ is said to shatter a subset $\mathcal{Z}_{\ell}=\{z_{1},\ldots,z_{\ell}\}\subseteq \mathcal{Z}$ if the cardinality of $\left\{\{\mathcal{Z}_{\ell}\cap C\}:C\in\mathcal{C}\right\}$ is equal to $2^{\ell}$.
The collection $\mathcal{C}$ is called a Vapnik-Cervonekis (VC) class if for some $\ell\in\mathbb{N}$, no subset of cardinality $\ell$ is shattered by $\mathcal{C}$.
A collection $\mathcal{F}$ is a VC-subgraph class if the collection $\big\{\{(x,t)\in \mathcal{X}\times \mathbb{R}:t<f(x)\}:f\in\mathcal{F}\big\}$ of all subgraphs is a VC class of sets in $\mathcal{X}\times \mathbb{R}$.
}
then the VC dimension $V_{k,c}$ equals the VC index of $\mathcal{F}_{k}$ minus one by Lemma 9.9 of \citet{Kosorok2008}.
We refer the reader to Section 2.6 of \citet{VaartWellner1996} and Section 9.1 of \citet{Kosorok2008} for properties of a VC-subgraph class.

The easily computable VC-type upper bound permits the construction of a distribution-free complexity penalty.
For each $k$, we consider an estimate of expected utility $S(\hat{f}_{k})$ to be
\begin{align*}
R^{\text{VC}}_{n,k}\equiv S_{n}(\hat{f}_{k})-8M\sqrt{\frac{2\log\{\psi_{c}(k,n)\}}{n}}.
\end{align*}
It follows that for each $k$, we obtain a non-asymptotic upper bound on the tail probability for $R^{\text{VC}}_{n,k}-S(\hat{f}_{k})$.
To be specific, we have for any $n,k\in\mathbb{N}$ and $\varepsilon >0$,
\begin{align}\label{boundofrisk}
\mathbb{P}(R^{\text{VC}}_{n,k}-S(\hat{f}_{k})>\varepsilon)
\leq &\mathbb{P}\left(\sup_{f\in\mathcal{F}_{k}}\left(S_{n}(f)-S(f)\right)
>8M\sqrt{\frac{2\log\{\psi_{c}(k,n)\}}{n}}+\varepsilon\right)\notag\\
\leq & \exp{\left\{-\frac{n\varepsilon^{2}}{32M^{2}}\right\}}.
\end{align}
Following the suggestion in \citet{BartlettBoucheronEtAl2002}, we consider the VC complexity penalty
\begin{align*}
C^{\text{VC}}_{n}(k;\alpha)
&\equiv 8M\sqrt{\frac{2\log\{\psi_{c}(k,n)\}}{n}}+8M\chi_{n}(k;\alpha),
\end{align*}
where
\begin{align*}
\chi_{n}(k;\alpha)\equiv\sqrt{\frac{(1+\alpha)\log\{V_{k,c}\}}{2n}}.
\end{align*}

The VC complexity penalty is the sum of the estimate $S_{n}(\hat{f}_{k})-R_{n,k}$ of the magnitude of overfitting and a technical term $8M\chi_{n}(k;\alpha)$.
Treating $S_{n}(\hat{f}_{k})-R_{n,k}$ as a component of the VC complexity penalty,
the non-asymptotic complexity-regularized approach explicitly accounts for the in-sample overfitting.
The VC complexity penalty differs from the penalties for the information-theoretic approach in the order of $n$.
For example, the VC complexity penalty (without the technical term) is $8M\sqrt{2(k+1)\log\{en/(k+1)\}/n}$ if we consider the specification $\mathcal{F}_{k}=\mathcal{P}_{k}$ of univariate polynomial functions for the UMPR.
Given the logistic transformation of the same specification $\mathcal{F}_{k}=\mathcal{P}_{k}$ in the empirical log-likelihood, the AIC and BIC have the penalties $(k+1)/n$ and $(k+1)\log\{n\}/(2n)$, respectively.
In this example, the VC complexity penalty would be greater than the penalties used in AIC and BIC when a large sample is available.
The different convergence rates should be attributed to the difference in the underlying objective function: the AIC and BIC are both associated with the empirical log-likelihood function, whereas the UMPR is associated with the empirical utility function.\footnote{
The derivation of penalties for AIC and BIC relies on twice differentiability of the empirical log-likelihood function.
See for example Sections 3.4 and 9.1 of \citet{KonishiKitagawa2008}.
The empirical utility function is, however, not differentiable.
Despite the non-differentiability, a minimax lower bound on $\Exp[S_{n}(\hat{f}_{k})]-\Exp[S(\hat{f}_{k})]$ could still be established.
Let $\mathcal{P}(\mathcal{F}_{k})$ be the set of all distributions of $(Y,X)$ such that $p^{*}\in\mathcal{F}_{k}$.
In the special case that $b(x)=b>0$ and $c(x)=1/2$ for all $x\in\mathcal{X}$, if $\mathcal{F}_{k}$ has a finite VC index greater than 2, then Inequality (38) of \citet{MassartNedelec2006} implies $\sup_{\mathbb{P}\in\mathcal{P}(\mathcal{F}_{k})}\Exp[S_{n}(\hat{f}_{k})]-\Exp[S(\hat{f}_{k})]\geq \Bigomega{1/\sqrt{n}}$, where the notation $\Bigomega{1/\sqrt{n}}$ indicates that there exist positive constants $\kappa_{0}$ and $n_{0}$ such that this lower bound is grater than $\kappa_{0}/\sqrt{n}$ for all $n\geq n_{0}$.
Thus, the VC complexity penalty is near optimal in the minimax sense.
This lower bound can be improved under a margin restriction on $p^{*}$.
Details about the margin restriction can be found in \citet{MassartNedelec2006}.
}
The technical term $\chi_{n}(k;\alpha)$ is included in the penalty to guarantee the summability of $\zeta(\alpha)\equiv \sum_{k=1}^{\infty}V_{k,c}^{-(1+\alpha)}$
for some $\alpha_{0}$ such that the union bound holds nontrivially in the proof of the following theorem.
\begin{Thm}\label{MainThm}
Suppose that (i) the data $\mathscr{D}_{n}=\{(Y_{i},X_{i})\}_{i=1}^{n}$ are i.i.d., (ii) Assumptions~\ref{A1}-\ref{A4} hold, (iii) $\mathcal{F}_{k}$ is a VC-subgraph class for each $k$, and (iv) $\zeta(\alpha_{0})<\infty$ for some $\alpha_{0}$.
If the UMPR $\tilde{f}_{n}$ is constructed based on the penalty $C^{\text{VC}}_{n}$ with tuning parameter $\alpha_{0}$, then
for any $n\in\mathbb{N}$ and $\varepsilon>0$,
\begin{align*}
\mathbb{P}\left(\tilde{S}_{n}(\tilde{f}_{n})-S(\tilde{f}_{n})>\varepsilon\right)
\leq \zeta(\alpha_{0})\exp{\left\{-\frac{n\varepsilon^{2}}{32M^{2}}\right\}},
\end{align*}
and
\begin{align*}
S^{*}-\Exp[S(\tilde{f}_{n})]
\leq \min_{k}\left\{C^{\text{VC}}_{n}(k;\alpha_{0})+\left(S^{*}-S_{k}^{*}\right)\right\}
+8M\sqrt{\frac{1+\log\{\zeta(\alpha_{0})\}}{2n}},
\end{align*}
where $S^{*}\equiv \sup_{f}S(f)$ and $S_{k}^{*}\equiv\sup_{f\in\mathcal{F}_{k}}S(f)$ for each $k$.
\end{Thm}

Theorem~\ref{MainThm} implies a probabilistic lower bound on the expected utility $S(\tilde{f}_{n})$; that is, for any $n\in\mathbb{N}$ and $\delta\in(0,1)$,
\begin{align*}
S(\tilde{f}_{n})\geq\tilde{S}_{n}(\tilde{f}_{n})-8M\sqrt{\frac{\log\{\zeta(\alpha_{0})/\delta\}}{2n}}
\end{align*}
with probability at least $1-\delta$.
Theorem~\ref{MainThm} also shows an upper bound on the difference between the maximal expected utility $S^{*}$ and the generalized expected utility $\Exp[S(\tilde{f}_{n})]$.
This upper bound takes into account the trade-off between the complexity penalty $C^{\text{VC}}_{n}(k;\alpha_{0})$ and the approximation error $S^{*}-S^{*}_{k}$.
Furthermore, if the approximation error is equal to zero for some $k$, then $\Exp[S(\tilde{f}_{n})]$ converges to $S^{*}$ because the upper bound on this difference shrinks to zero as the sample size tends to infinity.
In this case, the convergence of $\Exp[S(\tilde{f}_{n})]$ to $S^{*}$ is equivalent to the convergence of $S(\tilde{f}_{n})$ to $S^{*}$ in probability because $\sup_{f\in\mathcal{F}}|S(f)|\leq 4M$ under Assumption~\ref{A3}.
In fact, we can establish the almost sure convergence of $S(\tilde{f}_{n})$ as follows.
\begin{Cor}\label{Corollary2}
Suppose that the assumptions of Theorem~\ref{MainThm} hold.
The UMPR $\tilde{f}_{n}$ constructed based on the penalty $C^{\text{VC}}_{n}$ with tuning parameter $\alpha_{0}$ satisfies
\begin{align*}
\lim_{n\to \infty}S(\tilde{f}_{n})=S^{*}\;\;\text{with probability one}
\end{align*}
for any distribution of $(Y,X)$ such that $\lim_{k\to \infty}S_{k}^{*}=S^{*}$.
\end{Cor}

Corollary~\ref{Corollary2} shows that the UMPR $\tilde{f}_{n}$ with the VC penalty is universally utility consistent because the almost sure convergence holds for every distribution of $(Y,X)$ satisfying $\lim_{k\to \infty}S_{k}^{*}=S^{*}$.
To check this convergence of approximation error for some function classes $\{\mathcal{F}_{k}\}_{k=1}^{\infty}$, we prove the following proposition.

\begin{Pro}\label{PropX1}
Suppose that Assumptions~\ref{A1} and \ref{A2} hold.
For any (measurable) deterministic function $f:\mathcal{X}\mapsto\mathbb{R}$, we have
\begin{align*}
S^{*}-S(f)=4\Exp\left[b(X)[p^{*}(X)-c(X)](\Ind{[p^{*}(X)\geq c(X)]}-\Ind{[f(X)\geq c(X)]})\right]\geq 0,
\end{align*}
and the maximal expected utility $S^{*}$ satisfies
\begin{align*}
S^{*}=S(p^{*})=2\Exp\left[b(X)|p^{*}(X)-c(X)|\right].
\end{align*}

If, in addition, Assumption~\ref{A3} holds, then
\begin{align*}
S^{*}-S(f)\leq4\Exp\left[b(X)|p^{*}(X)-f(X)|\right]\leq 16M\sup_{x\in\mathcal{X}}|p^{*}(x)-f(x)|
\end{align*}
for any deterministic function $f$.
\end{Pro}
\vspace{0.2cm}
\noindent
This proposition implies that for each $k\in\mathbb{N}$,
\begin{align*}
0\leq S^{*}-S_{k}^{*}\leq 16M\inf_{f\in\mathcal{F}_{k}}\sup_{x\in\mathcal{X}}|f(x)-p^{*}(x)|.
\end{align*}
If we specify $\mathcal{F}_{k}$ as the class of polynomial transformations on $\mathcal{X}$ of order at most $k$,
then the Stone-Weierstrass approximation theorem ensures that $\inf_{f\in\mathcal{F}_{k}}\sup_{x\in\mathcal{X}}|f(x)-p^{*}(x)|$ converges to zero as $k$ tends to infinity whenever $p^{*}$ is continuous on the support $\mathcal{X}$ that is a compact subset of $\mathbb{R}^{d}$.
Moreover, if each $r$-th order partial derivative of $f:\mathcal{X}\to\mathbb{R}$ exists and is continuous on $\mathcal{X}$ for all $r\leq s\in\mathbb{N}$, and $\mathcal{X}$ is compact,
then the multivariate Jackson theorem of \citet{BagbyBosEtAl2002} implies that
$\inf_{f\in\mathcal{F}_{k}}\sup_{x\in\mathcal{X}}|f(x)-p^{*}(x)|=\Bigo{k^{-s}}$.
Rather than evaluating the global approximation to $p^{*}$, \citet{ElliottLieli2013} illustrate some preferences and data generating processes of $(Y,X)$ in which finite order polynomial functions in $X$ can completely replicate the crossing points between $p^{*}(x)$ and $c(x)$; more precisely, there is some polynomial function $f_{0}$ with sufficient order such that $\text{sign}(f_{0}(x)-c(x))=\text{sign}(p^{*}(x)-c(x))$ and thus $S^{*}_{k}\equiv\sup_{f\in\mathcal{P}_{k}}S(f)=S^{*}$ for some $k\in\mathbb{N}$ by Proposition~\ref{PropX1}.

By explicitly expressing $S^{*}-S(f)$ for any nonrandom function $f$, Proposition~\ref{PropX1} also confirms \citeauthor{ElliottLieli2013}'s insight that the correct specification of sign$(p^{*}(x)-c(x))$ is sufficient to achieve the maximal expected utility.
Furthermore, Proposition~\ref{PropX1} extends the properties of the Bayes decision rule to the cost-sensitive case.
In this case, the maximal expected utility $S^{*}$ depends on not only the distribution of $(Y,X)$ via the conditional probability $p^{*}$ but also the decision maker's preference via the weight function $b$ and cutoff function $c$.
Corresponding results in the traditional binary classification can be found in Sections 2.4 and 2.5 of \citet{DevroyeGyoerfiEtAl1996}.

\begin{Rmk}
The second part in the complexity penalty $C^{\text{VC}}_{n}(k;\alpha_{0})$ involves a technical term $((1+\alpha_{0})\log{\{V_{k,c}\}}/(2n))^{1/2}$.
Instead of using $(\log{\{k\}}/n)^{1/2}$ as in \citet{BartlettBoucheronEtAl2002},
we replace $k$ with the VC dimension $V_{k,c}$ of $\mathcal{F}_{k,c}$.
For example, when $\mathcal{F}_{k}$ is a class of univariate polynomial functions of order at most $k$,
then $V_{k,c}=k+1$.
We also replace the constant $1$ with $(1+\alpha_{0})/2$ such that $\zeta(\alpha_{0})= \sum_{k=1}^{\infty}V_{k,c}^{-(1+\alpha_{0})}$ is summable.
This condition may hold for different values of $\alpha$.
The selection of $\alpha_{0}$ by the cross-validation is discussed in Appendix~\ref{SelectAlpha}.

In practice, researchers may expect that only certain classes of functions are worth consideration.
For example, domain knowledge could suggest that higher-order interactions should be of limited importance, as argued in \citet{AtheyImbens2019}.
In this case, the UMPR is selected from a few classes of functions, and selection of $\alpha_{0}$ is not an issue because the technical term $8M\chi_{n}(k;\alpha)$ can be removed from the complexity penalty.
\end{Rmk}

\begin{Rmk}
The proposed complexity-regularized approach can also be applied to variable selection problems, which recently have attracted much attention in the econometrics literature.
To see this application, let $X=(X_{1},\dots,X_{d})^{\top}$ be a $d$-dimensional vector of covariates and $\mathcal{F}_{\mathscr{V}}$ be the class of linear functions of covariates in a nonempty set $\mathscr{V}\subseteq\{X_{1},\dots,X_{d}\}$ with cardinality $|\mathscr{V}|$.
Instead of specifying nested models, we consider nonnested models in variable selection problems.
For example, $\mathcal{F}_{\{X_{2}\}}=\{X_{2}\mapsto \beta_{0}+\beta_{1}X_{2}: (\beta_{0},\beta_{1})^{\top}\in\mathbb{R}^{2}\}$, $\mathcal{F}_{\{X_{1},X_{3}\}}=\{(X_{1},X_{3})\mapsto \beta_{0}+\beta_{1}X_{1}+\beta_{2}X_{3}: (\beta_{0},\beta_{1},\beta_{2})^{\top}\in\mathbb{R}^{3}\}$, and $\mathcal{F}_{\{X_{2}\}}$ is neither a subset nor a superset of $\mathcal{F}_{\{X_{1},X_{3}\}}$.
The VC complexity penalty of $\mathcal{F}_{\mathscr{V}}$ is
\begin{align*}
C^{\text{VC}}_{n}(\mathscr{V})=8M\sqrt{\frac{2\log\{\psi_{\mathscr{V}}(n)\}}{n}},\;\;\text{where}\;\;
\psi_{\mathscr{V}}(n)=
\begin{cases}
2^{n}, & \text{if $n\leq |\mathscr{V}|+1$,} \\
\left(\frac{en}{|\mathscr{V}|+1}\right)^{(|\mathscr{V}|+1)}, & \text{if $n> |\mathscr{V}|+1$}.
\end{cases}
\end{align*}
The nonempty subset $\hat{\mathscr{V}}_{n}$ of $\{X_{1},\dots,X_{d}\}$ is selected
if it has the largest associated complexity-penalized empirical utility among all classes in
\begin{align*}
\mathcal{F}=\cup\left\{\mathcal{F}_{\mathscr{V}}:\mathscr{V}\;\text{is a nonempty subset of}\; \{X_{1},\dots,X_{d}\}\right\}.
\end{align*}
Specifically,
\begin{align*}
\hat{\mathscr{V}}_{n}=\arg\max \left\{S_{n}(\hat{f}_{\mathscr{V}})-C^{\text{VC}}_{n}(\mathscr{V}):
\mathscr{V}\;\text{is a nonempty subset of}\; \{X_{1},\dots,X_{d}\}\right\},
\end{align*}
where $\hat{f}_{\mathscr{V}}\in\arg\max_{f\in\mathcal{F}_{\mathscr{V}}}S_{n}(f)$ is a maximum utility estimator.
The UMPR is defined as $\tilde{f}_{n}\equiv \hat{f}_{\hat{\mathscr{V}}_{n}}$.

A non-asymptotic upper bound on the difference between the maximal expected utility $S^{*}$ and the generalized expected utility $\Exp[S(\tilde{f}_{n})]$ can still be established.
Suppose that Assumptions~\ref{A1}-\ref{A4} hold.
It can be shown that for any $n\in\mathbb{N}$,
\begin{align*}
S^{*}-\Exp[S(\tilde{f}_{n})]
\leq \min_{\mathscr{V}}\left\{C^{\text{VC}}_{n}(\mathscr{V})
+\left(S^{*}-S_{\mathscr{V}}^{*}\right)\right\}
+8M\sqrt{\frac{1+d\log\{2\}}{2n}},
\end{align*}
where $S_{\mathscr{V}}^{*}\equiv\sup_{f\in\mathcal{F}_{\mathscr{V}}}S(f)$ for every $\mathscr{V}$.
The derivation details are omitted, as they are similar to the arguments in the proof of Theorem~\ref{MainThm}.
Thus, if the approximation error $S^{*}-S_{\mathscr{V}}^{*}$ is equal to zero for some $\mathscr{V}$ and $d=\Smallo{n/\log\{n\}}$,
then $S(\tilde{f}_{n})$ converges in mean and in probability to $S^{*}$.
\end{Rmk}

The upper bound in Proposition~\ref{Prop2} is distribution-free in the sense that it is valid for any distribution of $(Y,X)$.
Since the distributional properties are ignored, this VC-type upper bound is generally loose.
The looseness is even exacerbated by the replacement of the growth function with an upper bound via Sauer's lemma.
Although the distribution of $(Y,X)$ is unknown, its distributional information could be inferred from the sample.
As shown in \citeauthor{Lozano2000}'s (\citeyear{Lozano2000}) simulation results for the interval model selection problem, the data driven penalization can track the magnitude of overfitting better than the VC-type penalization.
Thus, we expect that using data-dependent complexity penalties, instead of the distribution-free complexity penalty $C^{\text{VC}}_{n}(k;\alpha_{0})$, will improve the predictive performance of the UMPR.

\subsection{UMPR with a Data-Dependent Penalty}\label{datadependent}
Heuristically, the magnitude of overfitting is bounded by $\max_{f\in\mathcal{F}_{k}}\left(S_{n}(f)-S'_{n}(f)\right)$,\footnote{
It can be shown that $\Exp[\sup_{f\in\mathcal{F}_{k}}\left(S_{n}(f)-S(f)\right)]\leq \Exp[\max_{f\in\mathcal{F}_{k}}\left(S_{n}(f)-S'_{n}(f)\right)]$ by the common symmetrization argument.
By \citeauthor{McDiarmid1989}'s (\citeyear{McDiarmid1989}) inequality, there is a constant $c_{0}>0$ such that for any $\delta\in(0,1)$,
$\sup_{f\in\mathcal{F}_{k}}\left(S_{n}(f)-S(f)\right)
-\max_{f\in\mathcal{F}_{k}}\left(S_{n}(f)-S'_{n}(f)\right)\leq\sqrt{\ln\{1/\delta\}/c_{0}n}$
with probability at least $1-\delta$.
}
where $S'_{n}(f)$ is the empirical utility of $f$ based on the ghost sample $\mathscr{D}'_{n}$,
in which the observations $(Y'_{1},X'_{1}),\ldots,(Y'_{n},X'_{n})$ are distributed as $(Y_{1},X_{1}),\ldots,(Y_{n},X_{n})$ and independent of them.
Although the lack of the ghost sample $\mathscr{D}'_{n}$ invalidates the direct estimation of $\max_{f\in\mathcal{F}_{k}}\left(S_{n}(f)-S'_{n}(f)\right)$, this idea allows us to develop data-dependent complexity penalties.
Each of them, similar to the VC counterpart, is the sum of a technical term with $\chi_{n}(k;\alpha)$ and an estimate of $\max_{f\in\mathcal{F}_{k}}\left(S_{n}(f)-S'_{n}(f)\right)$.
Different estimates generate different complexity penalties as follows.

\begin{enumerate}[label=\arabic*.]
\item \emph{Maximal Discrepancy }(MD)\\
We partition the sample into two nonoverlapping and roughly equal-sized subsamples.
For notational simplicity, suppose the sample $\mathscr{D}_{n}$ is partitioned into two subsamples $\mathscr{D}^{(1)}_{n/2}=\{(Y_{i},X_{i})\}_{i=1}^{n/2}$ and $\mathscr{D}^{(2)}_{n/2}=\{(Y_{i},X_{i})\}_{i=n/2+1}^{n}$, where the sample size $n$ is even.
We define the maximal discrepancy complexity penalty to be
\begin{align*}
C^{\text{MD}}_{n}(k;\alpha)
\equiv \max_{f\in\mathcal{F}_{k}}\left(\frac{2}{n}\sum_{i=1}^{n/2}s(Y_{i},X_{i},f)
-\frac{2}{n}\sum_{i=n/2+1}^{n}s(Y_{i},X_{i},f)\right)+24M\chi_{n}(k;\alpha)
\end{align*}
as if $\mathscr{D}^{(1)}_{n/2}$ and $\mathscr{D}^{(2)}_{n/2}$ were the sample and the ghost sample, respectively.
The penalization by maximal discrepancy is proposed by \citet{BartlettBoucheronEtAl2002} in the traditional binary classification.
We expect the maximal discrepancy complexity penalty is an appropriate estimate of $\max_{f\in\mathcal{F}_{k}}\left(S_{n}(f)-S'_{n}(f)\right)$ if the sample size is large.
\item \emph{Simulated Maximal Discrepancy }(SMD)\\
We can pair up observations between two tentatively pre-specified subsamples, randomly exchange the subsample labels of paired observations, and calculate more maximal discrepancy complexity penalty terms.
Repeating the random exchange mechanism $m$ times for the pre-specified subsamples $\{(Y_{2i-1},X_{2i-1})\}_{i=1}^{n/2}$ and $\{(Y_{2i},X_{2i})\}_{i=1}^{n/2}$ yields the simulated maximal discrepancy complexity penalty
\begin{align*}
C^{\text{SMD}}_{n}(k;\alpha,m)
\equiv&\frac{1}{m}\sum_{j=1}^{m}
\left(\max_{f\in\mathcal{F}_{k}}\frac{2}{n}\sum_{i=1}^{n/2}\sigma^{(j)}_{i}
\Big(s(Y_{2i-1},X_{2i-1},f)-s(Y_{2i},X_{2i},f)\Big)\right)\\
&+\gamma_{m,n}(M)\chi_{n}(k;\alpha),
\end{align*}
where $\{\sigma^{(j)}\}_{j=1}^{m}
=\{(\sigma^{(j)}_{1},\sigma^{(j)}_{2},\ldots,\sigma^{(j)}_{n/2})\}_{j=1}^{m}$
is the collection of i.i.d.\ Rademacher random vectors (i.e., $\mathbb{P}(\sigma^{(j)}_{i}=1)=\mathbb{P}(\sigma^{(j)}_{i}=-1)=1/2$) that are independent of $\mathscr{D}_{n}$,
and $\gamma_{m,n}$ is a deterministic function that satisfies
\begin{align}\label{gamma}
\gamma_{m,n}(M)=
\begin{cases}
40M, & \text{if}\;\; n \leq m<\infty, \\
(16\ell+40)M, & \text{if}\;\; n/(\ell+1)^{2} \leq m < n/\ell^{2}\;\;
\text{and}\;\; \ell\in\mathbb{N}.
\end{cases}
\end{align}
We need $\gamma_{m,n}$ to control the extra randomness introduced by the simulated random vectors.
Conceptually, we could set $\gamma_{m,n}(M)=24M$ as in the MD penalty if $m=\infty$, the case in which the extra randomness is eliminated.

\item \emph{Rademacher Complexity }(RC)\\
If the ghost sample $\mathscr{D}'_{n}$ were at hand, the pairing and exchange mechanism could be applied to $\mathscr{D}_{n}$ and $\mathscr{D}'_{n}$.
Suppose we draw a sequence $(\sigma_{1},\sigma_{2},\ldots,\sigma_{n})$ of i.i.d.\ Rademacher random variables that are independent of $\mathscr{D}_{n}$ and $\mathscr{D}'_{n}$.
Since observations are i.i.d., $\max_{f\in\mathcal{F}_{k}}\left(S_{n}(f)-S'_{n}(f)\right)$ is identically distributed as
\begin{align*}
\max_{f\in\mathcal{F}_{k}}\frac{1}{n}\sum_{i=1}^{n}
\sigma_{i}\left(s(Y_{i},X_{i},f)-s(Y'_{i},X'_{i},f)\right),
\end{align*}
which has expectation bounded above by the Rademacher complexity
\begin{align*}
\Exp\left[\max_{f\in\mathcal{F}_{k}}\frac{2}{n}\sum_{i=1}^{n}\sigma_{i}s(Y_{i},X_{i},f)\right].
\end{align*}

We can consider the simulated Rademacher complexity penalty
\begin{align*}
C^{\text{RC}}_{n}(k;\alpha,m)
\equiv
\frac{1}{m}\sum_{j=1}^{m}\left(\max_{f\in\mathcal{F}_{k}}
\frac{2}{n}\sum_{i=1}^{n}\sigma^{(j)}_{i}s(Y_{i},X_{i},f)\right)
+\gamma_{m,n}(M)\chi_{n}(k;\alpha),
\end{align*}
where $\{\sigma^{(j)}\}_{j=1}^{m}= \{(\sigma^{(j)}_{1},\sigma^{(j)}_{2},\ldots,\sigma^{(j)}_{n})\}_{j=1}^{m}$ is the collection of i.i.d.\ Rademacher random vectors that are independent of $\mathscr{D}_{n}$, and $\gamma_{m,n}$ is given in (\ref{gamma}).
Proposed by \citet{Koltchinskii2001} and \citet{BartlettBoucheronEtAl2002} in the traditional binary classification, the Rademacher complexity and its variants are usually applied to complexity regularization; see for example \citet{Koltchinskii2011}.

\item \emph{Bootstrap Complexity }(BC)\\
Following \citeauthor{Fromont2007}'s (\citeyear{Fromont2007}) idea, we can apply \citeauthor{Efron1979}'s (\citeyear{Efron1979}) bootstrap to the construction of complexity penalty by replacing the Rademacher random variables with the multinomial random weights minus one.
Specifically, we treat the bootstrap complexity penalty as
\begin{align*}
C_{n}^{\text{BC}}(k;\alpha,m)
\equiv&\left(\frac{n}{n-1}\right)^{n}\frac{1}{m}\sum_{j=1}^{m}
\left(\max_{f\in\mathcal{F}_{k}}\frac{1}{n}
\sum_{i=1}^{n}\left(W^{(j)}_{n,i}-1\right)s(Y_{i},X_{i},f)\right)\\
&+\gamma'_{m,n}(M)\chi_{n}(k;\alpha),
\end{align*}
where $\{W^{(j)}_{n}\}_{j=1}^{m}=\{(W^{(j)}_{n,1},W^{(j)}_{n,2},\ldots,W^{(j)}_{n,n})\}_{j=1}^{m}$ is the collection of i.i.d. multinomial vectors with parameters $n$ and $(1/n,1/n,\ldots,1/n)$ such that $\{W^{(j)}_{n}\}_{j=1}^{m}$ is independent of $\mathscr{D}_{n}$,
and
\begin{align}\label{gammaprime}
\gamma'_{m,n}(M)=
\begin{cases}
56M, & \text{if}\;\; n \leq m<\infty, \\
(32\ell+56)M, & \text{if}\;\; n/(\ell+1)^{2} \leq m < n/\ell^{2}\;\;
\text{and}\;\; \ell\in\mathbb{N}.
\end{cases}
\end{align}

\end{enumerate}

To study the performance of the UMPR $\tilde{f}_{n}$ with each of these data-dependent complexity penalties,
we evaluate the difference between the generalized expected utility $\Exp[S(\tilde{f}_{n})]$ and the maximal expected utility $S^{*}$.
The upper bounds on $S^{*}-\Exp[S(\tilde{f}_{n})]$ in Theorem~\ref{ThmMD}, resembling the VC counterpart in Theorem~\ref{MainThm}, have a similar trade-off between the associated expected complexity penalty and the approximation error.
Note that the data-dependent complexity penalties are all random, whereas the VC complexity penalty is deterministic.

\begin{Thm}\label{ThmMD}
Let $\gamma_{m,n}$ and $\gamma'_{m,n}$ be the functions given in (\ref{gamma}) and (\ref{gammaprime}), respectively, where $m$ is the number of simulation replications for SMD, RC, and BC penalties.
Suppose that (i) the data $\mathscr{D}_{n}=\{(Y_{i},X_{i})\}_{i=1}^{n}$ are i.i.d., (ii) Assumptions~\ref{A1}-\ref{A4} hold, (iii) $\mathcal{F}_{k}$ is a VC-subgraph class for each $k$, and (iv) $\zeta(\alpha_{0})<\infty$ for some $\alpha_{0}$.
\begin{enumerate}[label=\arabic*.]
\item If the UMPR $\tilde{f}_{n}$ is constructed based on the penalty $C^{\text{MD}}_{n}$ with tuning parameter $\alpha_{0}$,
then we have for any $n\in\mathbb{N}$ and $\varepsilon>0$,
\begin{align*}
\mathbb{P}\left(\tilde{S}_{n}(\tilde{f}_{n})-S(\tilde{f}_{n})>\varepsilon\right)
\leq \zeta(\alpha_{0})\exp{\left\{-\frac{n\varepsilon^{2}}{288M^{2}}\right\}},
\end{align*}
and
\begin{align*}
&S^{*}-\Exp[S(\tilde{f}_{n})]\\
\leq &\min_{k}\left\{\Exp\left[C^{\text{MD}}_{n}(k;\alpha_{0})\right]+\left(S^{*}-S_{k}^{*}\right)\right\}
+24M\sqrt{\frac{1+\log\{\zeta(\alpha_{0})\}}{2n}}.
\end{align*}
\item If the UMPR $\tilde{f}_{n}$ is constructed based on the penalty $C^{\text{SMD}}_{n}$ with tuning parameter $\alpha_{0}$,
then we have for any $n\in\mathbb{N}$ and $\varepsilon>0$,
\begin{align*}
\mathbb{P}\left(\tilde{S}_{n}(\tilde{f}_{n})-S(\tilde{f}_{n})>\varepsilon\right)
\leq 2\zeta(\alpha_{0})\exp{\left\{-\frac{2n\varepsilon^{2}}{(\gamma_{m,n}(M))^{2}}\right\}},
\end{align*}
and
\begin{align*}
&S^{*}-\Exp[S(\tilde{f}_{n})]\\
\leq &\min_{k}\left\{\Exp\left[C^{\text{SMD}}_{n}(k;\alpha_{0},m)\right]
+\left(S^{*}-S_{k}^{*}\right)\right\}
+\gamma_{m,n}(M)\sqrt{\frac{1+\log\{2\zeta(\alpha_{0})\}}{2n}}.
\end{align*}
\item If the UMPR $\tilde{f}_{n}$ is constructed based on the penalty $C_{n}^{\text{RC}}$ with tuning parameter $\alpha_{0}$,
then we have for any $n\in\mathbb{N}$ and $\varepsilon>0$,
\begin{align*}
\mathbb{P}\left(\tilde{S}_{n}(\tilde{f}_{n})-S(\tilde{f}_{n})>\varepsilon\right)
\leq 2\zeta(\alpha_{0})\exp{\left\{-\frac{2n\varepsilon^{2}}{(\gamma_{m,n}(M))^{2}}\right\}},
\end{align*}
and
\begin{align*}
&S^{*}-\Exp[S(\tilde{f}_{n})]\\
\leq &\min_{k}\left\{\Exp\left[C^{\text{RC}}_{n}(k;\alpha_{0},m)\right]
+\left(S^{*}-S_{k}^{*}\right)\right\}
+\gamma_{m,n}(M)\sqrt{\frac{1+\log\{2\zeta(\alpha_{0})\}}{2n}}.
\end{align*}
\item If the UMPR $\tilde{f}_{n}$ is constructed based on the penalty $C_{n}^{\text{BC}}$ with tuning parameter $\alpha_{0}$,
then we have for any integer $n\geq 2$ and $\varepsilon>0$,
\begin{align*}
\mathbb{P}\left(\tilde{S}_{n}(\tilde{f}_{n})-S(\tilde{f}_{n})>\varepsilon\right)
\leq 2\zeta(\alpha_{0})\exp{\left\{-\frac{2n\varepsilon^{2}}{(\gamma'_{m,n}(M))^{2}}\right\}},
\end{align*}
and
\begin{align*}
&S^{*}-\Exp[S(\tilde{f}_{n})]\\
\leq &\min_{k}\left\{\Exp\left[C^{\text{BC}}_{n}(k;\alpha_{0},m)\right]
+\left(S^{*}-S_{k}^{*}\right)\right\}
+\gamma'_{m,n}(M)\sqrt{\frac{1+\log\{2\zeta(\alpha_{0})\}}{2n}}.
\end{align*}
\end{enumerate}
\end{Thm}

We can show that if the ratio $m/n$ is bounded away from zero, then the expected value of each data-dependent complexity penalty in this section shrinks to zero at the rate $\Bigo{n^{-1/2}}$, which is slightly faster than the convergence rate of the VC complexity penalty.
Hence, if the approximation error $S^{*}-S^{*}_{k}$ is equal to zero for some $k$, then $S^{*}-\Exp[S(\tilde{f}_{n})]=\Bigo{n^{-1/2}}$.
Under the same assumptions, we can also demonstrate the universal utility consistency of the UMPR $\tilde{f}_{n}$ with any data-dependent complexity penalty above.
These results are summarized in Corollary~\ref{Corrollary3}.

\begin{Cor}\label{Corrollary3}
Suppose that the assumptions of Theorem~\ref{ThmMD} hold.
If, in addition, $m/n\geq 1/\bar{\ell}^{2}$ for some positive integer $\bar{\ell}$,
then there are positive constants $\kappa_{1}$ and $\kappa_{2}$ only depending on $M$, and $\kappa_{3}$ depending on $(M,\bar{\ell})$ such that for each $k\in\mathbb{N}$ and $n\geq 8$,
\begin{align*}
&\max\big\{\Exp\left[C^{\text{MD}}_{n}(k;\alpha_{0})\right], \Exp\left[C^{\text{SMD}}_{n}(k;\alpha_{0},m)\right],
\Exp\left[C^{\text{RC}}_{n}(k;\alpha_{0},m)\right], \Exp\left[C_{n}^{\text{BC}}(k;\alpha_{0},m)\right]\big\}\\
\leq & \kappa_{1}\sqrt{\frac{V_{k,c}}{n}}+\kappa_{2}V_{k,c}\frac{(\log\{n\})^{2}}{n}
+\kappa_{3}\sqrt{1+\alpha_{0}}\sqrt{\frac{\log{\{V_{k,c}\}}}{n}}.
\end{align*}
Moreover, the UMPR $\tilde{f}_{n}$ constructed based on any aforementioned data-dependent complexity penalty with tuning parameter $\alpha_{0}$ satisfies
\begin{align*}
\lim_{n\to \infty}S(\tilde{f}_{n})=S^{*}\;\;\text{with probability one}
\end{align*}
for any distribution of $(Y,X)$ such that $\lim_{k\to \infty}S_{k}^{*}=S^{*}$.
\end{Cor}

\section{Simulation}\label{simulation}
To study the finite-sample performance of the UMPR with any complexity penalty in the previous section, we carried out Monte Carlo experiments.
The simulation designs are those in \citet{ElliottLieli2013}. Specifically, we consider two data generating processes:
\begin{enumerate}[label=DGP \arabic*,leftmargin=*]
 \item \label{DGP1} The covariate $X$ follows the distribution $5\cdot \text{beta}(1,1.3)-2.5$ and $p^{*}(X)=\Lambda(-0.5X+0.2X^{3})$;
 \item \label{DGP2} Both covariates $X_{1}$ and $X_{2}$ are independent and uniformly distributed on $[-3.5,3.5]$ and $p^{*}(X_{1},X_{2})=\Lambda(Q(1.5X_{1}+1.5X_{2}))$, where $Q(v)=(1.5-0.1v)\exp\{-(0.25v+0.1v^{2}-0.04v^{3})\}$.
\end{enumerate}
In addition, we consider four preferences:
\begin{enumerate}[label=Preference \arabic*,leftmargin=*]
 \item \label{Pref1} \hspace{0.1cm} $b(X)=20$ and $c(X)=0.5$;
 \item \label{Pref2} \hspace{0.1cm} $b(X)=20$ and $c(X)=0.5+0.025X$;
 \item \label{Pref3} \hspace{0.1cm} $b(X_{1},X_{2})=20$ and $c(X_{1},X_{2})=0.75$;
 \item \label{Pref4} \hspace{0.1cm} $b(X_{1},X_{2})=20+40\cdot\Ind{[|X_{1}+X_{2}|<1.5]}$ and $c(X_{1},X_{2})=0.75$.
\end{enumerate}
The first two preferences are associated with \ref{DGP1}, whereas the last two preferences are associated with \ref{DGP2}.
For \ref{DGP1} together with either preference 1 or 2, not only the cubic ML but also the cubic MU are correctly specified because there are three crossing points between the conditional probability $p^{*}$ and the cutoff function $c$ in the support of $X$.\footnote{
By the cubic MU, we mean that the MU optimization is taken over the class $\mathcal{P}_{3}$ of polynomial transformations of order at most 3.
Similarly, we refer to the cubic ML as the maximum likelihood estimation with optimization taken over the class $\Lambda(\mathcal{P}_{3})$.
}
Although any logit model is misspecified for \ref{DGP2}, \citet{ElliottLieli2013} demonstrate that the cubic MU is correctly specified in the sense that for all $x\in\mathcal{X}$, $\text{sign}(f(x))=\text{sign}(p^{*}(x)-c(x))$ for some $f\in\mathcal{P}_{3}$.

We evaluate different selection methods for the cost-sensitive binary classification.
In addition to the proposed UMPR with different complexity penalties, we study the pretest estimator adopted by \citet{ElliottLieli2013} and a tenfold cross-validatory estimator in the maximum utility estimation.
For the MU, UMPR, pretest and cross-validatory estimators, we specify the hierarchy $\{\mathcal{F}_{k}\}_{k=1}^{\infty}$ of classes as $\mathcal{F}_{k}=\mathcal{P}_{k}$ for $k\in\{1,2\}$ and $\mathcal{F}_{k}=\mathcal{P}_{3}$ for all $k\geq 3$.
Moreover, we compare prediction rules based on penalized empirical utility criteria with those based on penalized likelihood criteria.
We consider the UMPR for the former criteria, but the AIC and BIC for the latter criteria.
For the ML, AIC and BIC, we specify the hierarchy $\{\mathcal{F}_{k}\}_{k=1}^{\infty}$ of classes as $\mathcal{F}_{k}=\Lambda(\mathcal{P}_{k})$ for $k\in\{1,2\}$ and $\mathcal{F}_{k}=\Lambda(\mathcal{P}_{3})$ for all $k\geq 3$.
We also compute the tenfold cross-validatory LASSO (i.e., logistic loss with an $\ell_{1}$ penalty) with optimization taken over the class $\Lambda(\mathcal{P}_{3})$ and $\ell_{1}$-norm SVM (i.e., hinge loss with an $\ell_{1}$ penalty) with optimization taken over the class $\mathcal{P}_{3}$.
The steps of pretesting and cross-validation in maximum utility estimation are described in Appendix~\ref{CrossValidation}, whereas the implementation of LASSO and $\ell_{1}$-norm SVM can be found in \citet{EfronHastie2016} and \citet{FungMangasarian2004}, respectively.\footnote{
As suggested for the $\ell_{2}$-norm SVM in \citet{Shawe-TaylorCristianini2004}, we construct the $\ell_{1}$-norm SVM estimator $\hat{f}_{\text{SVM}}$ based on standardized covariates in the following Monte Carlo experiments.
Additionally, we use the logistic transformation $\tilde{f}_{\text{SVM}}\equiv \Lambda(\hat{f}_{\text{SVM}})$ to evaluate its predictive performance.
This transformation not only makes $\tilde{f}_{\text{SVM}}$ comparable with other competing estimators in binary classification but also maintains the classification rule of $\hat{f}_{\text{SVM}}$ because $\text{sign}(\hat{f}_{\text{SVM}}(x))=\text{sign}(\tilde{f}_{\text{SVM}}(x)-c(x))$ provided $c(x)=1/2$ for all $x\in\mathcal{X}$.}
To evaluate the performance of a prediction rule $f^{\dag}_{n}$, we compute its relative generalized expected utility
\begin{align*}
\text{RGEU}(f^{\dag}_{n})\equiv \frac{\Exp[S(f^{\dag}_{n})]}{S^{*}},
\end{align*}
which can be approximated via simulation because
\begin{align*}
\text{RGEU}(f^{\dag}_{n})
=\Exp\left[\frac{S(f^{\dag}_{n})}{S(p^{*})}\right]
\simeq \frac{1}{\mathcal{S}}\sum_{j=1}^{\mathcal{S}}
\frac{S_{\ell,j}(f^{\dag}_{n}|\mathscr{D}_{n,j})}{S_{\ell,j}(p^{*})},
\end{align*}
where $S_{\ell,j}(f^{\dag}_{n}|\mathscr{D}_{n,j})$ is the $j$-th out-of-sample empirical utility with size $\ell$ of $f^{\dag}_{n}$, constructed by the $j$-th training data $\mathscr{D}_{n,j}$ with size $n$, $S_{\ell,j}(p^{*})$ is the $j$-th out-of-sample empirical utility with size $\ell$ of $p^{*}$, and $\mathcal{S}$ is the number of simulation replications.
In the following experiments, we set $n\in\{500,1000\}$, $\ell=5000$, and $\mathcal{S}=500$; additionally, we take $m=10$ for the SMD, RC, and BC penalties.

Table~\ref{Table1} presents the relative generalized expected utility of ML, MU, and UMPR with VC and MD complexity penalties under different designs when $n=500$ and $n=1000$.
As expected, a correctly specified ML achieves the largest relative generalized expected utility among these estimators for \ref{DGP1}.
However, a misspecified ML, compared with MU and UMPR, usually has the worst performance.
In addition, for each tuning parameter $\alpha\in\{1,0.5,0.1,0.05\}$, the UMPR with MD penalty outperforms its VC counterpart for~\ref{DGP1}, but the dominance is not clear for~\ref{DGP2}.
As $\alpha$ decreases (i.e., the technical term is smaller), the UMPR with MD penalty might have slightly larger relative generalized expected utility.
However, the UMPR with VC penalty has the same relative generalized expected utility as the linear MU for \ref{DGP2}.
This is a caveat that the correctly specified cubic MU is never selected out of 500 simulation replications by the UMPR with VC penalty.
This phenomenon arises probably because the distribution-free complexity penalty used to construct the UMPR is too conservative.
Using  tenfold cross-validated $\hat{\alpha}$ selected from $\{1,0.5,0.1,0.05\}$ does not improve the performance of the UMPR with VC and MD penalties;
however, excluding the technical term ($8M\chi_{n}(k;\alpha)$ and $24M\chi_{n}(k;\alpha)$ for VC and MD penalty, respectively) would yield an increase in relative generalized expected utility.
These results imply that the UMPR with MD penalty is more adept at selecting the MU estimator with the largest utility than its VC counterpart.
Thus, we focus on the UMPR with data-dependent complexity penalties exclusive of the technical term hereafter.

In addition to the comparison between MU and ML, we compare prediction rules based on penalized empirical utility criteria with those based on penalized likelihood criteria.
Table~\ref{Table2} reports the relative generalized expected utility of UMPR, AIC, BIC, LASSO, and $\ell_{1}$-norm SVM under different designs when $n=500$ and $n=1000$.
We see that the performance of the UMPR relies on the choice of data-dependent complexity penalties.
Among these penalties, SMD, RC, and BC might be better than MD in terms of relative generalized expected utility.
We also see that the AIC and BIC outweigh the UMPR for~\ref{DGP1}, and this reflects the consistent selection of the cubic ML by the AIC and BIC, a property shown in \citet{SinWhite1996}.
However, the AIC and BIC, selecting a mispecified logit model by the penalized likelihood, are dominated by the UMPR for~\ref{DGP2}.
Furthermore, compared with the AIC and BIC, the LASSO has poorer performance for~\ref{DGP1}, but almost the same performance for~\ref{DGP2}.
Outweighing the LASSO for~\ref{DGP1}, the SVM has the worst performance for~\ref{DGP2}.
Such bad performance could be attributed to the fact that the SVM aims to recover $\text{sign}(p^{*}(x)-1/2)$ rather than $\text{sign}(p^{*}(x)-c(x))$, but the cutoff function $c(x)$ is markedly different from $1/2$ for~\ref{DGP2}.
More importantly, when the number of in-sample observations increases from $n=500$ to $n=1000$, the relative generalized expected utility of UMPR increases for all designs.
This phenomenon is guaranteed by Theorem~\ref{ThmMD} and Corollary~\ref{Corrollary3} because the approximation error $S^{*}-S^{*}_{3}$ is equal to zero.
In contrast, for~\ref{DGP2} in which any ML is misspecified, a larger sample size does not improve the relative generalized expected utility of AIC, BIC, LASSO, and SVM.
These results demonstrate that the UMPR inherits the robustness of MU estimation, a feature that selection methods based on penalized likelihood criteria do not possess in general.

Finally, we compare the proposed UMPR with two pretest estimators, including a specific-to-general approach and a general-to-specific approach, and the tenfold cross-validatory estimator in maximum utility estimation.
Table~\ref{Table3} provides the relative generalized expected utility and the percentage of models selected out of 500 simulation replications when $n=500$ and $n=1000$.
Differences in the selection frequencies for the pretesting are shown across the left and right panels because \citeauthor{ElliottLieli2013}'s (\citeyear{ElliottLieli2013}) test statistic depends on the preference specifications.
In terms of selecting the correctly specified cubic MU, the UMPR with VC penalty, either inclusive or exclusive of a technical term, performs worst under all designs.
As can be seen, the cross-validatory estimator, in comparison with the UMPR, attains higher percentages of selecting the cubic MU under all designs.
Thus, the cross validation might be preferable if we attempt to select the model correctly specified in the maximum utility estimation.
However, if the goal is to capture the largest expected utility, we prefer the UMPR with data-dependent penalties to the pretest and cross-validatory estimators because the proposed penalty-based prediction rules perform better than the other estimators in terms of the relative generalized expected utility, as suggested by the experimental evidence.

\section{Conclusion}\label{conclusion}
The maximum utility estimation can be viewed as the binary classification with a decision-based utility function.
Despite its possible improvement in utility over traditional maximum likelihood methods, the maximum utility estimation has inherited the in-sample overfitting from the perceptron learning.

To alleviate the in-sample overfitting, we adopt the structural risk minimization approach to construct a utility-maximizing prediction rule.
For complexity penalization, we consider the distribution-free VC penalty and four data-dependent penalties (MD, SMD, RC, and BC).
For each penalty, we show that the difference between the maximal expected utility and the generalized expected utility of the utility-maximizing prediction rule is bounded.
The upper bounds are close to zero for a large sample if the approximation error is equal to zero for some pre-specified classes of functions.
In general, we prefer the simulated complexity penalties in terms of predictive performance, as suggested by the simulation results.
These simulation results also show that the utility-maximizing prediction rule with an appropriate data-dependent complexity penalty has better predictive performance than the pretest and cross-validatory estimators;
more importantly, it outperforms the AIC, BIC, and LASSO if the conditional probability of the binary outcome is misspecified, and the $\ell_{1}$-norm SVM if the cutoff function considerably deviates from $1/2$.
The utility-maximizing prediction rule is thus important for the decision-making based on binary prediction.

\begin{appendices}
\numberwithin{equation}{section}

\section{Selection of $\alpha_{0}$}\label{SelectAlpha}
Given a pre-specified finite set $\mathscr{A}$ in which every element $\alpha$ satisfies $\zeta(\alpha)<\infty$, we can select $\alpha_{0}=\hat{\alpha}\in\mathscr{A}$ by the $T$-fold cross-validation method.\footnote{
Although the optimal choice of $T$ is an open theoretical question, common choices of $T$ are 5 or 10, as suggested by \citet{HastieTibshiraniEtAl2009}.}
We randomly partition the data $\mathscr{D}_{n}$ into $T$ roughly equal-sized sets.
Let $\tau:\{1,2,\ldots,n\}\to \{1,2,\ldots,T\}$ be the indexing function such that the observation $(Y_{i},X_{i})$ is in the validation set $\tau(i)$.
We write $\mathscr{D}^{(-t)}_{n}$ for the data $\mathscr{D}_{n}$ from which the validation set $t$ is removed,
and $n_{t}$ for the sample size of $\mathscr{D}^{(-t)}_{n}$.
Let $S^{(-t)}_{n}(f)$ be the empirical utility of $f$ calculated based on $\mathscr{D}^{(-t)}_{n}$.
The $T$-fold cross-validation method in our framework is implemented as follows.
\begin{enumerate}[label=(\arabic*)]
\item For each $t\in\{1,2,\ldots,T\}$ and $\alpha\in\mathscr{A}$, we calculate the UMPR with tuning parameter $\alpha$ based on $\mathscr{D}^{(-t)}_{n}$ by $\tilde{f}^{(-t)}_{n}(\alpha)=\hat{f}^{(-t)}_{\hat{k}^{(-t)}_{n}(\alpha)}$,
where
\begin{align*}
\hat{k}^{(-t)}_{n}(\alpha)\in\arg\max_{k\in\mathbb{N}}
\left(S^{(-t)}_{n}(f)-C_{n_{t}}^{\text{VC}}(k;\alpha)\right)
\;\;\text{and}\;\;\hat{f}^{(-t)}_{k}\in\arg\max_{f\in\mathcal{F}_{k}}S^{(-t)}_{n}(f).
\end{align*}
\item
The cross-validated tuning parameter is defined as
\begin{align*}
\hat{\alpha}=\arg\max_{\alpha\in\mathscr{A}}CV(\alpha),
\end{align*}
where
\begin{align*}
CV(\alpha)=\frac{1}{T}\sum_{t=1}^{T}
\frac{\sum_{i:\tau(i)=t}s(Y_{i},X_{i},\tilde{f}^{(-t)}_{n}(\alpha))}
{\sum_{i=1}^{n}\Ind{[\tau(i)=t]}}.
\end{align*}
\item We calculate the UMPR $\tilde{f}_{n}(\hat{\alpha})$ with cross-validated $\hat{\alpha}$ based on the whole data $\mathscr{D}_{n}$.
\end{enumerate}
Similarly, we can show that for any $n\in\mathbb{N}$ and $\varepsilon>0$,
\begin{align*}
\mathbb{P}\left(\tilde{S}_{n}(\tilde{f}_{n}(\hat{\alpha}))
-S(\tilde{f}_{n}(\hat{\alpha}))>\varepsilon\right)
\leq\sum_{\alpha\in\mathscr{A}}\zeta(\alpha)\exp{\left\{-\frac{n\varepsilon^{2}}{32M^{2}}\right\}},
\end{align*}
and
\begin{align*}
S^{*}-\Exp[S(\tilde{f}_{n}(\hat{\alpha}))]
\leq \min_{k}\left\{\Exp[C^{\text{VC}}_{n}(k;\hat{\alpha})]+\left(S^{*}-S_{k}^{*}\right)\right\}
+8M\sqrt{\frac{1+\log\{\sum_{\alpha\in\mathscr{A}}\zeta(\alpha)\}}{2n}}.
\end{align*}

\section{Pretesting and Cross-Validation}\label{CrossValidation}
\subsection{Pretesting}
We consider the null hypothesis $H^{(k)}_{0}: S^{*}_{(k-1)}=S^{*}_{k}$ against the alternative hypothesis $H^{(k)}_{1}: S^{*}_{(k-1)}<S^{*}_{k}$ for $k\in\{2,3\}$.
\citet{ElliottLieli2013} propose a general-to-specific pretest estimator based on the test statistic developed in their Proposition 4.
They suggest selecting the model $\mathcal{F}_{\hat{k}(\text{G}\to\text{S})}$, where
\begin{align*}
\hat{k}(\text{G}\to\text{S})=
\begin{cases}
1, & \hspace{-4.45cm}\text{if}\;\; \text{neither}\;H^{(3)}_{0}\; \text{nor}\; H^{(2)}_{0}\; \text{is rejected}, \\
\max\Big\{k\in\{2,3\}:H^{(k)}_{0}\;\text{is rejected against}\; H^{(k)}_{1}\Big\}, & \text{otherwise}.
\end{cases}
\end{align*}
Similarly, we can apply a specific-to-general approach to selecting the model $\mathcal{F}_{\hat{k}(\text{S}\to\text{G})}$, where
\begin{align*}
\hat{k}(\text{S}\to\text{G})=
\begin{cases}
3, & \hspace{-4.43cm}\text{if both}\;\; H^{(3)}_{0}\; \text{and}\; H^{(2)}_{0}\; \text{are rejected}, \\
\min\Big\{k\in\{2,3\}:H^{(k)}_{0}\;\text{is not rejected against}\; H^{(k)}_{1}\Big\}-1, & \text{otherwise}.
\end{cases}
\end{align*}
For these two approaches, we conduct \citeauthor{ElliottLieli2013}'s test statistic with the size equal to 5\% and auxiliary i.i.d.\ random variables that follow a Bernoulli(0.75) distribution.

\subsection{Cross-Validation}
We randomly partition the data $\mathscr{D}_{n}$ into $T$ roughly equal-sized sets.
Let $\tau:\{1,2,\ldots,n\}\to \{1,2,\ldots,T\}$ be the indexing function such that the observation $(Y_{i},X_{i})$ is in the validation set $\tau(i)$.
We write $\mathscr{D}^{(-t)}_{n}$ for the data $\mathscr{D}_{n}$ from which the validation set $t$ is removed.
The $T$-fold cross-validation method in the maximum utility framework can be implemented as follows.
\begin{enumerate}[label=(\arabic*)]
\item We consider an integer $K$.
For each $k\in\{1,2,\ldots,K\}$ and $t\in\{1,2,\ldots,T\}$, we calculate the MU estimator based on $\mathscr{D}^{(-t)}_{n}$ by
\begin{align*}
\hat{f}^{(-t)}_{k}\in\arg\max_{f\in\mathcal{F}_{k}}S^{(-t)}_{n}(f),
\end{align*}
where $S^{(-t)}_{n}(f)$ is the empirical utility calculated by $f$ and $\mathscr{D}^{(-t)}_{n}$; that is,
\begin{align*}
S^{(-t)}_{n}(f)=\frac{\sum_{i:\tau(i)\neq t}s(Y_{i},X_{i},f)}{\sum_{i=1}^{n}\Ind{[\tau(i)\neq t]}}.
\end{align*}
\item
The cross-validated value of $k$ is defined as
\begin{align*}
\hat{k}_{n}=\arg\max_{k\in\{1,2,\ldots,K\}}CV(k),
\end{align*}
where
\begin{align*}
CV(k)=\frac{1}{T}\sum_{t=1}^{T}
\frac{\sum_{i:\tau(i)=t}s(Y_{i},X_{i},\hat{f}^{(-t)}_{k})}{\sum_{i=1}^{n}\Ind{[\tau(i)=t]}}.
\end{align*}
\item The cross-validated MU estimator is the MU estimator selected from $\mathcal{F}_{\hat{k}_{n}}$ based on $\mathscr{D}_{n}$; specifically,
\begin{align*}
\hat{f}^{\text{CV}}_{\hat{k}_{n}}\in\arg\max_{f\in\mathcal{F}_{\hat{k}_{n}}}S_{n}(f).
\end{align*}
\end{enumerate}

\section{Technical Proofs}\label{Appendix}
\subsection{Proof of Proposition~\ref{Prop2}}
\begin{proof}
Since the mapping
\begin{align*}
\mathscr{D}_{n}\mapsto \sup_{f\in\mathcal{F}_{k}}\left(S_{n}(f)-S(f)\right)
\end{align*}
satisfies the bounded differences property in Section 6.1 of \citet{BoucheronLugosiEtAl2013} with their notation $c_{i}=8M/n$ for each $i\in\{1,\ldots,n\}$,
we have
\begin{align*}
\mathbb{P}\left(\sup_{f\in\mathcal{F}_{k}}(S_{n}(f)-S(f))
-\Exp\left[\sup_{f\in\mathcal{F}_{k}}(S_{n}(f)-S(f))\right]>\varepsilon\right)
\leq \exp{\left\{-\frac{n\varepsilon^{2}}{32M^{2}}\right\}}
\end{align*}
by \citeauthor{McDiarmid1989}'s (\citeyear{McDiarmid1989}) inequality.

It is now sufficient to show that
\begin{align*}
\Exp\left[\sup_{f\in\mathcal{F}_{k}}(S_{n}(f)-S(f))\right]\leq 8M\sqrt{\frac{2\log\{\Pi_{k,c}(n)\}}{n}}.
\end{align*}
Let $\mathscr{D}'_{n}$ be the ghost sample, in which the observations $(Y'_{1},X'_{1}),\ldots,(Y'_{n},X'_{n})$ are distributed as $(Y_{1},X_{1}),\ldots,(Y_{n},X_{n})$ and independent of them.
We write $S'_{n}(f)$ for the empirical utility of the prediction rule $f$ based on the ghost sample $\mathscr{D}'_{n}$.
Let $\{\sigma_{i}\}_{i=1}^{n}$ be a sequence of i.i.d.\ Rademacher random variables that are independent of $\mathscr{D}_{n}$; that is, $\mathbb{P}(\sigma_{i}=1)=\mathbb{P}(\sigma_{i}=-1)=1/2$.
The common symmetrization argument shows that
\begin{align}\label{C1}
\Exp\left[\sup_{f\in\mathcal{F}_{k}}\left(S_{n}(f)-S(f)\right)\right]
\leq& \Exp\left[\max_{f\in\mathcal{F}_{k}}\left(S_{n}(f)-S'_{n}(f)\right)\right]\notag\\
=& \Exp\left[\max_{f\in\mathcal{F}_{k}}
\frac{1}{n}\sum_{i=1}^{n}\sigma_{i}\Big(s(Y_{i},X_{i},f)-s(Y'_{i},X'_{i},f)\Big)\right]\notag\\
\leq& 2\Exp\left[\max_{f\in\mathcal{F}_{k}} \frac{1}{n}\sum_{i=1}^{n}\sigma_{i}s(Y_{i},X_{i},f)\right].
\end{align}
Let $\psi_{i}\equiv b(X_{i})[Y_{i}+1-2c(X_{i})]$ for each $i\in\{1,\ldots,n\}$.
Lemma 26.9 of \citet{Shalev-ShwartzBen-David2014} implies that
\begin{align*}
\Exp\left[\max_{f\in\mathcal{F}_{k}} \frac{1}{n}\sum_{i=1}^{n}\sigma_{i}s(Y_{i},X_{i},f)\Big|\mathscr{D}_{n}\right]
=&\Exp\left[\max_{f\in\mathcal{F}_{k}} \frac{1}{n}\sum_{i=1}^{n}\sigma_{i}\psi_{i}\text{sign}(f(X_{i})-c(X_{i}))\Big|\mathscr{D}_{n}\right]\\
\leq& 4M\Exp\left[\max_{f\in\mathcal{F}_{k}} \frac{1}{n}\sum_{i=1}^{n}\sigma_{i}\text{sign}(f(X_{i})-c(X_{i}))\Big|\mathscr{D}_{n}\right].
\end{align*}
Applying Lemma 5.2 of \citet{Massart2000} yields
\begin{align*}
\Exp\left[\max_{f\in\mathcal{F}_{k}} \frac{1}{n}\sum_{i=1}^{n}\sigma_{i}\text{sign}(f(X_{i})-c(X_{i}))\Big|\mathscr{D}_{n}\right]
\leq \sqrt{\frac{2\log\{\Pi_{k,c}(n)\}}{n}}.
\end{align*}
It follows that
\begin{align}\label{C2}
\Exp\left[\max_{f\in\mathcal{F}_{k}} \frac{1}{n}\sum_{i=1}^{n}\sigma_{i}s(Y_{i},X_{i},f)\right]\leq 4M\sqrt{\frac{2\log\{\Pi_{k,c}(n)\}}{n}}.
\end{align}
Combining Inequalities~(\ref{C1}) and~(\ref{C2}) completes the proof.
\end{proof}

\subsection{Proof of Corollary~\ref{Corollary1}}
\begin{proof}
Since the growth function $\Pi_{k,c}$ is of polynomial order for each $k\in\mathbb{N}$,
there exists an integer $n^{*}$ such that
\begin{align*}
8M\sqrt{\frac{2\log{\{\Pi_{k,c}(n)\}}}{n}}<\frac{\varepsilon}{2}.
\end{align*}
for all $n\geq n^{*}$.
Following the argument in Proposition~\ref{Prop2} \textit{mutatis mutandis}, we have
\begin{align*}
\mathbb{P}\left(\sup_{f\in\mathcal{F}_{k}}(S(f)-S_{n}(f))>8M\sqrt{\frac{2\log\{\Pi_{k,c}(n)\}}{n}}
+\frac{\varepsilon}{2}\right)
\leq \exp{\left\{-\frac{n\varepsilon^{2}}{128M^{2}}\right\}}.
\end{align*}
Hence, for all $n\geq n^{*}$,
\begin{align*}
&\mathbb{P}\left(\sup_{f\in\mathcal{F}_{k}}|S_{n}(f)-S(f)|>\varepsilon\right)\\
\leq &
\mathbb{P}\left(\sup_{f\in\mathcal{F}_{k}}(S_{n}(f)-S(f))
>8M\sqrt{\frac{2\log\{\Pi_{k,c}(n)\}}{n}}+\frac{\varepsilon}{2}\right)\\
&+\mathbb{P}\left(\sup_{f\in\mathcal{F}_{k}}(S(f)-S_{n}(f))
>8M\sqrt{\frac{2\log\{\Pi_{k,c}(n)\}}{n}}+\frac{\varepsilon}{2}\right)\\
\leq& 2\exp{\left\{-\frac{n\varepsilon^{2}}{128M^{2}}\right\}}.
\end{align*}

\end{proof}

\subsection{Proof of Theorem~\ref{MainThm}}
\begin{proof}
\emph{Part 1}\\
By construction, we have
\begin{align*}
\tilde{S}_{n}(\hat{f}_{j};j)=R^{\text{VC}}_{n,j}-8M\sqrt{\frac{(1+\alpha_{0})\log{\{V_{j,c}\}}}{2n}}
\end{align*}
for each $j\in\mathbb{N}$.
So, for any $n\in\mathbb{N}$ and $\varepsilon>0$,
\begin{align*}
&\mathbb{P}\left(\tilde{S}_{n}(\tilde{f}_{n})-S(\tilde{f}_{n})>\varepsilon\right)\\
\leq &\mathbb{P}\left(\sup_{j\in\mathbb{N}}
\left(\tilde{S}_{n}(\hat{f}_{j};j)-S(\hat{f}_{j})\right)>\varepsilon\right)\\
\leq &\sum_{j=1}^{\infty}\mathbb{P}\left(\tilde{S}_{n}(\hat{f}_{j};j)-S(\hat{f}_{j})>\varepsilon\right)\\
\leq &\sum_{j=1}^{\infty}\mathbb{P}\left(R^{\text{VC}}_{n,j}-S(\hat{f}_{j})>
\varepsilon+8M\sqrt{\frac{(1+\alpha_{0})\log{\{V_{j,c}\}}}{2n}}\right).
\end{align*}
It follows by Inequality~(\ref{boundofrisk}) that
\begin{align*}
\mathbb{P}\left(\tilde{S}_{n}(\tilde{f}_{n})-S(\tilde{f}_{n})>\varepsilon\right)
&\leq \sum_{j=1}^{\infty}\exp{\left\{-\frac{n}{32M^{2}}
\left[\varepsilon+8M\sqrt{\frac{(1+\alpha_{0})\log{\{V_{j,c}\}}}{2n}}\right]^{2}\right\}}\\
&\leq \sum_{j=1}^{\infty}\exp{\left\{-\frac{n}{32M^{2}}
\left[\varepsilon^{2}+32M^{2}\frac{(1+\alpha_{0})\log{\{V_{j,c}\}}}{n}\right]\right\}}\\
&= \zeta(\alpha_{0})\exp{\left\{-\frac{n\varepsilon^{2}}{32M^{2}}\right\}}.
\end{align*}
\emph{Part 2}\\
By \textit{Part 1} and Lemma~\ref{expectationbound}, we have
\begin{align*}
\Exp\left[\tilde{S}_{n}(\tilde{f}_{n})-S(\tilde{f}_{n})\right]
\leq 8M\sqrt{\frac{1+\log\{\zeta(\alpha_{0})\}}{2n}}.
\end{align*}
In addition, for each $k\in\mathbb{N}$,
\begin{align*}
S_{k}^{*}-\Exp[\tilde{S}_{n}(\tilde{f}_{n})]
&\leq S_{k}^{*}-\Exp[\tilde{S}_{n}(\hat{f}_{k};k)]\\
&= S_{k}^{*}-\Exp[S_{n}(\hat{f}_{k})]+C^{\text{VC}}_{n}(k;\alpha_{0})\\
&\leq S_{k}^{*}-\Exp[S_{n}(f_{k}^{*})]+C^{\text{VC}}_{n}(k;\alpha_{0})\\
&= C^{\text{VC}}_{n}(k;\alpha_{0})
\end{align*}
because $\Exp[S_{n}(f_{k}^{*})]=S_{k}^{*}$ and $f_{k}^{*}\in\arg\max_{f\in\mathcal{F}_{k}}S(f)$.
So, for each $k\in\mathbb{N}$,
\begin{align*}
S^{*}-\Exp[S(\tilde{f}_{n})]
=&S^{*}-S_{k}^{*}+
S_{k}^{*}-\Exp\left[\tilde{S}_{n}(\tilde{f}_{n})\right]+
\Exp\left[\tilde{S}_{n}(\tilde{f}_{n})\right]-\Exp\left[S(\tilde{f}_{n})\right]\\
\leq&S^{*}-S_{k}^{*}+C^{\text{VC}}_{n}(k;\alpha_{0})+8M\sqrt{\frac{1+\log\{\zeta(\alpha_{0})\}}{2n}}.
\end{align*}
It follows that
\begin{align*}
S^{*}-\Exp\left[S(\tilde{f}_{n})\right]
\leq \min_{k}\left\{C^{\text{VC}}_{n}(k;\alpha_{0})+(S^{*}-S_{k}^{*})\right\}
+8M\sqrt{\frac{1+\log\{\zeta(\alpha_{0})\}}{2n}}.
\end{align*}
\end{proof}

\subsection{Proof of Corollary~\ref{Corollary2}}
\begin{proof}
Fix an $\varepsilon>0$.
Choose an integer $k_{0}=k_{0}(\varepsilon)$ such that
$S^{*}-\sup_{f\in\mathcal{F}_{k}}S(f)<\varepsilon/3$ for all $k\geq k_{0}$.
In addition, choose an integer $n_{0}=n_{0}(\varepsilon,k_{0})$ such that $C^{\text{VC}}_{n}(k_{0};\alpha_{0})<\varepsilon/6$ for all $n\geq n_{0}$.
For each $n\in\mathbb{N}$,
\begin{align*}
&\mathbb{P}\left(\left|S^{*}-S(\tilde{f}_{n})\right|>\varepsilon\right)\\
\leq &\mathbb{P}\left(\sup_{f\in\mathcal{F}_{k_{0}}}S(f)-\tilde{S}_{n}(\tilde{f}_{n})>\varepsilon/3\right)
+\mathbb{P}\left(\tilde{S}_{n}(\tilde{f}_{n})-S(\tilde{f}_{n})>\varepsilon/3\right)\\
\leq &\mathbb{P}\left(\sup_{f\in\mathcal{F}_{k_{0}}}S(f)-\tilde{S}_{n}(\hat{f}_{k_{0}};k_{0})
>\varepsilon/3\right)
+\mathbb{P}\left(\tilde{S}_{n}(\tilde{f}_{n})-S(\tilde{f}_{n})>\varepsilon/3\right).
\end{align*}
It follows from Theorem~\ref{MainThm} that the second term in the right hand side is bounded above by $A_{1}\exp{\left\{-B_{1}n\varepsilon^{2}\right\}}$ for some positive constants $A_{1}$ and $B_{1}$.
By Corollary~\ref{Corollary1}, there exist an integer $n_{1}=n_{1}(\varepsilon,k_{0})$ and two constants $A_{2}$ and $B_{2}$ such that
\begin{align*}
\mathbb{P}\left(\sup_{f\in\mathcal{F}_{k_{0}}}|S(f)-S_{n}(f)|>\varepsilon/6\right)
\leq A_{2}\exp{\left\{-B_{2}n\varepsilon^{2}\right\}}
\end{align*}
for all $n\geq n_{1}$.
Hence, for any $n\geq n_{2}\equiv\max\{n_{0},n_{1}\}$, we have
\begin{align*}
\mathbb{P}\left(\sup_{f\in\mathcal{F}_{k_{0}}}S(f)-\tilde{S}_{n}(\hat{f}_{k_{0}};k_{0})
>\varepsilon/3\right)
\leq &\mathbb{P}\left(\sup_{f\in\mathcal{F}_{k_{0}}}S(f)-S_{n}(\hat{f}_{k_{0}})>\varepsilon/6\right)\\
\leq &\mathbb{P}\left(\sup_{f\in\mathcal{F}_{k_{0}}}|S(f)-S_{n}(f)|>\varepsilon/6\right)\\
\leq &A_{2}\exp{\left\{-B_{2}n\varepsilon^{2}\right\}}.
\end{align*}
Therefore, we have
\begin{align*}
&\sum_{n=1}^{\infty}\mathbb{P}\left(\left|S^{*}-S(\tilde{f}_{n})\right|>\varepsilon\right)\\
\leq & (n_{2}-1)+\sum_{n=n_{2}}^{\infty}\Bigg[
\mathbb{P}\left(\sup_{f\in\mathcal{F}_{k_{0}}}S(f)-\tilde{S}_{n}(\hat{f}_{k_{0}};k_{0})
>\varepsilon/3\right)
+\mathbb{P}\left(\tilde{S}_{n}(\tilde{f}_{n})-S(\tilde{f}_{n})>\varepsilon/3\right)\Bigg]\\
\leq & (n_{2}-1)+\sum_{n=n_{2}}^{\infty}\left(A_{1}\exp{\left\{-B_{1}n\varepsilon^{2}\right\}}
+A_{2}\exp{\left\{-B_{2}n\varepsilon^{2}\right\}}\right)\\
<&\infty.
\end{align*}
Applying the Borel-Cantelli lemma yields the statement.
\end{proof}

\subsection{Proof of Proposition~\ref{PropX1}}
\begin{proof}
We first study $S^{*}-S(f)$.
Let $b_{1}(x)=u_{1,1}(x)-u_{-1,1}(x)$ and $b_{-1}(x)=u_{-1,-1}(x)-u_{1,-1}(x)$ for all $x\in\mathcal{X}$.
For any $f$, we have that with probability one,
\begin{align}\label{C3}
&\Exp\left[b(X)[Y(1-2c(X))+1]\Ind{[Y\neq \text{sign}(f(X)-c(X))]}|X\right]\notag\\
=&2p^{*}(X)b_{1}(X)\Ind{[1\neq \text{sign}(f(X)-c(X))]}
+2(1-p^{*}(X))b_{-1}(X)\Ind{[-1\neq \text{sign}(f(X)-c(X))]}\notag\\
=&2p^{*}(X)b_{1}(X)\Ind{[f(X)<c(X)]}+2(1-p^{*}(X))b_{-1}(X)\Ind{[f(X)\geq c(X)]}.
\end{align}
It follows from (\ref{C3}) that with probability one,
\begin{align}\label{C4}
&\Exp\left[b(X)[Y(1-2c(X))+1](\Ind{[Y\neq \text{sign}(f(X)-c(X))]}-\Ind{[Y\neq \text{sign}(p^{*}(X)-c(X))]})|X\right]\notag\\
=&2p^{*}(X)b_{1}(X)(\Ind{[f(X)<c(X)]}-\Ind{[p^{*}(X)<c(X)]})\notag\\
&+2(1-p^{*}(X))b_{-1}(X)(\Ind{[f(X)\geq c(X)]}-\Ind{[p^{*}(X)\geq c(X)]})\notag\\
=&2[p^{*}(X)b_{1}(X)-(1-p^{*}(X))b_{-1}(X)](\Ind{[p^{*}(X)\geq c(X)]}-\Ind{[f(X)\geq c(X)]})\notag\\
=&2b(X)[p^{*}(X)-c(X)](\Ind{[p^{*}(X)\geq c(X)]}-\Ind{[f(X)\geq c(X)]}).
\end{align}
Note that for any $f$, we have
\begin{align}\label{C5}
S(f)
=& \Exp\left[ b(X)[Y(1-2c(X))+1]Y\text{sign}(f(X)-c(X))\right]\notag\\
=& \Exp\left[b(X)[Y(1-2c(X))+1]\right]
-2\Exp\left[b(X)[Y(1-2c(X))+1]\Ind{[Y\neq \text{sign}(f(X)-c(X))]}\right].
\end{align}
Combining~(\ref{C4}) and~(\ref{C5}) yields
\begin{align*}
S(p^{*})-S(f)
=4\Exp\left[b(X)[p^{*}(X)-c(X)](\Ind{[p^{*}(X)\geq c(X)]}-\Ind{[f(X)\geq c(X)]})\right]\geq 0
\end{align*}
for any $f$.
It immediately follows that $S^{*}=S(p^{*})$.

Next, we calculate the maximal expected utility $S^{*}$.
The derivation in~(\ref{C3}) implies that with probability one,
\begin{align*}
&\Exp\left[b(X)[Y(1-2c(X))+1]\Ind{[Y\neq \text{sign}(p^{*}(X)-c(X))]}|X\right]\\
=&2p^{*}(X)b_{1}(X)\Ind{[p^{*}(X)<c(X)]}+2(1-p^{*}(X))b_{-1}(X)\Ind{[p^{*}(X)\geq c(X)]}\\
=&2[p^{*}(X)b_{1}(X)-(1-p^{*}(X))b_{-1}(X)][\Ind{[p^{*}(X)<c(X)]}-\Ind{[p^{*}(X)\geq c(X)]}]\\
&+2p^{*}(X)b_{1}(X)\Ind{[p^{*}(X)\geq c(X)]}+2(1-p^{*}(X))b_{-1}(X)\Ind{[p^{*}(X)< c(X)]}\\
=&-2b(X)|p^{*}(X)-c(X)|+\Exp\left[b(X)[Y(1-2c(X))+1]\Ind{[Y= \text{sign}(p^{*}(X)-c(X))]}|X\right].
\end{align*}
After rearrangement, we obtain
\begin{align*}
2b(X)|p^{*}(X)-c(X)|
=&\Exp\left[b(X)[Y(1-2c(X))+1]|X\right]\\
&-2\Exp\left[b(X)[Y(1-2c(X))+1]\Ind{[Y\neq \text{sign}(p^{*}(X)-c(X))]}|X\right].
\end{align*}
Taking expectation on both sides yields
$S^{*}=S(p^{*})=2\Exp\left[b(X)|p^{*}(X)-c(X)|\right]$ by (\ref{C5}).

Finally, note that $|p^{*}(X)-c(X)|\leq |p^{*}(X)-f(X)|$ whenever $\Ind{[p^{*}(X)\geq c(X)]}\neq\Ind{[f(X)\geq c(X)]}$.
It follows that
\begin{align*}
S^{*}-S(f)=&4\Exp\left[b(X)[p^{*}(X)-c(X)](\Ind{[p^{*}(X)\geq c(X)]}-\Ind{[f(X)\geq c(X)]})\right]\\
=&4\Exp\left[b(X)|p^{*}(X)-c(X)||\Ind{[p^{*}(X)\geq c(X)]}-\Ind{[f(X)\geq c(X)]}|\right]\\
\leq&4\Exp\left[b(X)|p^{*}(X)-f(X)||\Ind{[p^{*}(X)\geq c(X)]}-\Ind{[f(X)\geq c(X)]}|\right]\\
\leq&4\Exp\left[b(X)|p^{*}(X)-f(X)|\right].
\end{align*}
Hence, we obtain
$S^{*}-S(f)\leq 16M\sup_{x\in\mathcal{X}}|p^{*}(x)-f(x)|$
by Assumption~\ref{A3}.
\end{proof}

\subsection{Proof of Theorem~\ref{ThmMD}}
\begin{proof}
\emph{Part 1}\\
We write $S'_{n}(f)$ for the empirical utility of the prediction rule $f$ based on the ghost sample $\mathscr{D}'_{n}$, in which the observations $(Y'_{1},X'_{1}),\ldots,(Y'_{n},X'_{n})$ are distributed as $(Y_{1},X_{1}),\ldots,(Y_{n},X_{n})$ and independent of them.
For ease of notation, let
\begin{align*}
S^{(1)}_{n}(f)\equiv \frac{2}{n}\sum_{i=1}^{n/2}s(Y_{i},X_{i},f)\;\;\text{and}\;\;
S^{(2)}_{n}(f)\equiv \frac{2}{n}\sum_{i=n/2+1}^{n}s(Y_{i},X_{i},f).
\end{align*}

As in the proof of Theorem~\ref{MainThm}, the desired results follow from the exponential tail inequality
\begin{align*}
\mathbb{P}\left(R^{\text{MD}}_{n,k}-S(\hat{f}_{k})>\varepsilon\right)
\leq \exp{\left\{-\frac{n\varepsilon^{2}}{288M^{2}}\right\}},
\end{align*}
where $R^{\text{MD}}_{n,k}\equiv S_{n}(\hat{f}_{k})-\max_{f\in\mathcal{F}_{k}}(S^{(1)}_{n}(f)-S^{(2)}_{n}(f))$.
To establish this tail inequality, we note that if
\begin{align}\label{C6}
\Exp\left[\sup_{f\in\mathcal{F}_{k}}\left(S_{n}(f)-S(f)\right)  \right]
\leq\Exp\left[\max_{f\in\mathcal{F}_{k}}\left(S^{(1)}_{n}(f)-S^{(2)}_{n}(f)\right)  \right],
\end{align}
then
\begin{align*}
\mathbb{P}\left(R^{\text{MD}}_{n,k}-S(\hat{f}_{k})>\varepsilon\right)
\leq& \mathbb{P}\Bigg(\sup_{f\in\mathcal{F}_{k}}\left(S_{n}(f)-S(f)\right)
-\max_{f\in\mathcal{F}_{k}}\left(S^{(1)}_{n}(f)-S^{(2)}_{n}(f)\right)\\
&\hspace{0.3cm}-\Exp\left[
\sup_{f\in\mathcal{F}_{k}}\left(S_{n}(f)-S(f)\right)-\max_{f\in\mathcal{F}_{k}}
\left(S^{(1)}_{n}(f)-S^{(2)}_{n}(f)\right)
\right]
>\varepsilon\Bigg).
\end{align*}
It follows from \citeauthor{McDiarmid1989}'s (\citeyear{McDiarmid1989}) inequality that the latter probability is bounded above by $\exp{\left\{-2n\varepsilon^{2}/(24M)^{2}\right\}}=\exp{\left\{-n\varepsilon^{2}/288M^{2}\right\}}$ because the mapping
\begin{align*}
\mathscr{D}_{n}\mapsto \sup_{f\in\mathcal{F}_{k}}\left(S_{n}(f)-S(f)\right)
-\max_{f\in\mathcal{F}_{k}}\left(S^{(1)}_{n}(f)-S^{(2)}_{n}(f)\right)
\end{align*}
satisfies the bounded differences property in Section 6.1 of \citet{BoucheronLugosiEtAl2013} with their notation $c_{i}=24M/n$ for each $i\in\{1,\ldots,n\}$.

It remains to prove Inequality~(\ref{C6}).
Since $\{(Y'_{i},X'_{i})\}_{i=1}^{n/2}$, $\{(Y'_{i},X'_{i})\}_{i=n/2+1}^{n}$, $\{(Y_{i},X_{i})\}_{i=1}^{n/2}$, and $\{(Y_{i},X_{i})\}_{i=n/2+1}^{n}$ are independent and identically distributed,
the common symmetrization argument shows that
\begin{align*}
\Exp\left[\sup_{f\in\mathcal{F}_{k}}\left(S_{n}(f)-S(f)\right)  \right]
\leq& \Exp\left[\max_{f\in\mathcal{F}_{k}}\left(S_{n}(f)-S'_{n}(f)\right)  \right]\\
\leq& \Exp\left[\max_{f\in\mathcal{F}_{k}}
\frac{2}{n}\sum_{i=1}^{n/2}\Big(s(Y_{i},X_{i},f)-s(Y'_{i},X'_{i},f)\Big)  \right]\\
=& \Exp\left[\max_{f\in\mathcal{F}_{k}}\left(S^{(1)}_{n}(f)-S^{(2)}_{n}(f)\right)  \right].
\end{align*}

\noindent
\emph{Part 2}\\
For ease of notation, let $\sigma\equiv\{\sigma^{(j)}\}_{j=1}^{m}= \{(\sigma^{(j)}_{1},\sigma^{(j)}_{2},\ldots,\sigma^{(j)}_{n/2})\}_{j=1}^{m}$,
\begin{align*}
Q^{\text{SMD}}_{n,k}(\sigma;\mathscr{D}_{n})\equiv \frac{1}{m}\sum_{j=1}^{m}\left(\max_{f\in\mathcal{F}_{k}}\frac{2}{n}\sum_{i=1}^{n/2}\sigma^{(j)}_{i}
\Big(s(Y_{2i-1},X_{2i-1},f)-s(Y_{2i},X_{2i},f)\Big)\right),
\end{align*}
and $R^{\text{SMD}}_{n,k}\equiv S_{n}(\hat{f}_{k})-Q^{\text{SMD}}_{n,k}(\sigma;\mathscr{D}_{n})$.
For every $m,n\in\mathbb{N}$, let
\begin{align*}
\eta_{m,n}\equiv
\begin{cases}
3/5, & \text{if}\;\; n\leq m, \\
3/(2\ell+5), & \text{if}\;\; n/(\ell+1)^{2} \leq m < n/\ell^{2}\;\;
\text{and}\;\; \ell\in\mathbb{N}.
\end{cases}
\end{align*}

The proof is similar to that of \emph{Part 1} in the sense that our goal is to establish an appropriate exponential tail inequality of $R^{\text{SMD}}_{n,k}-S(\hat{f}_{k})$.
The additional trick is to deal with the randomness arising from simulated Rademacher random vectors.
To disentangle such randomness from the randomness of data $\mathscr{D}_{n}$, we consider the inequality
\begin{align}
\mathbb{P}\left(R^{\text{SMD}}_{n,k}-S(\hat{f}_{k})>\varepsilon\right)
\leq& \mathbb{P}\left(\sup_{f\in\mathcal{F}_{k}}\left(S_{n}(f)-S(f)\right)
-\Exp\left[Q^{\text{SMD}}_{n,k}(\sigma;\mathscr{D}_{n})|\mathscr{D}_{n}\right]
>\eta\varepsilon\right)\notag\\
&+\mathbb{P}\left(\Exp\left[Q^{\text{SMD}}_{n,k}(\sigma;\mathscr{D}_{n})|\mathscr{D}_{n}\right]
-Q^{\text{SMD}}_{n,k}(\sigma;\mathscr{D}_{n}) >(1-\eta)\varepsilon\right) \label{C7}
\end{align}
for any $\eta\in(0,1)$.
Given $\mathscr{D}_{n}$, the mapping $\{\sigma^{(j)}\}_{j=1}^{m}\mapsto Q^{\text{SMD}}_{n,k}(\sigma;\mathscr{D}_{n})$ satisfies the bounded differences property in Section 6.1 of \citet{BoucheronLugosiEtAl2013} with their notation $c_{j}=16M/m$ for each $j\in\{1,\ldots,m\}$.
It follows from \citeauthor{McDiarmid1989}'s (\citeyear{McDiarmid1989}) inequality that
\begin{align*}
\mathbb{P}\left(\Exp\left[Q^{\text{SMD}}_{n,k}(\sigma;\mathscr{D}_{n})|\mathscr{D}_{n}\right]
-Q^{\text{SMD}}_{n,k}(\sigma;\mathscr{D}_{n}) >(1-\eta)\varepsilon|\mathscr{D}_{n}\right)
\leq \exp{\left\{-\frac{2m(1-\eta)^{2}\varepsilon^{2}}{(16M)^{2}}\right\}}
\end{align*}
with probability one.
Taking expectation with respect to $\mathscr{D}_{n}$ on both sides yields
\begin{align}\label{C8}
\mathbb{P}\left(\Exp\left[Q^{\text{SMD}}_{n,k}(\sigma;\mathscr{D}_{n})|\mathscr{D}_{n}\right]
-Q^{\text{SMD}}_{n,k}(\sigma;\mathscr{D}_{n}) >(1-\eta)\varepsilon\right)
\leq \exp{\left\{-\frac{2m(1-\eta)^{2}\varepsilon^{2}}{(16M)^{2}}\right\}}.
\end{align}
In addition, since the observations in $\mathscr{D}_{n}$ and $\mathscr{D}'_{n}$ are i.i.d.,
the common symmetrization argument shows that
\begin{align*}
\Exp\left[\sup_{f\in\mathcal{F}_{k}}\left(S_{n}(f)-S(f)\right)\right]
\leq\Exp\left[Q^{\text{SMD}}_{n,k}(\sigma;\mathscr{D}_{n})\right].
\end{align*}
We apply \citeauthor{McDiarmid1989}'s (\citeyear{McDiarmid1989}) inequality again and obtain
\begin{align}\label{C9}
& \mathbb{P}\left(\sup_{f\in\mathcal{F}_{k}}\left(S_{n}(f)-S(f)\right)
-\Exp\left[Q^{\text{SMD}}_{n,k}(\sigma;\mathscr{D}_{n})|\mathscr{D}_{n}\right]
>\eta\varepsilon\right)\notag\\
\leq& \mathbb{P}\Big(\sup_{f\in\mathcal{F}_{k}}\left(S_{n}(f)-S(f)\right)
-\Exp\left[Q^{\text{SMD}}_{n,k}(\sigma;\mathscr{D}_{n})|\mathscr{D}_{n}\right]\notag\\
&\hspace{0.3cm}-\Exp\left[\sup_{f\in\mathcal{F}_{k}}\left(S_{n}(f)-S(f)\right)\right]
+\Exp\left[Q^{\text{SMD}}_{n,k}(\sigma;\mathscr{D}_{n})\right]
>\eta\varepsilon\Big)\notag\\
\leq & \exp{\left\{-\frac{2n\eta^{2}\varepsilon^{2}}{(24M)^{2}}\right\}}.
\end{align}
because the mapping
\begin{align*}
\mathscr{D}_{n}\mapsto \sup_{f\in\mathcal{F}_{k}}\left(S_{n}(f)-S(f)\right)
-\Exp\left[Q^{\text{SMD}}_{n,k}(\sigma;\mathscr{D}_{n})|\mathscr{D}_{n}\right]
\end{align*}
satisfies the bounded differences property in Section 6.1 of \citet{BoucheronLugosiEtAl2013} with their notation $c_{i}=24M/n$ for each $i\in\{1,\ldots,n\}$.
Combining Inequalities~(\ref{C7})-(\ref{C9}) and setting $\eta=\eta_{m,n}$, we obtain
\begin{align*}
\mathbb{P}\left(R^{\text{SMD}}_{n,k}-S(\hat{f}_{k})>\varepsilon\right)
\leq &\exp{\left\{-\frac{2n\eta^{2}\varepsilon^{2}}{(24M)^{2}}\right\}}+
\exp{\left\{-\frac{2m(1-\eta)^{2}\varepsilon^{2}}{(16M)^{2}}\right\}}\\
\leq & 2\exp{\left\{-\frac{2n\varepsilon^{2}}{(\gamma_{m,n}(M))^{2}}\right\}}.
\end{align*}
The desired results follow from the exponential tail probability above and similar arguments used in Theorem~\ref{MainThm}.\\[0.3cm]
\noindent
\emph{Part 3}\\
Let $\sigma\equiv\{\sigma^{(j)}\}_{j=1}^{m}= \{(\sigma^{(j)}_{1},\sigma^{(j)}_{2},\ldots,\sigma^{(j)}_{n})\}_{j=1}^{m}$ and
$R^{\text{RC}}_{n,k}\equiv S_{n}(\hat{f}_{k})-Q^{\text{RC}}_{n,k}(\sigma;\mathscr{D}_{n})$, where
\begin{align*}
Q^{\text{RC}}_{n,k}(\sigma;\mathscr{D}_{n})\equiv \frac{1}{m}\sum_{j=1}^{m}\left(\max_{f\in\mathcal{F}_{k}}\frac{2}{n}\sum_{i=1}^{n}\sigma^{(j)}_{i}
s(Y_{i},X_{i},f)\right).
\end{align*}
Following the argument in \emph{Part 2} \textit{mutatis mutandis}, we have for any $\eta\in(0,1)$,
\begin{align*}
\mathbb{P}\left(\Exp\left[Q^{\text{RC}}_{n,k}(\sigma;\mathscr{D}_{n})|\mathscr{D}_{n}\right]
-Q^{\text{RC}}_{n,k}(\sigma;\mathscr{D}_{n}) >(1-\eta)\varepsilon\right)
\leq \exp{\left\{-\frac{2m(1-\eta)^{2}\varepsilon^{2}}{(16M)^{2}}\right\}},
\end{align*}
and
\begin{align*}
\mathbb{P}\left(\sup_{f\in\mathcal{F}_{k}}\left(S_{n}(f)-S(f)\right)
-\Exp\left[Q^{\text{RC}}_{n,k}(\sigma;\mathscr{D}_{n})|\mathscr{D}_{n}\right]>\eta\varepsilon\right)
\leq & \exp{\left\{-\frac{2n\eta^{2}\varepsilon^{2}}{(24M)^{2}}\right\}}.
\end{align*}
Combining these two inequalities with $\eta=\eta_{m,n}$ yields
\begin{align*}
\mathbb{P}\left(R^{\text{RC}}_{n,k}-S(\hat{f}_{k})>\varepsilon\right)
\leq 2\exp{\left\{-\frac{2n\varepsilon^{2}}{(\gamma_{m,n}(M))^{2}}\right\}}.
\end{align*}
The desired results follow from the exponential tail probability above and similar arguments used in Theorem~\ref{MainThm}.\\[0.3cm]
\noindent
\emph{Part 4}\\
For ease of notation, let $W\equiv\{W^{(j)}_{n}\}_{j=1}^{m}=\{(W^{(j)}_{n,1},W^{(j)}_{n,2},\ldots,W^{(j)}_{n,n})\}_{j=1}^{m}$,
\begin{align*}
Q^{\text{BC}}_{n,k}(W;\mathscr{D}_{n})\equiv \left(\frac{n}{n-1}\right)^{n}\frac{1}{m}\sum_{j=1}^{m}
\left(\max_{f\in\mathcal{F}_{k}}\frac{1}{n}\sum_{i=1}^{n}
\left(W^{(j)}_{n,i}-1\right)s(Y_{i},X_{i},f)\right),
\end{align*}
and $R^{\text{BC}}_{n,k}\equiv S_{n}(\hat{f}_{k})-Q^{\text{BC}}_{n,k}(W;\mathscr{D}_{n})$.

As in the proof of \emph{Part 2}, the desired results follow from the exponential tail inequality
\begin{align}
\mathbb{P}\left(R^{\text{BC}}_{n,k}-S(\hat{f}_{k})>\varepsilon\right)
\leq 2\exp{\left\{-\frac{2n\varepsilon^{2}}{(\gamma'_{m,n}(M))^{2}}\right\}}.\label{C10}
\end{align}
To establish this tail probability, we have
\begin{align}
\mathbb{P}\left(R^{\text{BC}}_{n,k}-S(\hat{f}_{k})>\varepsilon\right)
\leq& \mathbb{P}\left(\sup_{f\in\mathcal{F}_{k}}\left(S_{n}(f)-S(f)\right)
-\Exp\left[Q^{\text{BC}}_{n,k}(W;\mathscr{D}_{n})|\mathscr{D}_{n}\right]
>\eta\varepsilon\right)\notag\\
&+\mathbb{P}\left(\Exp\left[Q^{\text{BC}}_{n,k}(W;\mathscr{D}_{n})|\mathscr{D}_{n}\right]
-Q^{\text{BC}}_{n,k}(W;\mathscr{D}_{n}) >(1-\eta)\varepsilon\right)\label{C11}
\end{align}
for any $\eta\in(0,1)$.
Given $\mathscr{D}_{n}$, the mapping $\{W_{n}^{(j)}\}_{j=1}^{m}\mapsto Q^{\text{BC}}_{n,k}(W;\mathscr{D}_{n})$ satisfies the bounded differences property in Section 6.1 of \citet{BoucheronLugosiEtAl2013} with their notation $c_{j}=(8M/m)(n/(n-1))^{n}$ for each $j\in\{1,\ldots,m\}$.
It follows from \citeauthor{McDiarmid1989}'s (\citeyear{McDiarmid1989}) inequality that
\begin{align*}
&\mathbb{P}\left(\Exp\left[Q^{\text{BC}}_{n,k}(W;\mathscr{D}_{n})|\mathscr{D}_{n}\right]
-Q^{\text{BC}}_{n,k}(W;\mathscr{D}_{n}) >(1-\eta)\varepsilon|\mathscr{D}_{n}\right)\\
\leq&
\exp{\left\{-\frac{2(\frac{n-1}{n})^{2n}m(1-\eta)^{2}\varepsilon^{2}}{(8M)^{2}}\right\}}
\end{align*}
with probability one.
Taking expectation with respect to $\mathscr{D}_{n}$ on both sides yields
\begin{align}\label{C12}
\mathbb{P}\left(\Exp\left[Q^{\text{BC}}_{n,k}(W;\mathscr{D}_{n})|\mathscr{D}_{n}\right]
-Q^{\text{BC}}_{n,k}(W;\mathscr{D}_{n}) >(1-\eta)\varepsilon\right)
\leq \exp{\left\{-\frac{2m(1-\eta)^{2}\varepsilon^{2}}{(32M)^{2}}\right\}}
\end{align}
because $((n-1)/n)^{n}\geq 1/4$ for all $n\geq 2$.
If we have
\begin{align}\label{C13}
\Exp\left[\sup_{f\in\mathcal{F}_{k}}\left(S_{n}(f)-S(f)\right)\right]
\leq\Exp\left[Q^{\text{BC}}_{n,k}(W;\mathscr{D}_{n})\right],
\end{align}
then we can apply \citeauthor{McDiarmid1989}'s (\citeyear{McDiarmid1989}) inequality again and obtain
\begin{align*}
& \mathbb{P}\left(\sup_{f\in\mathcal{F}_{k}}\left(S_{n}(f)-S(f)\right)
-\Exp\left[Q^{\text{BC}}_{n,k}(W;\mathscr{D}_{n})|\mathscr{D}_{n}\right]
>\eta\varepsilon\right)\notag\\
\leq& \mathbb{P}\Big(\sup_{f\in\mathcal{F}_{k}}\left(S_{n}(f)-S(f)\right)
-\Exp\left[Q^{\text{BC}}_{n,k}(W;\mathscr{D}_{n})|\mathscr{D}_{n}\right]\notag\\
&\hspace{0.3cm}-\Exp\left[\sup_{f\in\mathcal{F}_{k}}\left(S_{n}(f)-S(f)\right)\right]
+\Exp\left[Q^{\text{BC}}_{n,k}(W;\mathscr{D}_{n})\right]
>\eta\varepsilon\Big)\notag\\
\leq & \exp{\left\{-\frac{2n\eta^{2}\varepsilon^{2}}{(8M)^{2}\Big[1+\left(\frac{n}{n-1}\right)^{n}
\Exp\left[|W^{(j)}_{n,1}-1|\right]\Big]^{2}}\right\}}
\end{align*}
because the mapping
\begin{align*}
\mathscr{D}_{n}\mapsto \sup_{f\in\mathcal{F}_{k}}\left(S_{n}(f)-S(f)\right)
-\Exp\left[Q^{\text{BC}}_{n,k}(W;\mathscr{D}_{n})|\mathscr{D}_{n}\right]
\end{align*}
satisfies the bounded differences property in Section 6.1 of \citet{BoucheronLugosiEtAl2013} with their notation
\begin{align*}
c_{i}=\frac{8M}{n}\left(1+\left(\frac{n}{n-1}\right)^{n}\Exp\left[|W^{(j)}_{n,1}-1|\right]\right)
\end{align*}
for each $i\in\{1,\ldots,n\}$.
By Lemma~\ref{Multinomial},
\begin{align*}
\left(\frac{n}{n-1}\right)^{n}\Exp\left[|W^{(j)}_{n,1}-1|\right]
=\frac{\Exp\left[(W^{(j)}_{n,1}-1)_{+}\right]+\Exp\left[(W^{(j)}_{n,1}-1)_{-}\right]}
{\Exp\left[(W^{(j)}_{n,1}-1)_{+}\right]}
=2.
\end{align*}
It follows that
\begin{align}\label{C14}
\mathbb{P}\left(\sup_{f\in\mathcal{F}_{k}}\left(S_{n}(f)-S(f)\right)
-\Exp\left[Q^{\text{BC}}_{n,k}(W;\mathscr{D}_{n})|\mathscr{D}_{n}\right]>\eta\varepsilon\right)
\leq \exp{\left\{-\frac{2n\eta^{2}\varepsilon^{2}}{(24M)^{2}}\right\}}.
\end{align}
Combining Inequalities~(\ref{C11})-(\ref{C14}) and setting
\begin{align*}
\eta=\eta'_{m,n}\equiv
\begin{cases}
3/7, & \text{if}\;\; n\leq m, \\
3/(4\ell+7), & \text{if}\;\; n/(\ell+1)^{2} \leq m < n/\ell^{2}\;\;
\text{and}\;\; \ell\in\mathbb{N},
\end{cases}
\end{align*}
yield Inequality~(\ref{C10}).

Now, it suffices to show Inequality~(\ref{C13}).
Note that
\begin{align*}
&\Exp\left[\sup_{f\in\mathcal{F}_{k}}\left(S_{n}(f)-S(f)\right)\right]\\
=& \frac{\Exp\left[\sup_{f\in\mathcal{F}_{k}}
\frac{1}{n}\sum_{i=1}^{n}\Exp\left[(W_{n,i}-1)\Ind{[W_{n,i}\geq 1]}\right]\Big(s(Y_{i},X_{i},f)-\Exp[s(Y_{i},X_{i},f)]\Big)\right]}
{\Exp\left[(W_{n,1}-1)_{+}\right]}.
\end{align*}
Applying Jensen's inequality to the numerator, we have
\begin{align}
&\Exp\left[\sup_{f\in\mathcal{F}_{k}}
\frac{1}{n}\sum_{i=1}^{n}\Exp\left[(W_{n,i}-1)\Ind{[W_{n,i}\geq 1]}\right]\Big(s(Y_{i},X_{i},f)-\Exp[s(Y_{i},X_{i},f)]\Big)\right]\notag\\
= & \Exp\left[\sup_{f\in\mathcal{F}_{k}}\Exp\left[
\frac{1}{n}\sum_{i=1}^{n}(W_{n,i}-1)\Ind{[W_{n,i}\geq 1]}\Big(s(Y_{i},X_{i},f)-\Exp[s(Y_{i},X_{i},f)]\Big)\Big|\mathscr{D}_{n}\right]\right]\notag\\
\leq & \Exp\left[\sup_{f\in\mathcal{F}_{k}}
\frac{1}{n}\sum_{i=1}^{n}(W_{n,i}-1)\Ind{[W_{n,i}\geq 1]}\Big(s(Y_{i},X_{i},f)-\Exp[s(Y_{i},X_{i},f)]\Big)\right].\label{C15}
\end{align}
Let $\mathscr{D}_{n,\geq 1}$ be the largest subset of $\mathscr{D}_{n}$ such that each observation $(Y_{i},X_{i})$ in $\mathscr{D}_{n,\geq 1}$ has the concomitant $W_{n,i}$ greater than or equal to $1$.
Note that
\begin{align*}
&\Exp\left[\frac{1}{n}\sum_{i=1}^{n}(W_{n,i}-1)\Ind{[W_{n,i}\geq 1]}
\Big(s(Y_{i},X_{i},f)-\Exp[s(Y_{i},X_{i},f)]\Big)\Big|\{W_{n,i}\}_{i=1}^{n},\mathscr{D}_{n,\geq 1}\right]\\
=&\frac{1}{n}\sum_{i=1}^{n}(W_{n,i}-1)\Ind{[W_{n,i}\geq 1]}
\Big(s(Y_{i},X_{i},f)-\Exp[s(Y_{i},X_{i},f)]\Big)
\end{align*}
and
\begin{align*}
\Exp\left[\frac{1}{n}\sum_{i=1}^{n}(W_{n,i}-1)\Ind{[W_{n,i}< 1]}
\Big(s(Y_{i},X_{i},f)-\Exp[s(Y_{i},X_{i},f)]\Big)\Big|\{W_{n,i}\}_{i=1}^{n},\mathscr{D}_{n,\geq 1}\right]
=0.
\end{align*}
It follows that
\begin{align}
&\Exp\left[\sup_{f\in\mathcal{F}_{k}}
\frac{1}{n}\sum_{i=1}^{n}(W_{n,i}-1)\Ind{[W_{n,i}\geq 1]}
\Big(s(Y_{i},X_{i},f)-\Exp[s(Y_{i},X_{i},f)]\Big)\right]\notag\\
=&\Exp\left[\sup_{f\in\mathcal{F}_{k}}\Exp\left[
\frac{1}{n}\sum_{i=1}^{n}(W_{n,i}-1)\Ind{[W_{n,i}\geq 1]}
\Big(s(Y_{i},X_{i},f)-\Exp[s(Y_{i},X_{i},f)]\Big)\Big|\{W_{n,i}\}_{i=1}^{n},\mathscr{D}_{n,\geq 1}\right]\right]\notag\\
=&\Exp\left[\sup_{f\in\mathcal{F}_{k}}\Exp\left[
\frac{1}{n}\sum_{i=1}^{n}(W_{n,i}-1)
\Big(s(Y_{i},X_{i},f)-\Exp[s(Y_{i},X_{i},f)]\Big)\Big|\{W_{n,i}\}_{i=1}^{n},\mathscr{D}_{n,\geq 1}\right]\right]\notag\\
\leq& \Exp\left[\sup_{f\in\mathcal{F}_{k}}
\frac{1}{n}\sum_{i=1}^{n}(W_{n,i}-1)\Big(s(Y_{i},X_{i},f)-\Exp[s(Y_{i},X_{i},f)]\Big)\right]\notag\\
=& \Exp\left[\sup_{f\in\mathcal{F}_{k}}
\frac{1}{n}\sum_{i=1}^{n}(W_{n,i}-1)s(Y_{i},X_{i},f)\right].\label{C16}
\end{align}
Combining~(\ref{C15}) and~(\ref{C16}) yields
\begin{align*}
&\Exp\left[\sup_{f\in\mathcal{F}_{k}}
\frac{1}{n}\sum_{i=1}^{n}\Exp\left[(W_{n,i}-1)\Ind{[W_{n,i}\geq 1]}\right]\Big(s(Y_{i},X_{i},f)-\Exp[s(Y_{i},X_{i},f)]\Big)\right]\\
\leq& \Exp\left[\sup_{f\in\mathcal{F}_{k}}
\frac{1}{n}\sum_{i=1}^{n}(W_{n,i}-1)s(Y_{i},X_{i},f)\right].
\end{align*}
Lemma~\ref{Multinomial} shows that $\Exp\left[(W_{n,1}-1)_{+}\right]=\left(1-1/n\right)^{n}$.
Therefore, we obtain
\begin{align*}
\Exp\left[\sup_{f\in\mathcal{F}_{k}}\left(S_{n}(f)-S(f)\right)\right]
\leq& \frac{\Exp\left[\max_{f\in\mathcal{F}_{k}}
\frac{1}{n}\sum_{i=1}^{n}(W_{n,i}-1)s(Y_{i},X_{i},f)\right]}{\Exp\left[(W_{n,1}-1)_{+}\right]}\\
=&\Exp\left[Q^{\text{BC}}_{n,k}(W;\mathscr{D}_{n})\right].
\end{align*}

\end{proof}

\subsection{Proof of Corollary~\ref{Corrollary3}}
\begin{proof}
We first consider the MD, SMD, and RC penalties.
Since
\begin{align*}
24M\chi_{n}(k;\alpha_{0})<\gamma_{m,n}(M)\chi_{n}(k;\alpha_{0})
= (24+16\bar{\ell})M\sqrt{\frac{1+\alpha_{0}}{2}}\sqrt{\frac{\log{\{V_{k,c}\}}}{n}}
\end{align*}
and
\begin{align*}
&\Exp\left[
\max_{f\in\mathcal{F}_{k}}\left(\frac{2}{n}\sum_{i=1}^{n/2}s(Y_{i},X_{i},f)
-\frac{2}{n}\sum_{i=n/2+1}^{n}s(Y_{i},X_{i},f)\right)
\right]\\
=&\Exp\left[
\frac{1}{m}\sum_{j=1}^{m}\left(\max_{f\in\mathcal{F}_{k}}\frac{2}{n}\sum_{i=1}^{n/2}\sigma^{(j)}_{i}
\Big(s(Y_{2i-1},X_{2i-1},f)-s(Y_{2i},X_{2i},f)\Big)\right)
\right]\\
\leq& 2\Exp\left[
\frac{1}{m}\sum_{j=1}^{m}\left(\max_{f\in\mathcal{F}_{k}}
\frac{2}{n}\sum_{i=1}^{n/2}\sigma^{(j)}_{i}s(Y_{i},X_{i},f)\right)
\right],
\end{align*}
it suffices to find an appropriate upper bound on $\Exp\left[\max_{f\in\mathcal{F}_{k}}\frac{1}{n}\sum_{i=1}^{n/2}\sigma_{i}s(Y_{i},X_{i},f)\right]$.
Let $Z_{i}= \sigma_{i}b(X_{i})[Y_{i}+1-2c(X_{i})]$ for each $i$.
We have $\Exp[Z_{i}|\mathscr{D}_{n}]=0$ for each $i$ and $\Exp[|Z_{1}|^{\ell}|\mathscr{D}_{n}]\leq (\ell!/2)(4M)^{\ell}$ for each $\ell\geq 2$.
Notice that
\begin{align*}
\Exp\left[\max_{f\in\mathcal{F}_{k}}\frac{1}{n}\sum_{i=1}^{n/2}
\sigma_{i}s(Y_{i},X_{i},f)\Big|\mathscr{D}_{n}\right]
=&\frac{2}{n}\Exp\left[
\max_{f\in\mathcal{F}_{k}}\sum_{i=1}^{n/2}\Ind{[f(X_{i})-c(X_{i})\geq 0]}Z_{i}\Big|\mathscr{D}_{n}\right]\\
=&\frac{2}{n}\Exp\left[
\max_{a=(a_{1},\ldots,a_{n})\in\mathbb{A}_{k,c}(\mathscr{D}_{n})}
\sum_{i=1}^{n/2}a_{i}Z_{i}\Big|\mathscr{D}_{n}\right]
\end{align*}
where $\mathbb{A}_{k,c}(\mathscr{D}_{n})\equiv \{(\Ind{B}(X_{1}),\ldots, \Ind{B}(X_{n})):B\in\mathcal{B}_{k,c}\}$ and $\mathcal{B}_{k,c}\equiv\{\{x\in\mathcal{X}:f(x)-c(x)\geq 0\}: f\in\mathcal{F}_{k}\}$.
Note that the VC dimension of $\mathcal{B}_{k,c}$ is $V_{k,c}$.
It follows by the chaining technique in Lemma 3 and Theorem 4 of \citet{Fromont2007} that there are positive constants $\bar{\kappa}_{1}$ and $\bar{\kappa}_{2}$, only depending on $M$, such that
\begin{align*}
\Exp\left[\max_{a=(a_{1},\ldots,a_{n})\in\mathbb{A}_{k,c}(\mathscr{D}_{n})}
\sum_{i=1}^{n/2}a_{i}Z_{i}\Big|\mathscr{D}_{n}\right]
\leq \bar{\kappa}_{1}\sqrt{V_{k,c}}\sqrt{n}+\bar{\kappa}_{2}V_{k,c}(\log\{n\})^{2}
\end{align*}
for each $k\in\mathbb{N}$ and $n\geq 8$.
Hence, taking expectation on both sides, we obtain
\begin{align*}
\Exp\left[\max_{f\in\mathcal{F}_{k}}\frac{1}{n}\sum_{i=1}^{n/2}\sigma_{i}s(Y_{i},X_{i},f)\right]
\leq 2\left(\bar{\kappa}_{1}\sqrt{\frac{V_{k,c}}{n}}
+\bar{\kappa}_{2}V_{k,c}\frac{(\log\{n\})^{2}}{n}\right).
\end{align*}

Next, we consider the BC penalty.
As in the proof of \citet{Fromont2007}, we apply Poissonization to remove the dependence of $(W_{n,1},\ldots, W_{n,n})$.
Let $\{U_{i}\}_{i=1}^{n}$ be a sequence of i.i.d.\ random variables independent of $\mathscr{D}_{n}$ and uniformly distributed on $(0,1)$ such that
\begin{align*}
W_{n,i}=\sum_{\jmath=1}^{n}\Ind{[U_{\jmath}\in((i-1)/n,i/n]]}.
\end{align*}
Let $N$ be the Poisson random variable with parameter $n$ that is independent of $\mathscr{D}_{n}$ and $\{U_{i}\}_{i=1}^{n}$.
For each $i\in\{1,\ldots, n\}$, we define
\begin{align*}
N_{i}=\sum_{\jmath=1}^{N}\Ind{[U_{\jmath}\in((i-1)/n,i/n]]}.
\end{align*}
It can be shown that
$\{N_{i}\}_{i=1}^{n}$ is a sequence of i.i.d.\ random variables independent of $\mathscr{D}_{n}$ and each $N_{i}$ follows a Poisson distribution with parameter 1; additionally,
\begin{align*}
&\left|\Exp\left[\frac{1}{m}\sum_{j=1}^{m}
\Bigg(\max_{f\in\mathcal{F}_{k}}\frac{1}{n}\sum_{i=1}^{n}
\left(W^{(j)}_{n,i}-1\right)s(Y_{i},X_{i},f)\right)\Big|\mathscr{D}_{n}\right]\\
&\hspace{2.7cm}-\Exp\left[\max_{f\in\mathcal{F}_{k}}
\frac{1}{n}\sum_{i=1}^{n}\left(N_{i}-1\right)s(Y_{i},X_{i},f)\Big|\mathscr{D}_{n}\right]\Bigg|
\leq \frac{4M}{n}\Exp[|N-n|]
\leq \frac{4M}{\sqrt{n}}.
\end{align*}
Let $\tilde{Z}_{i}=(N_{i}-1)b(X_{i})[Y_{i}+1-2c(X_{i})]$ for each $i$.
As in the proof of \citet{Fromont2007}, we have $\Exp[\tilde{Z}_{i}|\mathscr{D}_{n}]=0$ for each $i$ and $\Exp[|\tilde{Z}_{1}|^{\ell}|\mathscr{D}_{n}]\leq 1.4\cdot(\ell!/2)(4M)^{\ell}$ for each $\ell\geq 2$.
Then there are positive constants $\tilde{\kappa}_{1}$ and $\tilde{\kappa}_{2}$, only depending on $M$, such that for each $k\in\mathbb{N}$ and $n\geq 8$,
\begin{align*}
\Exp\left[\max_{f\in\mathcal{F}_{k}}\frac{1}{n}\sum_{i=1}^{n}
\left(N_{i}-1\right)s(Y_{i},X_{i},f)\Big|\mathscr{D}_{n}\right]
=&\frac{2}{n}\Exp\left[\max_{f\in\mathcal{F}_{k}}
\sum_{i=1}^{n}\Ind{[f(X_{i})-c(X_{i})>0]}\tilde{Z}_{i}\Big|\mathscr{D}_{n}\right]\\
=&\frac{2}{n}\Exp\left[\max_{a=(a_{1},\ldots,a_{n})\in\mathbb{A}_{k,c}(\mathscr{D}_{n})}
\sum_{i=1}^{n}a_{i}\tilde{Z}_{i}\Big|\mathscr{D}_{n}\right]\\
\leq& \tilde{\kappa}_{1}\sqrt{\frac{V_{k,c}}{n}}+\tilde{\kappa}_{2}V_{k,c}\frac{(\log\{n\})^{2}}{n}
\end{align*}
by the chaining technique in Lemma 3 and Theorem 4 of \citet{Fromont2007}.
Therefore, taking expectation on both sides yields
\begin{align*}
&\Exp\left[\left(\frac{n}{n-1}\right)^{n}\frac{1}{m}\sum_{j=1}^{m}
\left(\max_{f\in\mathcal{F}_{k}}\frac{1}{n}\sum_{i=1}^{n}
\left(W^{(j)}_{n,i}-1\right)s(Y_{i},X_{i},f)\right)\right]\\
\leq &\left(\frac{n}{n-1}\right)^{n}
\left\{\Exp\left[\max_{f\in\mathcal{F}_{k}}\frac{1}{n}
\sum_{i=1}^{n}\left(N_{i}-1\right)s(Y_{i},X_{i},f)\right]
+\frac{4M}{\sqrt{n}}\right\}\\
\leq &\left(\frac{n}{n-1}\right)^{n}
\left\{\tilde{\kappa}_{1}\sqrt{\frac{V_{k,c}}{n}}+\tilde{\kappa}_{2}V_{k,c}\frac{(\log\{n\})^{2}}{n}
+\frac{4M}{\sqrt{n}}\right\}\\
\leq &3\left\{(\tilde{\kappa}_{1}+4M)\sqrt{\frac{V_{k,c}}{n}}
+\tilde{\kappa}_{2}V_{k,c}\frac{(\log\{n\})^{2}}{n}\right\}
\end{align*}
for each $k\in\mathbb{N}$ and $n\geq 8$.
Note that the technical term in BC penalty satisfies
\begin{align*}
\gamma'_{m,n}(M)\chi_{n}(k;\alpha_{0})
= (24+32\bar{\ell})M\sqrt{\frac{1+\alpha_{0}}{2}}\sqrt{\frac{\log{\{V_{k,c}\}}}{n}}.
\end{align*}

Finally, the proof is completed by showing the universal utility consistency.
We have established that for each $k\in\mathbb{N}$, $\Exp\left[C^{\text{MD}}_{n}(k;\alpha_{0})\right]=\Bigo{n^{-1/2}}$, $\Exp\left[C^{\text{SMD}}_{n}(k;\alpha_{0},m)|\mathscr{D}_{n}\right]=\Bigo{n^{-1/2}}$,
$\Exp\left[C^{\text{RC}}_{n}(k;\alpha_{0},m)|\mathscr{D}_{n}\right]=\Bigo{n^{-1/2}}$, and $\Exp\left[C_{n}^{\text{BC}}(k;\alpha_{0},m)|\mathscr{D}_{n}\right]=\Bigo{n^{-1/2}}$ with probability one.
Hence, a simple modification of the proof of Corollary~\ref{Corollary2} yields the results.

\end{proof}

\noindent
Lemma~\ref{expectationbound} below is a slightly revised version of Problem 12.1 of \citet{DevroyeGyoerfiEtAl1996}.
\begin{Lem}\label{expectationbound}
If a random variable $Z$ satisfies
\begin{align*}
\mathbb{P}(Z>\epsilon)\leq c_{1}\exp\{-c_{2}\epsilon^{2}\}
\end{align*}
for all $\epsilon>0$ and some positive numbers $c_{1}$ and $c_{2}$, then
\begin{align*}
\Exp[Z]\leq \sqrt{\frac{(1+\log{\{c_{1}\}})}{c_{2}}}.
\end{align*}
\end{Lem}
\begin{proof}
For any $v>0$, we have
\begin{align*}
\Exp[Z^{2}\Ind{[Z\geq 0]}]
=&\int_{0}^{\infty}\mathbb{P}(Z^{2}\Ind{[Z\geq 0]}>t)\Myd t\\
\leq& v+\int_{v}^{\infty}\mathbb{P}\left(Z\Ind{[Z\geq 0]}>\sqrt{t}\right)\Myd t\\
=& v+\int_{v}^{\infty}\mathbb{P}\left(Z>\sqrt{t}\right)\Myd t\\
\leq& v+c_{1}\int_{v}^{\infty}\exp\{-c_{2}t\}\Myd t\\
=& v+\frac{c_{1}}{c_{2}}\exp\{-c_{2}v\}.
\end{align*}
Taking $v=\frac{\log{\{c_{1}\}}}{c_{2}}$ yields
\begin{align*}
\Exp[Z]
\leq \Exp[Z\Ind{[Z\geq 0]}]
\leq \sqrt{\frac{(1+\log{\{c_{1}\}})}{c_{2}}}.
\end{align*}
\end{proof}

\begin{Lem}\label{Multinomial}
Suppose $(W_{n,1},W_{n,2},\ldots,W_{n,n})$ is a multinomial vector with parameters $n$ and $(1/n,1/n,\ldots,1/n)$.
Then for each $i\in\{1,2,\ldots,n\}$,
\begin{align*}
\Exp\left[(W_{n,i}-1)_{+}\right]=\Exp\left[(W_{n,i}-1)_{-}\right]=\left(\frac{n-1}{n}\right)^{n}.
\end{align*}
\end{Lem}
\begin{proof}
For each $i$, $\Exp\left[(W_{n,i}-1)_{+}\right]=\Exp\left[W_{n,i}-1\right]+\Exp\left[(W_{n,i}-1)_{-}\right]
=\Exp\left[(W_{n,i}-1)_{-}\right]$ because $\Exp\left[W_{n,i}\right]=1$.
Hence, we have
\begin{align*}
\Exp\left[(W_{n,i}-1)_{+}\right]
=\sum_{w=0}^{n}(w-1)_{-}\mathbb{P}(W_{n,i}=w)
=\mathbb{P}(W_{n,i}=0)
=\left(1-\frac{1}{n}\right)^{n}.
\end{align*}
\end{proof}

\end{appendices}

\normalsize
\bibliography{Model_Selection_Final_Submission}

\providecommand{\cleftquote}{〈} \providecommand{\crightquote}{〉}
  \providecommand{\cleftqquote}{《} \providecommand{\crightqquote}{》}
  \providecommand{\cperiod}{。} \providecommand{\cin}{收錄於}
  \providecommand{\cedit}{編} \providecommand{\cedn}{版}
  \providecommand{\cmaster}{碩士論文} \providecommand{\cphd}{博士論文}
  \providecommand{\ctechreport}{技術報告}
\begin{thebibliography}{50}
\expandafter\ifx\csname natexlab\endcsname\relax\def\natexlab#1{#1}\fi
\providecommand{\url}[1]{\texttt{#1}}
\providecommand{\href}[2]{#2}
\providecommand{\path}[1]{#1}
\providecommand{\DOIprefix}{doi:}
\providecommand{\ArXivprefix}{arXiv:}
\providecommand{\URLprefix}{URL: }
\providecommand{\Pubmedprefix}{pmid:}
\providecommand{\doi}[1]{\href{http://dx.doi.org/#1}{\path{#1}}}
\providecommand{\Pubmed}[1]{\href{pmid:#1}{\path{#1}}}
\providecommand{\bibinfo}[2]{#2}
\ifx\xfnm\relax \def\xfnm[#1]{\unskip,\space#1}\fi
\bibitem[{Akaike(1973)}]{Akaike1973}
\bibinfo{author}{Akaike, H.}, \bibinfo{year}{1973}.
\newblock \bibinfo{title}{Information theory and an extension of the maximum
  likelihood principle}, in: \bibinfo{editor}{Petrov, B.N.},
  \bibinfo{editor}{Csaki, F.} (Eds.), \bibinfo{booktitle}{Second International
  Symposium on Information Theory}, pp. \bibinfo{pages}{267--281}.
\bibitem[{Anthony and Bartlett(1999)}]{AnthonyBartlett1999}
\bibinfo{author}{Anthony, M.}, \bibinfo{author}{Bartlett, P.L.},
  \bibinfo{year}{1999}.
\newblock \bibinfo{title}{Neural Network Learning: Theoretical Foundations}.
\newblock \bibinfo{publisher}{Cambridge University Press}.
\newblock \DOIprefix\doi{10.1017/CBO9780511624216}.
\bibitem[{Arlot and Celisse(2010)}]{ArlotCelisse2010}
\bibinfo{author}{Arlot, S.}, \bibinfo{author}{Celisse, A.},
  \bibinfo{year}{2010}.
\newblock \bibinfo{title}{A survey of cross-validation procedures for model
  selection}.
\newblock \bibinfo{journal}{Statistics Surveys} \bibinfo{volume}{4},
  \bibinfo{pages}{40--79}.
\newblock \DOIprefix\doi{10.1214/09-SS054}.
\bibitem[{Athey and Imbens(2019)}]{AtheyImbens2019}
\bibinfo{author}{Athey, S.}, \bibinfo{author}{Imbens, G.W.},
  \bibinfo{year}{2019}.
\newblock \bibinfo{title}{Machine learning methods that economists should know
  about}.
\newblock \bibinfo{journal}{Annual Review of Economics} \bibinfo{volume}{11},
  \bibinfo{pages}{685--725}.
\newblock \DOIprefix\doi{10.1146/annurev-economics-080217-053433}.
\bibitem[{Bagby et~al.(2002)Bagby, Bos and Levenberg}]{BagbyBosEtAl2002}
\bibinfo{author}{Bagby, T.}, \bibinfo{author}{Bos, L.},
  \bibinfo{author}{Levenberg, N.}, \bibinfo{year}{2002}.
\newblock \bibinfo{title}{Multivariate simultaneous approximation}.
\newblock \bibinfo{journal}{Constructive Approximation} \bibinfo{volume}{18},
  \bibinfo{pages}{569--577}.
\newblock \DOIprefix\doi{10.1007/s00365-001-0024-6}.
\bibitem[{Barberis and Xiong(2012)}]{BarberisXiong2012}
\bibinfo{author}{Barberis, N.}, \bibinfo{author}{Xiong, W.},
  \bibinfo{year}{2012}.
\newblock \bibinfo{title}{Realization utility}.
\newblock \bibinfo{journal}{Journal of Financial Economics}
  \bibinfo{volume}{104}, \bibinfo{pages}{251--271}.
\newblock \DOIprefix\doi{http://dx.doi.org/10.1016/j.jfineco.2011.10.005}.
\bibitem[{Bartlett et~al.(2002)Bartlett, Boucheron and
  Lugosi}]{BartlettBoucheronEtAl2002}
\bibinfo{author}{Bartlett, P.L.}, \bibinfo{author}{Boucheron, S.},
  \bibinfo{author}{Lugosi, G.}, \bibinfo{year}{2002}.
\newblock \bibinfo{title}{Model selection and error estimation}.
\newblock \bibinfo{journal}{Machine Learning} \bibinfo{volume}{48},
  \bibinfo{pages}{85--113}.
\newblock \DOIprefix\doi{10.1023/A:1013999503812}.
\bibitem[{Boucheron et~al.(2013)Boucheron, Lugosi and
  Massart}]{BoucheronLugosiEtAl2013}
\bibinfo{author}{Boucheron, S.}, \bibinfo{author}{Lugosi, G.},
  \bibinfo{author}{Massart, P.}, \bibinfo{year}{2013}.
\newblock \bibinfo{title}{Concentration Inequalities: A Nonasymptotic Theory of
  Independence}.
\newblock \bibinfo{publisher}{Oxford University Press}.
\newblock \DOIprefix\doi{10.1093/acprof:oso/9780199535255.001.0001}.
\bibitem[{Chen and Lee(2018a)}]{ChenLee2018}
\bibinfo{author}{Chen, L.Y.}, \bibinfo{author}{Lee, S.}, \bibinfo{year}{2018}a.
\newblock \bibinfo{title}{Best subset binary prediction}.
\newblock \bibinfo{journal}{Journal of Econometrics} \bibinfo{volume}{206},
  \bibinfo{pages}{39--56}.
\newblock \DOIprefix\doi{https://doi.org/10.1016/j.jeconom.2018.05.001}.
\bibitem[{Chen and Lee(2018b)}]{ChenLee2018b}
\bibinfo{author}{Chen, L.Y.}, \bibinfo{author}{Lee, S.}, \bibinfo{year}{2018}b.
\newblock \bibinfo{title}{{High Dimensional Classification through
  $\ell_{0}$-Penalized Empirical Risk Minimization}}.
\newblock \bibinfo{journal}{arXiv e-prints}
  \href{http://arxiv.org/abs/1811.09540}{{\tt arXiv:1811.09540}}.
\bibitem[{Chen(2007)}]{Chen2007}
\bibinfo{author}{Chen, X.}, \bibinfo{year}{2007}.
\newblock \bibinfo{title}{Large sample sieve estimation of semi-nonparametric
  models}, \bibinfo{publisher}{Elsevier}. volume \bibinfo{volume}{6, Part B} of
  \textit{\bibinfo{series}{Handbook of Econometrics}}.
  chapter~\bibinfo{chapter}{76}, pp. \bibinfo{pages}{5549--5632}.
\newblock \DOIprefix\doi{https://doi.org/10.1016/S1573-4412(07)06076-X}.
\bibitem[{Christensen et~al.(2012)Christensen, Larsen and
  Munk}]{ChristensenLarsenEtAl2012}
\bibinfo{author}{Christensen, P.O.}, \bibinfo{author}{Larsen, K.},
  \bibinfo{author}{Munk, C.}, \bibinfo{year}{2012}.
\newblock \bibinfo{title}{Equilibrium in securities markets with heterogeneous
  investors and unspanned income risk}.
\newblock \bibinfo{journal}{Journal of Economic Theory} \bibinfo{volume}{147},
  \bibinfo{pages}{1035--1063}.
\newblock \DOIprefix\doi{https://doi.org/10.1016/j.jet.2012.01.007}.
\bibitem[{Claeskens and Hjort(2008)}]{ClaeskensHjort2008}
\bibinfo{author}{Claeskens, G.}, \bibinfo{author}{Hjort, N.L.},
  \bibinfo{year}{2008}.
\newblock \bibinfo{title}{Model Selection and Model Averaging}.
\newblock Cambridge Series in Statistical and Probabilistic Mathematics,
  \bibinfo{publisher}{Cambridge University Press}.
\newblock \DOIprefix\doi{10.1017/CBO9780511790485}.
\bibitem[{Devroye et~al.(1996)Devroye, Gy{\"o}rfi and
  Lugosi}]{DevroyeGyoerfiEtAl1996}
\bibinfo{author}{Devroye, L.}, \bibinfo{author}{Gy{\"o}rfi, L.},
  \bibinfo{author}{Lugosi, G.}, \bibinfo{year}{1996}.
\newblock \bibinfo{title}{A Probabilistic Theory of Pattern Recognition}.
\newblock \bibinfo{publisher}{Springer}.
\newblock \DOIprefix\doi{10.1007/978-1-4612-0711-5}.
\bibitem[{Efron(1979)}]{Efron1979}
\bibinfo{author}{Efron, B.}, \bibinfo{year}{1979}.
\newblock \bibinfo{title}{Bootstrap methods: Another look at the jackknife}.
\newblock \bibinfo{journal}{Annals of Statistics} \bibinfo{volume}{7},
  \bibinfo{pages}{1--26}.
\newblock \DOIprefix\doi{10.1214/aos/1176344552}.
\bibitem[{Efron and Hastie(2016)}]{EfronHastie2016}
\bibinfo{author}{Efron, B.}, \bibinfo{author}{Hastie, T.},
  \bibinfo{year}{2016}.
\newblock \bibinfo{title}{Computer Age Statistical Inference: Algorithms,
  Evidence, and Data Science}.
\newblock \bibinfo{publisher}{Cambridge University Press}.
\newblock \DOIprefix\doi{10.1017/CBO9781316576533}.
\bibitem[{Elliott and Lieli(2013)}]{ElliottLieli2013}
\bibinfo{author}{Elliott, G.}, \bibinfo{author}{Lieli, R.P.},
  \bibinfo{year}{2013}.
\newblock \bibinfo{title}{Predicting binary outcomes}.
\newblock \bibinfo{journal}{Journal of Econometrics} \bibinfo{volume}{174},
  \bibinfo{pages}{15--26}.
\newblock \DOIprefix\doi{http://dx.doi.org/10.1016/j.jeconom.2013.01.003}.
\bibitem[{Elliott and Timmermann(2016)}]{ElliottTimmermann2016}
\bibinfo{author}{Elliott, G.}, \bibinfo{author}{Timmermann, A.},
  \bibinfo{year}{2016}.
\newblock \bibinfo{title}{Forecasting in economics and finance}.
\newblock \bibinfo{journal}{Annual Review of Economics} \bibinfo{volume}{8},
  \bibinfo{pages}{81--110}.
\newblock \DOIprefix\doi{10.1146/annurev-economics-080315-015346}.
\bibitem[{Fromont(2007)}]{Fromont2007}
\bibinfo{author}{Fromont, M.}, \bibinfo{year}{2007}.
\newblock \bibinfo{title}{Model selection by bootstrap penalization for
  classification}.
\newblock \bibinfo{journal}{Machine Learning} \bibinfo{volume}{66},
  \bibinfo{pages}{165--207}.
\newblock \DOIprefix\doi{10.1007/s10994-006-7679-y}.
\bibitem[{Fung and Mangasarian(2004)}]{FungMangasarian2004}
\bibinfo{author}{Fung, G.M.}, \bibinfo{author}{Mangasarian, O.},
  \bibinfo{year}{2004}.
\newblock \bibinfo{title}{A feature selection newton method for support vector
  machine classification}.
\newblock \bibinfo{journal}{Computational Optimization and Applications}
  \bibinfo{volume}{28}, \bibinfo{pages}{185--202}.
\newblock \DOIprefix\doi{10.1023/B:COAP.0000026884.66338.df}.
\bibitem[{Geman and Hwang(1982)}]{GemanHwang1982}
\bibinfo{author}{Geman, S.}, \bibinfo{author}{Hwang, C.R.},
  \bibinfo{year}{1982}.
\newblock \bibinfo{title}{Nonparametric maximum likelihood estimation by the
  method of sieves}.
\newblock \bibinfo{journal}{Annals of Statistics} \bibinfo{volume}{10},
  \bibinfo{pages}{401--414}.
\newblock \DOIprefix\doi{10.1214/aos/1176345782}.
\bibitem[{Granger and Machina(2006)}]{GrangerMachina2006}
\bibinfo{author}{Granger, C.W.}, \bibinfo{author}{Machina, M.J.},
  \bibinfo{year}{2006}.
\newblock \bibinfo{title}{Forecasting and decision theory},
  \bibinfo{publisher}{Elsevier}. volume~\bibinfo{volume}{1} of
  \textit{\bibinfo{series}{Handbook of Economic Forecasting}}.
  chapter~\bibinfo{chapter}{2}, pp. \bibinfo{pages}{81--98}.
\newblock \DOIprefix\doi{https://doi.org/10.1016/S1574-0706(05)01002-5}.
\bibitem[{Greenshtein(2006)}]{Greenshtein2006}
\bibinfo{author}{Greenshtein, E.}, \bibinfo{year}{2006}.
\newblock \bibinfo{title}{Best subset selection, persistence in
  high-dimensional statistical learning and optimization under $l_1$
  constraint}.
\newblock \bibinfo{journal}{Annals of Statistics} \bibinfo{volume}{34},
  \bibinfo{pages}{2367--2386}.
\newblock \DOIprefix\doi{10.1214/009053606000000768}.
\bibitem[{Greenshtein and Ritov(2004)}]{GreenshteinRitov2004}
\bibinfo{author}{Greenshtein, E.}, \bibinfo{author}{Ritov, Y.},
  \bibinfo{year}{2004}.
\newblock \bibinfo{title}{Persistence in high-dimensional linear predictor
  selection and the virtue of overparametrization}.
\newblock \bibinfo{journal}{Bernoulli} \bibinfo{volume}{10},
  \bibinfo{pages}{971--988}.
\newblock \DOIprefix\doi{10.3150/bj/1106314846}.
\bibitem[{Hastie et~al.(2009)Hastie, Tibshirani and
  Friedman}]{HastieTibshiraniEtAl2009}
\bibinfo{author}{Hastie, T.}, \bibinfo{author}{Tibshirani, R.},
  \bibinfo{author}{Friedman, J.}, \bibinfo{year}{2009}.
\newblock \bibinfo{title}{The Elements of Statistical Learning}.
\newblock \bibinfo{edition}{2nd} ed., \bibinfo{publisher}{Springer}.
\newblock \DOIprefix\doi{10.1007/978-0-387-84858-7}.
\bibitem[{Homrighausen and McDonald(2013)}]{HomrighausenMcDonald2013}
\bibinfo{author}{Homrighausen, D.}, \bibinfo{author}{McDonald, D.},
  \bibinfo{year}{2013}.
\newblock \bibinfo{title}{The lasso, persistence, and cross-validation}, in:
  \bibinfo{booktitle}{Proceedings of the 30th International Conference on
  Machine Learning}, pp. \bibinfo{pages}{1031--1039}.
\newblock \URLprefix
  \url{http://proceedings.mlr.press/v28/homrighausen13.html}.
\bibitem[{Koltchinskii(2001)}]{Koltchinskii2001}
\bibinfo{author}{Koltchinskii, V.}, \bibinfo{year}{2001}.
\newblock \bibinfo{title}{Rademacher penalties and structural risk
  minimization}.
\newblock \bibinfo{journal}{IEEE Transactions on Information Theory}
  \bibinfo{volume}{47}, \bibinfo{pages}{1902--1914}.
\newblock \DOIprefix\doi{10.1109/18.930926}.
\bibitem[{Koltchinskii(2011)}]{Koltchinskii2011}
\bibinfo{author}{Koltchinskii, V.}, \bibinfo{year}{2011}.
\newblock \bibinfo{title}{Oracle Inequalities in Empirical Risk Minimization
  and Sparse Recovery Problems}.
\newblock \bibinfo{publisher}{Springer}.
\newblock \DOIprefix\doi{https://doi.org/10.1007/978-3-642-22147-7}.
\bibitem[{Konishi and Kitagawa(2008)}]{KonishiKitagawa2008}
\bibinfo{author}{Konishi, S.}, \bibinfo{author}{Kitagawa, G.},
  \bibinfo{year}{2008}.
\newblock \bibinfo{title}{Information Criteria and Statistical Modeling}.
\newblock \bibinfo{publisher}{Springer}.
\newblock \DOIprefix\doi{https://doi.org/10.1007/978-0-387-71887-3}.
\bibitem[{Kosorok(2008)}]{Kosorok2008}
\bibinfo{author}{Kosorok, M.R.}, \bibinfo{year}{2008}.
\newblock \bibinfo{title}{Introduction to Empirical Processes and
  Semiparametric Inference}.
\newblock \bibinfo{publisher}{Springer}.
\newblock \DOIprefix\doi{10.1007/978-0-387-74978-5}.
\bibitem[{Leeb and P\"{o}tscher(2005)}]{LeebPoetscher2005}
\bibinfo{author}{Leeb, H.}, \bibinfo{author}{P\"{o}tscher, B.M.},
  \bibinfo{year}{2005}.
\newblock \bibinfo{title}{Model selection and inference: Facts and fiction}.
\newblock \bibinfo{journal}{Econometric Theory} \bibinfo{volume}{21},
  \bibinfo{pages}{21--59}.
\newblock \DOIprefix\doi{https://doi.org/10.1017/S0266466605050036}.
\bibitem[{Leeb and P\"{o}tscher(2008)}]{LeebPoetscher2008a}
\bibinfo{author}{Leeb, H.}, \bibinfo{author}{P\"{o}tscher, B.M.},
  \bibinfo{year}{2008}.
\newblock \bibinfo{title}{Can one estimate the unconditional distribution of
  post-model-selection estimators?}
\newblock \bibinfo{journal}{Econometric Theory} \bibinfo{volume}{24},
  \bibinfo{pages}{338--376}.
\newblock \DOIprefix\doi{10.1017/S0266466608080158}.
\bibitem[{Lozano(2000)}]{Lozano2000}
\bibinfo{author}{Lozano, F.}, \bibinfo{year}{2000}.
\newblock \bibinfo{title}{Model selection using rademacher penalization}, in:
  \bibinfo{booktitle}{Proceedings of the Second ICSC Symposia on Neural
  Computation (NC2000). ICSC Adademic}.
\newblock \URLprefix
  \url{http://citeseerx.ist.psu.edu/viewdoc/citations;jsessionid=6BD3CA2AAD5E5CD6DA6E9E00DA8A11A0?doi=10.1.1.33.524}.
\bibitem[{Lugosi and Zeger(1996)}]{LugosiZeger1996}
\bibinfo{author}{Lugosi, G.}, \bibinfo{author}{Zeger, K.},
  \bibinfo{year}{1996}.
\newblock \bibinfo{title}{Concept learning using complexity regularization}.
\newblock \bibinfo{journal}{IEEE Transactions on Information Theory}
  \bibinfo{volume}{42}, \bibinfo{pages}{48--54}.
\newblock \DOIprefix\doi{10.1109/18.481777}.
\bibitem[{Manski(1975)}]{Manski1975}
\bibinfo{author}{Manski, C.F.}, \bibinfo{year}{1975}.
\newblock \bibinfo{title}{Maximum score estimation of the stochastic utility
  model of choice}.
\newblock \bibinfo{journal}{Journal of Econometrics} \bibinfo{volume}{3},
  \bibinfo{pages}{205--228}.
\newblock \DOIprefix\doi{http://dx.doi.org/10.1016/0304-4076(75)90032-9}.
\bibitem[{Manski(1985)}]{Manski1985}
\bibinfo{author}{Manski, C.F.}, \bibinfo{year}{1985}.
\newblock \bibinfo{title}{Semiparametric analysis of discrete response}.
\newblock \bibinfo{journal}{Journal of Econometrics} \bibinfo{volume}{27},
  \bibinfo{pages}{313--333}.
\newblock \DOIprefix\doi{http://dx.doi.org/10.1016/0304-4076(85)90009-0}.
\bibitem[{Massart(2000)}]{Massart2000}
\bibinfo{author}{Massart, P.}, \bibinfo{year}{2000}.
\newblock \bibinfo{title}{Some applications of concentration inequalities to
  statistics}.
\newblock \bibinfo{journal}{Annales de la Facult\'{e} des Sciences de Toulouse:
  Math\'{e}matiques} \bibinfo{volume}{9}, \bibinfo{pages}{245--303}.
\newblock \URLprefix \url{http://www.numdam.org/item/AFST_2000_6_9_2_245_0}.
\bibitem[{Massart and N\'{e}d\'{e}lec(2006)}]{MassartNedelec2006}
\bibinfo{author}{Massart, P.}, \bibinfo{author}{N\'{e}d\'{e}lec, E.},
  \bibinfo{year}{2006}.
\newblock \bibinfo{title}{Risk bounds for statistical learning}.
\newblock \bibinfo{journal}{Annals of Statistics} \bibinfo{volume}{34},
  \bibinfo{pages}{2326--2366}.
\newblock \DOIprefix\doi{10.1214/009053606000000786}.
\bibitem[{McDiarmid(1989)}]{McDiarmid1989}
\bibinfo{author}{McDiarmid, C.}, \bibinfo{year}{1989}.
\newblock \bibinfo{title}{On the method of bounded differences}, in:
  \bibinfo{booktitle}{Surveys in Combinatorics 1989}.
  \bibinfo{publisher}{Cambridge University Press}, pp.
  \bibinfo{pages}{148--188}.
\newblock \DOIprefix\doi{https://doi.org/10.1017/CBO9781107359949.008}.
\bibitem[{Richter and Wong(1999)}]{RichterWong1999}
\bibinfo{author}{Richter, M.K.}, \bibinfo{author}{Wong, K.C.},
  \bibinfo{year}{1999}.
\newblock \bibinfo{title}{Non-computability of competitive equilibrium}.
\newblock \bibinfo{journal}{Economic Theory} \bibinfo{volume}{14},
  \bibinfo{pages}{1--27}.
\newblock \DOIprefix\doi{https://doi.org/10.1007/s001990050281}.
\bibitem[{Schwarz(1978)}]{Schwarz1978}
\bibinfo{author}{Schwarz, G.}, \bibinfo{year}{1978}.
\newblock \bibinfo{title}{Estimating the dimension of a model}.
\newblock \bibinfo{journal}{Annals of Statistics} \bibinfo{volume}{6},
  \bibinfo{pages}{461--464}.
\newblock \DOIprefix\doi{10.1214/aos/1176344136}.
\bibitem[{Shalev-Shwartz and Ben-David(2014)}]{Shalev-ShwartzBen-David2014}
\bibinfo{author}{Shalev-Shwartz, S.}, \bibinfo{author}{Ben-David, S.},
  \bibinfo{year}{2014}.
\newblock \bibinfo{title}{Understanding Machine Learning: From Theory to
  Algorithms}.
\newblock \bibinfo{publisher}{Cambridge University Press}.
\newblock \DOIprefix\doi{https://doi.org/10.1017/CBO9781107298019}.
\bibitem[{Shawe-Taylor and Cristianini(2004)}]{Shawe-TaylorCristianini2004}
\bibinfo{author}{Shawe-Taylor, J.}, \bibinfo{author}{Cristianini, N.},
  \bibinfo{year}{2004}.
\newblock \bibinfo{title}{Kernel Methods for Pattern Analysis}.
\newblock \bibinfo{publisher}{Cambridge University Press}.
\newblock \DOIprefix\doi{https://doi.org/10.1017/CBO9780511809682}.
\bibitem[{Sin and White(1996)}]{SinWhite1996}
\bibinfo{author}{Sin, C.Y.}, \bibinfo{author}{White, H.}, \bibinfo{year}{1996}.
\newblock \bibinfo{title}{Information criteria for selecting possibly
  misspecified parametric models}.
\newblock \bibinfo{journal}{Journal of Econometrics} \bibinfo{volume}{71},
  \bibinfo{pages}{207--225}.
\newblock \DOIprefix\doi{https://doi.org/10.1016/0304-4076(94)01701-8}.
\bibitem[{Tibshirani(1996)}]{Tibshirani1996}
\bibinfo{author}{Tibshirani, R.}, \bibinfo{year}{1996}.
\newblock \bibinfo{title}{Regression shrinkage and selection via the lasso}.
\newblock \bibinfo{journal}{Journal of the Royal Statistical Society, Series B
  (Methodological)} \bibinfo{volume}{58}, \bibinfo{pages}{267--288}.
\newblock \DOIprefix\doi{https://doi.org/10.1111/j.2517-6161.1996.tb02080.x}.
\bibitem[{van~der Vaart and Wellner(1996)}]{VaartWellner1996}
\bibinfo{author}{van~der Vaart, A.W.}, \bibinfo{author}{Wellner, J.A.},
  \bibinfo{year}{1996}.
\newblock \bibinfo{title}{Weak Convergence and Empirical Processes}.
\newblock \bibinfo{publisher}{Springer}.
\newblock \DOIprefix\doi{10.1007/978-1-4757-2545-2}.
\bibitem[{Vapnik(1982)}]{Vapnik1982}
\bibinfo{author}{Vapnik, V.}, \bibinfo{year}{1982}.
\newblock \bibinfo{title}{Estimation of Dependences Based on Empirical Data}.
\newblock \bibinfo{publisher}{Springer}.
\bibitem[{Vapnik and Chervonenkis(1971)}]{VapnikChervonenkis1971}
\bibinfo{author}{Vapnik, V.}, \bibinfo{author}{Chervonenkis, A.},
  \bibinfo{year}{1971}.
\newblock \bibinfo{title}{On the uniform convergence of relative frequencies of
  events to their probabilities}.
\newblock \bibinfo{journal}{Theory of Probability \& Its Applications}
  \bibinfo{volume}{16}, \bibinfo{pages}{264--280}.
\newblock \DOIprefix\doi{10.1137/1116025}.
\bibitem[{Vapnik(2000)}]{Vapnik2000}
\bibinfo{author}{Vapnik, V.N.}, \bibinfo{year}{2000}.
\newblock \bibinfo{title}{The Nature of Statistical Learning Theory}.
\newblock \bibinfo{edition}{2nd} ed., \bibinfo{publisher}{Springer}.
\newblock \DOIprefix\doi{10.1007/978-1-4757-2440-0}.
\bibitem[{Zhu et~al.(2004)Zhu, Rosset, Hastie and
  Tibshirani}]{ZhuRossetEtAl2004}
\bibinfo{author}{Zhu, J.}, \bibinfo{author}{Rosset, S.},
  \bibinfo{author}{Hastie, T.}, \bibinfo{author}{Tibshirani, R.},
  \bibinfo{year}{2004}.
\newblock \bibinfo{title}{1-norm support vector machines}, in:
  \bibinfo{editor}{Thrun, S.}, \bibinfo{editor}{Saul, L.K.},
  \bibinfo{editor}{Sch\"{o}lkopf, B.} (Eds.), \bibinfo{booktitle}{Advances in
  Neural Information Processing Systems}. volume~\bibinfo{volume}{16}, pp.
  \bibinfo{pages}{49--56}.
\newblock \URLprefix \url{https://dl.acm.org/doi/10.5555/2981345.2981352}.

\end{thebibliography}
\bibliographystyle{elsarticle-harv}

\begin{table}[b]
\thisfloatpagestyle{empty}
\centering
\caption{Relative Generalized Expected Utility of ML, MU, and UMPR}
\label{Table1}
\scalebox{0.7}{
\begin{threeparttable}

\begin{tabular}{llllllllll}
\toprule
\multicolumn{2}{c}{n=500} & \multicolumn{8}{c}{} \\[0.3cm]
\underline{DGP1} & \multicolumn{9}{c}{$p^{*}(x)=\Lambda(-0.5x+0.2x^{3})\hspace{2.5cm}$} \\[0.3cm]
Preference & \multicolumn{4}{c}{$b(x)=20$ and $c(x)=0.5$} & & \multicolumn{4}{c}{$b(x)=20$ and $c(x)=0.5+0.025x$} \\
        \cmidrule(r){2-5}\cmidrule(r){7-10}
 & $k=1$ & $k=2$ & $k=3$ &  &  & $k=1$ & $k=2$ & $k=3$ & \\
ML  & 34.69 & 31.72 & 93.93 &  &  & 8.50 & 11.52 & 94.70 & \\
MU        & 51.05 & 55.33 & 67.15 &  &  & 33.44 & 45.56 & 58.40 & \\[0.3cm]
         & $\alpha=1$   & $\alpha=0.5$ & $\alpha=0.1$  & $\alpha=0.05$ &              & $\alpha=1$   & $\alpha=0.5$  & $\alpha=0.1$  & $\alpha=0.05$  \\
UMPR (VC) & 54.60 & 55.07 & 55.35 & 55.35 &  & 36.88 & 37.01 & 37.57 & 37.70 \\
UMPR (MD) & 58.93 & 59.77 & 60.71 & 60.83 &  & 47.64 & 49.59 & 51.07 & 51.10 \\[0.3cm]
        & \multicolumn{2}{l}{Cross-Validated $\hat{\alpha}$} & \multicolumn{2}{l}{No Technical Term} & & \multicolumn{2}{l}{Cross-Validated $\hat{\alpha}$} & \multicolumn{2}{l}{No Technical Term}\\
UMPR (VC) & \multicolumn{2}{c}{54.72} & \multicolumn{2}{c}{56.51} &  & \multicolumn{2}{c}{37.12} & \multicolumn{2}{c}{\hspace{-2.9cm} 41.32} \\
UMPR (MD) & \multicolumn{2}{c}{59.59} & \multicolumn{2}{c}{64.65} &  & \multicolumn{2}{c}{49.27} & \multicolumn{2}{c}{\hspace{-2.9cm} 58.43} \\[0.5cm]
\addlinespace[7pt]
\underline{DGP2} & \multicolumn{9}{c}{$p^{*}(x_{1},x_{2})=\Lambda(Q(1.5x_{1}+1.5x_{2}))$ where $Q(u)=\frac{1.5-0.1u}{\exp\{0.25u+0.1u^{2}-0.04u^{3}\}}$} \\[0.3cm]
 Preference & \multicolumn{4}{c}{$b(x_{1},x_{2})=20$ and $c(x_{1},x_{2})=0.75$} & & \multicolumn{4}{c}{$b(x_{1},x_{2})=20+40\cdot\Ind{[|x_{1}+x_{2}|<1.5]}$ and $c(x_{1},x_{2})=0.75$} \\
        \cmidrule(r){2-5}\cmidrule(r){7-10}
  & $k=1$ & $k=2$ & $k=3$ &  &  & $k=1$ & $k=2$ & $k=3$ & \\
ML  & 60.26 & 59.41 & 60.09 &  &  & 30.86 & 29.19 & 34.60 & \\
MU        & 67.44 & 51.78 & 68.14 &  &  & 49.33 & 33.10 & 51.87 & \\[0.3cm]
       & $\alpha=1$   & $\alpha=0.5$ & $\alpha=0.1$  & $\alpha=0.05$ &              & $\alpha=1$   & $\alpha=0.5$  & $\alpha=0.1$  & $\alpha=0.05$  \\
UMPR (VC) & 67.44 & 67.44 & 67.44 & 67.44 &  & 49.33 & 49.33 & 49.33 & 49.33 \\
UMPR (MD) & 68.28 & 68.25 & 68.36 & 68.34 &  & 49.30 & 49.29 & 49.48 & 49.47\\[0.3cm]
 & \multicolumn{2}{l}{Cross-Validated $\hat{\alpha}$} & \multicolumn{2}{l}{No Technical Term} & & \multicolumn{2}{l}{Cross-Validated $\hat{\alpha}$} & \multicolumn{2}{l}{No Technical Term}\\
UMPR (VC) & \multicolumn{2}{c}{67.44} & \multicolumn{2}{c}{67.44} & & \multicolumn{2}{c}{49.33} & \multicolumn{2}{c}{\hspace{-2.9cm} 49.33} \\
UMPR (MD) & \multicolumn{2}{c}{68.22} & \multicolumn{2}{c}{68.13} & & \multicolumn{2}{c}{49.38} & \multicolumn{2}{c}{\hspace{-2.9cm} 50.07}\\[0.3cm]
\multicolumn{2}{c}{n=1000} & \multicolumn{8}{c}{} \\[0.3cm]
\underline{DGP1} & \multicolumn{9}{c}{$p^{*}(x)=\Lambda(-0.5x+0.2x^{3})\hspace{2.5cm}$} \\[0.3cm]
Preference & \multicolumn{4}{c}{$b(x)=20$ and $c(x)=0.5$} & & \multicolumn{4}{c}{$b(x)=20$ and $c(x)=0.5+0.025x$} \\
        \cmidrule(r){2-5}\cmidrule(r){7-10}
  & $k=1$ & $k=2$ & $k=3$ &  &  & $k=1$ & $k=2$ & $k=3$ & \\
ML  & 32.24 & 31.32 & 97.21  &  &  & 6.83  & 7.09  & 97.48 & \\
MU  & 54.88 & 58.48 & 69.94 &  &  & 35.17 & 48.03 & 60.04 & \\[0.3cm]
       & $\alpha=1$   & $\alpha=0.5$ & $\alpha=0.1$  & $\alpha=0.05$ &              & $\alpha=1$   & $\alpha=0.5$  & $\alpha=0.1$  & $\alpha=0.05$  \\
UMPR (VC) & 60.26 & 60.73 & 60.84 & 60.84 &  & 44.26 & 44.74 & 45.34 & 45.55 \\
UMPR (MD) & 64.07 & 64.74 & 65.19 & 65.47 &  & 55.20 & 56.11 & 57.68 & 57.78\\[0.3cm]
 & \multicolumn{2}{l}{Cross-Validated $\hat{\alpha}$} & \multicolumn{2}{l}{No Technical Term} & & \multicolumn{2}{l}{Cross-Validated $\hat{\alpha}$} & \multicolumn{2}{l}{No Technical Term}\\
UMPR (VC) & \multicolumn{2}{c}{60.40} & \multicolumn{2}{c}{62.26} & & \multicolumn{2}{c}{44.30} & \multicolumn{2}{c}{\hspace{-2.9cm} 49.30} \\
UMPR (MD) & \multicolumn{2}{c}{64.51} & \multicolumn{2}{c}{67.94} & & \multicolumn{2}{c}{56.13} & \multicolumn{2}{c}{\hspace{-2.9cm} 63.11}\\[0.5cm]
\addlinespace[7pt]
\underline{DGP2} & \multicolumn{9}{c}{$p^{*}(x_{1},x_{2})=\Lambda(Q(1.5x_{1}+1.5x_{2}))$ where $Q(u)=\frac{1.5-0.1u}{\exp\{0.25u+0.1u^{2}-0.04u^{3}\}}$} \\[0.3cm]
 Preference & \multicolumn{4}{c}{$b(x_{1},x_{2})=20$ and $c(x_{1},x_{2})=0.75$} & & \multicolumn{4}{c}{$b(x_{1},x_{2})=20+40\cdot\Ind{[|x_{1}+x_{2}|<1.5]}$ and $c(x_{1},x_{2})=0.75$} \\
        \cmidrule(r){2-5}\cmidrule(r){7-10}
  & $k=1$ & $k=2$ & $k=3$ &  &  & $k=1$ & $k=2$ & $k=3$ & \\
ML  & 59.03 & 57.64 & 59.74  &  &  & 28.13 & 25.07 & 31.77 & \\
MU  & 69.87 & 55.43 & 71.19  &  &  & 52.89 & 39.38 & 56.19 & \\[0.3cm]
       & $\alpha=1$   & $\alpha=0.5$ & $\alpha=0.1$  & $\alpha=0.05$ &              & $\alpha=1$   & $\alpha=0.5$  & $\alpha=0.1$  & $\alpha=0.05$  \\
UMPR (VC) & 69.87 & 69.87 & 69.87 & 69.87 &  & 52.89 & 52.89 & 52.89 & 52.89 \\
UMPR (MD) & 70.46 & 70.50 & 70.54 & 70.55 &  & 53.02 & 53.15 & 54.03 & 54.07\\[0.3cm]
 & \multicolumn{2}{l}{Cross-Validated $\hat{\alpha}$} & \multicolumn{2}{l}{No Technical Term} & & \multicolumn{2}{l}{Cross-Validated $\hat{\alpha}$} & \multicolumn{2}{l}{No Technical Term}\\
UMPR (VC) & \multicolumn{2}{c}{69.87} & \multicolumn{2}{c}{69.87} & & \multicolumn{2}{c}{52.89} & \multicolumn{2}{c}{\hspace{-2.9cm} 52.89} \\
UMPR (MD) & \multicolumn{2}{c}{70.48} & \multicolumn{2}{c}{70.84} & & \multicolumn{2}{c}{53.31} & \multicolumn{2}{c}{\hspace{-2.9cm} 56.55}\\[0.3cm]
\bottomrule
\end{tabular}
\smallskip
\begin{tablenotes}[flushleft]
\item[] Note: (i) Every relative generalized expected utility is expressed as a percentage. (ii) The \emph{glmfit} and \emph{simulannealbnd} algorithms in MATLAB\textsuperscript{\textregistered} are used to compute ML and MU, respectively. The tuning parameters in both algorithms are all set by default.
\end{tablenotes}
\end{threeparttable}}
\end{table}

\begin{table}[b]
\thisfloatpagestyle{empty}
\centering
\caption{Relative Generalized Expected Utility of UMPR, AIC, BIC, LASSO and SVM}
\label{Table2}

\scalebox{0.8}{
\begin{threeparttable}

\begin{tabular}{llllllllll}
\toprule
\multicolumn{2}{c}{n=500} & \multicolumn{8}{c}{} \\[0.3cm]
\underline{DGP1}  & \multicolumn{9}{c}{$p^{*}(x)=\Lambda(-0.5x+0.2x^{3})\hspace{2.5cm}$}       \\[0.3cm]
      Preference & \multicolumn{4}{c}{$b(x)=20$ and $c(x)=0.5$} & &\multicolumn{4}{c}{$b(x)=20$ and $c(x)=0.5+0.025x$} \\
      \cmidrule(r){2-5}\cmidrule(r){7-10}
      UMPR               & MD    & SMD   & RC    & BC    & & MD    & SMD   & RC    & BC    \\
                         & 65.36 & 66.68 & 66.86 & 65.74 & & 55.00 & 58.87 & 58.58 & 57.65 \\[0.3cm]
      Information        & AIC   & BIC   &       &       & & AIC   & BIC   &       &       \\
      Criterion          & 93.93 & 89.95 &       &       & & 94.70 & 88.81 &       &       \\[0.3cm]
      $\ell_{1}$-Penalty & LASSO & SVM   &       &       & & LASSO & SVM   &       &       \\
                         & 60.60 & 87.77 &       &       & & 65.20 & 83.91 &       &       \\[0.5cm]
      \underline{DGP2}  & \multicolumn{9}{c}{$p^{*}(x_{1},x_{2})=\Lambda(Q(1.5x_{1}+1.5x_{2}))$ where $Q(u)=\frac{1.5-0.1u}{\exp\{0.25u+0.1u^{2}-0.04u^{3}\}}$}  \\[0.3cm]
      Preference & \multicolumn{4}{c}{$b(x_{1},x_{2})=20$ and $c(x_{1},x_{2})=0.75$} &  & \multicolumn{4}{c}{$b(x_{1},x_{2})=20+40\cdot\Ind{[|x_{1}+x_{2}|<1.5]}$ and $c(x_{1},x_{2})=0.75$} \\
      \cmidrule(r){2-5}\cmidrule(r){7-10}
      UMPR               & MD    & SMD   & RC    & BC    &  & MD    & SMD   & RC    & BC    \\
                         & 68.55 & 69.52 & 69.47 & 69.11 &  & 50.41 & 53.87 & 53.32 & 52.90 \\[0.3cm]
      Information        & AIC   & BIC   &       &       &  & AIC   & BIC   &       &       \\
      Criterion          & 60.07 & 60.27 &       &       &  & 33.20 & 30.90 &       &       \\[0.3cm]
      $\ell_{1}$-Penalty & LASSO & SVM   &       &       &  & LASSO & SVM   &       &       \\
                         & 59.75 & 26.86 &       &       &  & 32.93 & 5.92  &       &       \\[0.6cm]
\multicolumn{2}{c}{n=1000} & \multicolumn{8}{c}{} \\[0.3cm]
\underline{DGP1}  & \multicolumn{9}{c}{$p^{*}(x)=\Lambda(-0.5x+0.2x^{3})\hspace{2.5cm}$}       \\[0.3cm]
      Preference & \multicolumn{4}{c}{$b(x)=20$ and $c(x)=0.5$} & &\multicolumn{4}{c}{$b(x)=20$ and $c(x)=0.5+0.025x$} \\
      \cmidrule(r){2-5}\cmidrule(r){7-10}
      UMPR               & MD    & SMD   & RC    & BC    &  & MD    & SMD   & RC    & BC    \\
                         & 69.32 & 72.51 & 72.23 & 71.75 &  & 63.30 & 67.12 & 67.01 & 65.81 \\[0.3cm]
      Information        & AIC   & BIC   &       &       &  & AIC   & BIC   &       &       \\
      Criterion          & 97.21 & 97.13 &       &       &  & 97.48 & 97.29 &       &       \\[0.3cm]
      $\ell_{1}$-Penalty & LASSO & SVM   &       &       &  & LASSO & SVM   &       &       \\
                         & 68.82 & 93.26 &       &       &  & 78.92 & 91.14 &       &       \\[0.5cm]
      \underline{DGP2}  & \multicolumn{9}{c}{$p^{*}(x_{1},x_{2})=\Lambda(Q(1.5x_{1}+1.5x_{2}))$ where $Q(u)=\frac{1.5-0.1u}{\exp\{0.25u+0.1u^{2}-0.04u^{3}\}}$}  \\[0.3cm]
      Preference & \multicolumn{4}{c}{$b(x_{1},x_{2})=20$ and $c(x_{1},x_{2})=0.75$} &  & \multicolumn{4}{c}{$b(x_{1},x_{2})=20+40\cdot\Ind{[|x_{1}+x_{2}|<1.5]}$ and $c(x_{1},x_{2})=0.75$} \\
      \cmidrule(r){2-5}\cmidrule(r){7-10}
      UMPR               & MD    & SMD   & RC    & BC    &  & MD    & SMD   & RC    & BC    \\
                         & 71.09 & 71.91 & 71.97 & 71.89 &  & 57.13 & 59.61 & 60.08 & 58.96 \\[0.3cm]
      Information        & AIC   & BIC   &       &       &  & AIC   & BIC   &       &       \\
      Criterion          & 59.72 & 59.06 &       &       &  & 31.49 & 28.16 &       &       \\[0.3cm]
      $\ell_{1}$-Penalty & LASSO & SVM   &       &       &  & LASSO & SVM   &       &       \\
                         & 59.68 & 25.93 &       &       &  & 29.08 & 5.10  &       &       \\[0.3cm]
\bottomrule
\end{tabular}
\smallskip
\begin{tablenotes}[flushleft]
\item [] Note: (i) Every relative generalized expected utility is expressed as a percentage. (ii) The number of simulation replications for complexity penalty (SMD, RC, and BC) is $m=10$. (iii) The \emph{simulannealbnd} algorithm in MATLAB\textsuperscript{\textregistered} is used to compute UMPR; the \emph{glmfit} algorithm in MATLAB\textsuperscript{\textregistered} is used to compute AIC and BIC; the \emph{lassoglm} algorithm in MATLAB\textsuperscript{\textregistered} is used to compute LASSO; and the \emph{lpsvm} algorithm provided by \citet{FungMangasarian2004}, which is available at \url{http://research.cs.wisc.edu/dmi/svm/lpsvm/}, is used to compute SVM. The $\ell_{1}$ regularization parameter in LASSO and SVM is determined by tenfold cross-validation, whereas the other tuning parameters in the algorithms are all set by default.
\end{tablenotes}
\end{threeparttable}}
\end{table}

\begingroup
\afterpage{
\thispagestyle{empty}
\renewcommand\arraystretch{0.2}
\begin{landscape}
\begin{longtable}{llllllllll}
\caption{Percentage of Models Selected by UMPR, Pretest, and Cross-Validatory Estimators}
\label{Table3}\\
\toprule
\endfirsthead
\caption* {Table \ref{Table3}: Percentage of Models Selected by UMPR, Pretest, and Cross-Validatory Estimators \textit{(Continued)}}\\
\toprule
\endhead
\midrule
\multicolumn{10}{r}{\textit{Continued on next page}}\\
\endfoot
\multicolumn{10}{p{20cm}}{\footnotesize{Note: (i) Every relative generalized expected utility is expressed as a percentage. (ii) The specific-to-general pretest estimator is referred to as Pretest(S$\to$G), whereas the general-to-specific pretest estimator is referred to as Pretest(G$\to$S). VC($\hat{\alpha}$) indicates a VC penalty inclusive of a technical term with cross-validated $\hat{\alpha}$.
(iii) The number of simulation replications for complexity penalty (SMD, RC, and BC) is $m=10$. (iv) The \emph{simulannealbnd} algorithm in MATLAB\textsuperscript{\textregistered} is used to compute UMPR, pretest, and cross-validatory estimators. The tuning parameters in the algorithm are all set by default.}}
\endlastfoot

\multicolumn{2}{c}{n=500} & \multicolumn{8}{c}{} \\[0.3cm]
\underline{DGP1} & \multicolumn{9}{c}{$p^{*}(x)=\Lambda(-0.5x+0.2x^{3})\hspace{2.5cm}$} \\[0.3cm]
Preference       & \multicolumn{4}{c}{$b(x)=20$ and $c(x)=0.5$} &  & \multicolumn{4}{c}{$b(x)=20$ and $c(x)=0.5+0.025x$} \\
      \cmidrule(r){2-5}\cmidrule(r){7-10}
      & \multicolumn{3}{c}{Hierarchy} & RGEU &  & \multicolumn{3}{c}{Hierarchy} & RGEU \\[0.1cm]
      & $k=1$ & $k=2$ & $k=3$ &  &  & $k=1$ & $k=2$ & $k=3$ &  \\[0.1cm]
Pretest (S$\to$G)& 69.2 & 19.8 & 11.0 & 59.27 &  & 56.4 & 27.4 & 16.2 & 45.63\\[0.1cm]
Pretest (G$\to$S)& 36.4 & 19.8 & 43.8 & 62.69 &  & 33.6 & 27.4 & 39.0 & 48.69\\[0.1cm]
Cross-Validation & 25.0 & 29.4 & 45.6 & 61.30 &  & 17.0 & 26.8 & 56.2 & 50.42\\[0.1cm]
UMPR (MD)        & 33.4 & 27.2 & 39.4 & 65.36 &  & 27.0 & 32.6 & 40.4 & 55.00\\[0.1cm]
UMPR (SMD)       & 38.2 & 18.4 & 43.4 & 66.68 &  & 30.0 & 27.0 & 43.0 & 58.87\\[0.1cm]
UMPR (RC)        & 40.6 & 17.8 & 41.6 & 66.86 &  & 28.6 & 30.2 & 41.2 & 58.58\\[0.1cm]
UMPR (BC)        & 44.8 & 19.0 & 36.2 & 65.74 &  & 33.8 & 27.2 & 39.0 & 57.65\\[0.1cm]
UMPR (VC)        & 88.4 & 7.6 & 4.0 & 56.06 &  & 85.2 & 11.8 & 3.0 & 39.95\\[0.1cm]
UMPR (VC ($\hat{\alpha}$))        & 93.0 & 5.0 & 2.0 & 54.72 &  & 91.4 & 7.8 & 0.8 & 37.12\\[0.1cm]
\addlinespace[10pt]
\underline{DGP2} & \multicolumn{9}{c}{$p^{*}(x_{1},x_{2})=\Lambda(Q(1.5x_{1}+1.5x_{2}))$ where $Q(u)=\frac{1.5-0.1u}{\exp\{0.25u+0.1u^{2}-0.04u^{3}\}}$} \\[0.3cm]
Preference       & \multicolumn{4}{c}{$b(x_{1},x_{2})=20$ and $c(x_{1},x_{2})=0.75$} &  & \multicolumn{4}{c}{$b(x_{1},x_{2})=20+40\cdot\Ind{[|x_{1}+x_{2}|<1.5]}$ and $c(x_{1},x_{2})=0.75$} \\
      \cmidrule(r){2-5}\cmidrule(r){7-10}
      & \multicolumn{3}{c}{Hierarchy} & RGEU &  & \multicolumn{3}{c}{Hierarchy} & RGEU \\[0.1cm]
      & $k=1$ & $k=2$ & $k=3$ &  &  & $k=1$ & $k=2$ & $k=3$ &  \\[0.1cm]
Pretest (S$\to$G)& 85.2 & 6.4 & 8.4 & 68.72 &  & 78.8 & 9.4 & 11.8 & 50.62\\[0.1cm]
Pretest (G$\to$S)& 33.4 & 6.4 & 60.2 & 68.34 &  & 27.6 & 9.4 & 63.0 & 49.91\\[0.1cm]
Cross-Validation & 48.0 & 5.8 & 46.2 & 67.30 &  & 43.2 & 15.6 & 41.2 & 48.26\\[0.1cm]
UMPR (MD)        & 52.4 & 14.2 & 33.4 & 68.55 &  & 52.2 & 19.6 & 28.2 & 50.41\\[0.1cm]
UMPR (SMD)       & 66.0 & 9.4 & 24.6 & 69.52 &  & 58.6 & 12.8 & 28.6 & 53.87\\[0.1cm]
UMPR (RC)        & 68.2 & 8.4 & 23.4 & 69.47 &  & 61.4 & 11.8 & 26.8 & 53.32\\[0.1cm]
UMPR (BC)        & 74.8 & 8.4 & 16.8 & 69.11 &  & 66.2 & 12.8 & 21.0 & 52.90\\[0.1cm]
UMPR (VC)        & 99.8 & 0.2 & 0.0 & 67.54 &  & 100.0 & 0.0 & 0.0 & 49.33\\[0.1cm]
UMPR (VC ($\hat{\alpha}$))        & 100.0 & 0.0 & 0.0 & 67.44 &  & 100.0 & 0.0 & 0.0 & 49.33\\[0.1cm]
\pagebreak
\multicolumn{2}{c}{n=1000} & \multicolumn{8}{c}{} \\[0.3cm]
\underline{DGP1} & \multicolumn{9}{c}{$p^{*}(x)=\Lambda(-0.5x+0.2x^{3})\hspace{2.5cm}$} \\[0.3cm]
Preference       & \multicolumn{4}{c}{$b(x)=20$ and $c(x)=0.5$} &  & \multicolumn{4}{c}{$b(x)=20$ and $c(x)=0.5+0.025x$} \\
      \cmidrule(r){2-5}\cmidrule(r){7-10}
      & \multicolumn{3}{c}{Hierarchy} & RGEU &  & \multicolumn{3}{c}{Hierarchy} & RGEU \\[0.1cm]
      & $k=1$ & $k=2$ & $k=3$ &  &  & $k=1$ & $k=2$ & $k=3$ &  \\[0.1cm]
Pretest (S$\to$G)& 68.0 & 17.4 & 14.6 & 62.60 &  & 50.2 & 25.8 & 24.0 & 50.14\\[0.1cm]
Pretest (G$\to$S)& 35.2 & 17.4 & 47.4 & 65.20 &  & 27.4 & 25.8 & 46.8 & 53.52\\[0.1cm]
Cross-Validation & 17.8 & 24.4 & 57.8 & 64.81 &  & 11.8 & 23.6 & 64.6 & 55.19\\[0.1cm]
UMPR (MD)        & 30.8 & 23.6 & 45.6 & 69.32 &  & 21.2 & 30.8 & 48.0 & 63.30\\[0.1cm]
UMPR (SMD)       & 32.6 & 17.6 & 49.8 & 72.51 &  & 20.0 & 29.6 & 50.4 & 67.12\\[0.1cm]
UMPR (RC)        & 34.2 & 16.6 & 49.2 & 72.23 &  & 19.6 & 30.8 & 49.6 & 67.01\\[0.1cm]
UMPR (BC)        & 37.4 & 18.6 & 44.0 & 71.75 &  & 25.0 & 28.6 & 46.4 & 65.81\\[0.1cm]
UMPR (VC)        & 81.0 & 11.0 & 0.8 & 62.40 &  & 70.6 & 21.4 & 8.0 & 49.30\\[0.1cm]
UMPR (VC ($\hat{\alpha}$))        & 86.2 & 9.8 & 4.0 & 60.40 &  & 82.0 & 14.2 & 3.8 & 44.30\\[0.1cm]
\addlinespace[10pt]
\underline{DGP2} & \multicolumn{9}{c}{$p^{*}(x_{1},x_{2})=\Lambda(Q(1.5x_{1}+1.5x_{2}))$ where $Q(u)=\frac{1.5-0.1u}{\exp\{0.25u+0.1u^{2}-0.04u^{3}\}}$} \\[0.3cm]
Preference       & \multicolumn{4}{c}{$b(x_{1},x_{2})=20$ and $c(x_{1},x_{2})=0.75$} &  & \multicolumn{4}{c}{$b(x_{1},x_{2})=20+40\cdot\Ind{[|x_{1}+x_{2}|<1.5]}$ and $c(x_{1},x_{2})=0.75$} \\
      \cmidrule(r){2-5}\cmidrule(r){7-10}
      & \multicolumn{3}{c}{Hierarchy} & RGEU &  & \multicolumn{3}{c}{Hierarchy} & RGEU \\[0.1cm]
      & $k=1$ & $k=2$ & $k=3$ &  &  & $k=1$ & $k=2$ & $k=3$ &  \\[0.1cm]
Pretest (S$\to$G)& 87.8 & 6.6 & 5.6 & 70.90 &  & 77.0 & 6.8 & 16.2 & 56.64\\[0.1cm]
Pretest (G$\to$S)& 34.2 & 6.6 & 59.2 & 71.20 &  & 29.4 & 6.8 & 63.8 & 56.48\\[0.1cm]
Cross-Validation & 42.0 & 4.4 & 53.6 & 69.93 &  & 43.0 & 9.6 & 47.4 & 54.51\\[0.1cm]
UMPR (MD)        & 54.0 & 14.8 & 31.2 & 71.09 &  & 48.2 & 18.2 & 33.6 & 57.13\\[0.1cm]
UMPR (SMD)       & 65.6 & 8.6 & 25.8 & 71.91 &  & 55.0 & 10.8 & 34.2 & 59.61\\[0.1cm]
UMPR (RC)        & 68.4 & 9.0 & 22.6 & 71.97 &  & 55.8 & 8.8 & 35.4 & 60.08\\[0.1cm]
UMPR (BC)        & 69.8 & 10.0 & 20.2 & 71.89 &  & 60.2 & 9.8 & 30.0 & 58.96\\[0.1cm]
UMPR (VC)        & 99.8 & 0.2 & 0.0 & 69.99 &  & 100.0 & 0.0 & 0.0 & 52.89\\[0.1cm]
UMPR (VC ($\hat{\alpha}$))        & 100.0 & 0.0 & 0.0 & 69.87 &  & 100.0 & 0.0 & 0.0 & 52.89\\[0.1cm]
\bottomrule\\
\end{longtable}
\thispagestyle{empty}
\end{landscape}}
\endgroup

\end{document}